\documentclass[
notitlepage,
floatfix,
aps,pra,
reprint,
twocolumn,
superscriptaddress
]{revtex4-2}
\usepackage[lmargin=.7in,rmargin=.7in,tmargin=.6in,bmargin=0.8in]{geometry}
\usepackage{amsmath,amssymb,amsfonts,enumitem}
\usepackage{amsthm}
\usepackage{graphicx}
\usepackage{textcomp}
\usepackage[dvipsnames]{xcolor}
\usepackage{braket}
\usepackage{mathtools}
\usepackage{bm}
\usepackage{times}
\usepackage[ruled,vlined]{algorithm2e}
\usepackage{algcompatible}
\renewcommand{\Re}{\operatorname{Re}}
\renewcommand{\Im}{\operatorname{Im}}
\usepackage[colorlinks]{hyperref}

\usepackage{amsmath,amssymb,amsfonts}
\newtheorem{theorem}{Theorem}
\newtheorem{corollary}[theorem]{Corollary}
\newtheorem{lemma}[theorem]{Lemma}

\newtheorem{prop}[theorem]{Proposition}

\usepackage{braket}

\usepackage[normalem]{ulem}
\newcommand\erase{\bgroup\markoverwith{\textcolor{red}{\rule[.5ex]{2pt}{0.4pt}}}\ULon}
\newcommand\eraseblue{\bgroup\markoverwith{\textcolor{blue}{\rule[.5ex]{2pt}{0.8pt}}}\ULon}

\usepackage{tikz}
\usepackage{float}
\hypersetup{colorlinks,linkcolor={blue},citecolor={blue},urlcolor={blue}}

\definecolor{tableblue}{HTML}{ccffff}

\usepackage{minitoc}
\usepackage{tocloft}

\begin{document}

\doparttoc 
\faketableofcontents 
\part{}


\title{
Grover's algorithm is an approximation of imaginary-time evolution \\
}

\newcommand{\EPFLQSE}{Centre for Quantum Science and Engineering, \'{E}cole Polytechnique F\'{e}d\'{e}rale de Lausanne (EPFL), Lausanne, Switzerland}
\newcommand{\EPFL}{Institute of Physics, \'{E}cole Polytechnique F\'{e}d\'{e}rale de
Lausanne (EPFL), Lausanne, Switzerland}
\newcommand{\KEIO}{Quantum Computing Center, Keio University, Hiyoshi 3-14-1, Kohoku-ku, Yokohama 223-8522, Japan}
\author{Yudai Suzuki}
\affiliation{\EPFL}\affiliation{\EPFLQSE}\affiliation{\KEIO}
\author{Marek Gluza}
\affiliation{School of Physical and Mathematical Sciences, Nanyang Technological University, 637371 Singapore}
\author{Jeongrak Son}
\affiliation{School of Physical and Mathematical Sciences, Nanyang Technological University, 637371 Singapore}
\author{Bi Hong Tiang}
\affiliation{School of Physical and Mathematical Sciences, Nanyang Technological University, 637371 Singapore}
\author{Nelly H. Y. Ng }
\affiliation{School of Physical and Mathematical Sciences, Nanyang Technological University, 637371 Singapore}
\affiliation{Centre for Quantum Technologies, Nanyang Technological University, 50 Nanyang Avenue, 639798 Singapore}
\author{Zo\"{e} Holmes}
\affiliation{\EPFL}\affiliation{\EPFLQSE}

\date{\today}

\begin{abstract}
We reveal the power of Grover's algorithm from thermodynamic and geometric perspectives by showing that it is a product formula approximation of imaginary-time evolution (ITE), a Riemannian gradient flow on the special unitary group.
This ITE formulation provides a unified perspective on Grover’s algorithm, its variants and extensions to widely used quantum subroutines including amplitude amplification and oblivious amplitude amplification.
Specifically, the framework explains the choice of angles in the original Grover's algorithm and $\pi/3$-algorithm. It also motivates a new $\pi/2$-algorithm, for cases a modest failure probability is acceptable, that converges faster than the  $\pi/3$-algorithm without overshooting.
Our analysis further provides a link between ITE and quantum signal processing, which yields a new implementation of the fixed-point quantum search algorithm.
Moreover, the ITE formulation can systematically reproduce widely-used subroutines in modern quantum algorithms, such as (oblivious) amplitude amplification.
These results collectively establish a deeper understanding of Grover’s algorithm and suggest a potential role for thermodynamics and geometry in quantum algorithm design.

\end{abstract}
%
\maketitle
%
%

\begin{figure}[t]
\centering
\begin{tikzpicture}
\node[anchor=center] (russell) at (0.,0.0){\centering\includegraphics[width=0.42\textwidth]{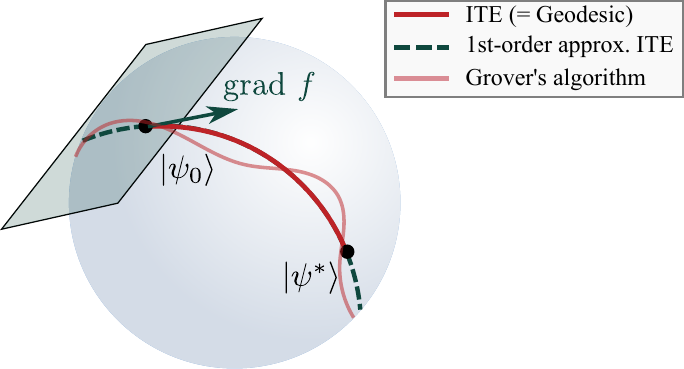}};
\end{tikzpicture}
\caption{\textbf{Geometrical picture of the unstructured search by Grover's algorithm.}
We demonstrate that Grover’s algorithm can be viewed as a product formula approximation of Imaginary-Time Evolution (ITE), which corresponds to the steepest descent direction, $\text{grad }f$, of a least-squares cost $f$ on the special unitary group.
Moreover, we find that ITE and its first-order approximation trace the shortest path (i.e., a geodesic) between the initial state $\ket{\psi_{0}}$ and the solution state $\ket{\psi^{*}}$.
Accordingly, Grover’s algorithm can be understood as a product formula approximation of the geodesic.
}
\label{fig:summary}
\end{figure}

\section{Introduction}
Grover's algorithm provides a quadratic speed-up for the unstructured search problem~\cite{grover1996fast}, where the goal is to identify marked items within a dataset based on a binary function $f: \mathcal X \rightarrow \{0,1\}$.
Given a search space $\mathcal{X}\in\{0,1,\ldots,N-1\}$, the aim is to find $M$ target items satisfying $f(x)=1$.
Classically, solving this problem requires $\mathcal{O}(N)$ queries, while Grover's algorithm performs the task with only $\mathcal{O}(\sqrt{N})$ queries~\cite{nielsen2010quantum}.
Thus far, significant efforts have been devoted to elucidating the characteristics of Grover's algorithm, including its optimal complexity~\cite{bennett1997strengths,beals1998tight,zalka1999grover,nielsen2010quantum} and remedies for the issue of ``overshooting"~\cite{brassard1997searching,grover2005fixed,yoder2014fixed}.
Grover's framework has been further extended to the quantum singular value transformation~\cite{gilyen2019quantum,martyn2021grand}.
These advances underscore Grover’s fundamental role in quantum algorithm design.

On the other hand, it is widely observed that, despite decades of research, there are remarkably few quantum algorithmic primitives~\cite{shor2003haven, zimboras2025myths,dalzell2023quantum}. The search for new primitives is held back, in part, by the difficulty of intuitive unifying understandings of why quantum algorithms work~\cite{shor2003haven}.
Uncovering the key factors underlying the success of Grover's algorithm, one of the field's most fundamental primitives, would offer valuable guidance for the principled design of novel, optimal quantum algorithms.

Here, we analyze Grover’s algorithm through the lens of Imaginary-Time Evolution (ITE), a thermodynamically-inspired approach to prepare the ground state of a target Hamiltonian. 
Specifically, we show that the Grover iteration is a product formula approximation of ITE. These ITE dynamics are then equivalent to the steepest descent direction, i.e., the Riemannian gradient flow on the special unitary group~\cite{gluza_DB_QITE_2024,mcmahon2025equating}. Thus we can interpret the success of Grover's algorithm as a product formula approximation of ITE, that in turn follows the steepest descent direction on the unitary manifold (see Fig.~\ref{fig:summary}).

\begin{figure*}[t]
\centering
\begin{tikzpicture}
\definecolor{lightgray}{HTML}{F4F4F4}
\definecolor{littlelightgray}{HTML}{ecececff}
\definecolor{pale}{HTML}{7ca3d4ff}
\definecolor{lightred}{HTML}{d8a2a2}


\node[draw=white, anchor=west] at (-14.25,-0.5) {\textbf{ITE formulation}};

\node[draw=black] at (-12.5,-1.3) {Unstructured Search};
\node[draw=black] at (-8.0,-1.28) {Grover's Algorithm};
\node[draw=black] at (-10.5,-3.4) {Imaginary-Time Evolution};

\draw[<->, ultra thick] (-12.5,-1.8) -- (-11.4,-2.8);
\draw[<->, ultra thick] (-8.8,-1.8) -- (-9.9,-2.8);

\node[draw=white] at (-13.0,-2.45) {\textbullet\  Lemma~\ref{lem:solvability}};
\node[draw=white] at (-8.0,-2.24) {\textbullet\  Lemma~\ref{lem:equiv_DBR_ITE}};
\node[draw=white] at (-7.88,-2.74) {\textbullet\  Corollary~\ref{lem:ITE_to_Grover}};

\draw[thick ] (-14.2, -0.8) rectangle (-6.4, -3.9);



\draw[->, ultra thick] (-5.9,-2.5) -- (-4.9,-2.5);

\node[draw=white, anchor=west] at (-4.4,-0.5) {\textbf{Applications}};
\draw[thick ] (-4.4, -0.8) rectangle (3.6, -3.9);

\node[draw=black, anchor=west] at (-4.2,-1.2) {{Strategies on Choosing Angles}};
\node[draw=white, anchor=west] at (-3.6,-1.73) {\textbullet\  Original Grover's Algorithm: Theorem~\ref{thm:justification_orig_Grover}};
\node[draw=white, anchor=west] at (-3.6,-2.13) {\textbullet\  $\pi/3$-Algorithm ($\pi/2$-Algorithm): Proposition~\ref{prop:pi/3_alg} };
\node[draw=white, anchor=west] at (-3.6,-2.53) {\textbullet\  Fixed-Point Quantum Search: Theorem~\ref{thm:qsp_formula_ITE} };

\node[draw=black, anchor=west] at (-4.2,-3.1) {{Extensions to Other Algorithms}};
\node[draw=white, anchor=west] at (-3.6,-3.63) {\textbullet\  (Oblivious) Amplitude Amplification: Theorem~\ref{prop:ite_oaa}};


\end{tikzpicture}
\caption{\textbf{Summary of our work.}
We show that ITE provides a unified perspective on Grover's algorithm, its variants, and extensions to other quantum algorithms.
To this end, we demonstrate that Grover's algorithm can be viewed as an approximation of ITE through Lemmas~\ref{lem:solvability} and~\ref{lem:equiv_DBR_ITE} and Corollary~\ref{lem:ITE_to_Grover}.  
We then provide rationales behind the choices of angles for existing methods from the ITE perspective (Theorems~\ref{thm:justification_orig_Grover} and~\ref{thm:qsp_formula_ITE}, and Proposition~\ref{prop:pi/3_alg}).
Furthermore, we extend this framework to (oblivious) amplitude amplification in Theorem~\ref{prop:ite_oaa}.
}
\label{fig:summary_flow}
\end{figure*}

Utilizing these findings, we present a unified perspective on Grover’s algorithm, its variants and extensions to widely used quantum subroutines, including amplitude amplification (AA)~\cite{brassard2002quantum} and oblivious amplitude amplification (OAA)~\cite{berry2014exponential}: see Fig.~\ref{fig:summary_flow} for the summary.
Concretely, we provide an ITE-based interpretation of the angle choices for the original Grover's algorithm and the $\pi/3$-algorithm.
As a byproduct, it also motivates a new $\pi/2$-algorithm that offers faster convergence than the $\pi/3$-algorithm without overshooting, in cases that tolerate a moderately large failure probability.
This framework further yields a new implementation of fixed-point quantum search by connecting the ITE dynamics for unstructured search and quantum signal processing (QSP)~\cite{low2017quantum,low2017optimal,equiangular,motlagh2024generalized}.
Moreover, we show that our ITE formulation (equivalently interpreted as a Riemannian optimization) naturally reproduces AA and OAA.
This interpretation exposes a unified structure across essential subroutines in modern quantum algorithms.
Thus, our findings offer not only a new interpretation of Grover’s algorithm, but also potentially provide useful thermodynamic and geometric perspectives to inspire the design of future quantum algorithms.

\section{Grover's algorithm as an approximation of ITE}
We begin by recalling Grover’s algorithm. 
To perform unstructured search, the algorithm first prepares a uniform superposition over all $N = 2^n$ computational basis states, 
\begin{equation} \label{eq:init}
    \ket{\psi_{0}} = \frac{1}{\sqrt{N}} \sum_{x=0}^{N-1} \ket{x}.
\end{equation}
The algorithm proceeds by repeatedly applying two key operations: the diffusion operator $D(\alpha)=e^{i\alpha \psi_{0}}$ and the oracle operator $U_f(\beta) = e^{i\beta \hat{H}_{f}}$, where $\alpha,\beta\in\mathbb{R}$, $\psi_{0}=|\psi_{0}\rangle\langle\psi_{0}|$ and
\begin{equation} \label{eq:projector}
    \hat{H}_f=\sum_{x\in\{x|f(x)=1\}} |x\rangle\langle x|
\end{equation}
is the projector onto the subspace of $M$ marked states.
Through $\mathcal{N}$ repeated applications of $G(\alpha_k,\beta_k) = -D(\alpha_k) U_f(\beta_k)$, we obtain a final state $\prod_{k=1}^{\mathcal{N}}G(\alpha_k,\beta_k)\ket{\psi_{0}}$ that approximates the solution state,
\begin{equation} \label{eq:solution}
\ket{\psi^*} = \frac{1}{\sqrt{M}} \sum_{x\in\{x|f(x) = 1\}} \ket{x}.
\end{equation}
The query complexity is then defined as the number of calls $\mathcal{N}$ to the oracle operator $U_{f}$.
The original work in Ref.~\cite{grover1996fast} uses $\alpha_k = \beta_k = \pi$, while its variants such as $\pi/3$-algorithm~\cite{grover2005fixed} and the fixed-point algorithm~\cite{yoder2014fixed} have been proposed to avoid the so-called souffl\'{e} problem~\cite{brassard1997searching}, where excessive Grover iterations cause the algorithm to overshoot the solution state; see App.~\ref{app:related_work} for comprehensive discussions and reviews.

We first show that the unstructured search problem can be solved using Imaginary-Time Evolution (ITE), which drives an initial state to the ground state of the Hamiltonian $\hat{H}$ via a non-unitary operator $e^{-\tau \hat{H}}$ with $\tau\in\mathbb{R}$.
Provided the initial state has non-zero overlap with the ground state, the evolution provably converges to the ground state as $\tau\to~\infty$~\cite{GellmannLow}.
Similarly, applying ITE with $\hat{H}_{f}$ to $\ket{\psi_{0}}$ yields the solution state in Eq.~\eqref{eq:solution} in the limit of large $\tau$; 
see App.~\ref{app:proof_of_solvability} for the proof.

\begin{lemma}[ITE solves the unstructured search problem] \label{lem:solvability}
     Given the projector Hamiltonian $\hat{H}_f$ in Eq.~\eqref{eq:projector} and the initial state in Eq.~\eqref{eq:init}, the ITE state converges to the solution state in Eq.~\eqref{eq:solution} as $\tau\to\infty$, i.e., 
    \begin{equation} \label{eq:grover_solution}
         \lim_{\tau\to\infty} \frac{e^{\tau \hat{H}_{f}}\ket{\psi_0}}{\|e^{\tau \hat{H}_{f}}\ket{\psi_0}\|_2}=\ket{\psi^{*}} 
\end{equation}
with the normalization factor $\|e^{\tau \hat{H}_{f}}\ket{\psi_0}\|_2=\sqrt{\braket{\psi_{0}|e^{2\tau \hat{H}_{f}}|\psi_{0}}}$.
\end{lemma}
Note that Eq.~\eqref{eq:grover_solution} involves $e^{\tau \hat{H}_{f}}$ without a negative sign, as unstructured search seeks the eigenstates of $\hat{H}_{f}$ with the largest eigenvalue.
Yet, this can also be viewed as the standard ITE, since setting the target Hamiltonian to $\hat{H}_{f} \rightarrow -\hat{H}_{f}$ converts the task into preparing its ground state and introduces a minus sign in the exponents in Eq.~\eqref{eq:grover_solution}.

Given Lemma~\ref{lem:solvability}, a natural question arises: how can Grover’s algorithm be related to ITE for unstructured search? 
We answer this using the recently-developed framework of Double-Bracket Quantum Algorithms (DBQA)~\cite{double_bracket2024}, which implements continuous differential equations called double-bracket flows (DBF)~\cite{bloch1985completely,bloch1990steepest,bloch1992completely,Brockett1991DBF,BLOCH1985103,moore1994numerical,BROCKETT1989761,deift1983ordinary,Chu_iterations,wegner1994flow,wegner2006flow,hastings2022lieb,GlazekWilson,GlazekWilson2,kehrein_flow,smith1993geometric,optimization2012,brockett2005smooth} on quantum computers.
By exploiting the equivalence between DBF and ITE (further elaborated in App.~\ref{app:ITE_and_Riemannian_grad_flow} and e.g., in Refs.~\cite{gluza_DB_QITE_2024,mcmahon2025equating}), Refs.~\cite{gluza_DB_QITE_2024,zander2025role,suzuki2025double,alghadeer2025double} have leveraged the DBQAs to design quantum algorithms for ITE. 

DBQAs first proceed by recognizing that ITE is a solution to the differential equation 
\begin{equation}\label{eq:DifEq}
    \frac{\partial \Psi(\tau)}{\partial \tau} = [[\Psi(\tau), \hat{H}],\Psi(\tau)] \, 
\end{equation}
with a pure density matrix $\Psi(\tau)=|\Psi(\tau)\rangle\langle\Psi(\tau)|$ at time~$\tau$ and $[A,B]=AB-BA$ is the standard commutator. 
This is straightforward to verify by differentiating $\Psi(\tau)$  with respect to $\tau$ for $\ket{\Psi(\tau)} = e^{-\tau\hat{H}}\ket{\Psi(0)}/\|e^{-\tau\hat{H}}\ket{\Psi(0)}\|_{2}$. 
Since the commutator $[\Psi(\tau), \hat{H}]$ in Eq.~\eqref{eq:DifEq} is anti-Hermitian, this operator can be regarded as a generator of unitary evolution, in analogy with the Schr\"{o}dinger equation.
For arbitrary Hamiltonians $\hat{H}$, the discretized solution of Eq.~\eqref{eq:DifEq} can hence be approximated to first order as  
\begin{equation} 
    \ket{\psi_s} = e^{s[\psi_{0},\hat{H}]}\ket{\psi_0}, 
\end{equation}
where $s$ is a time duration. This approximation generally incurs an error of $\mathcal{O}(s^{2})$~\cite{double_bracket2024,robbiati2024double,xiaoyue2024strategies, gluza_DB_QITE_2024}.
However, whenever the Hamiltonian is a projector, i.e., $\hat{H}^2=\hat{H}$ as is the case for Grover's algorithm, the ITE state for any imaginary time $\tau$ can be realized exactly by Eq.~\eqref{eq:dbr} (see App.~\ref{app:DBR_ITE_relation} for the proof). 

\begin{lemma}[Equivalence of ITE and commutator flow for projector Hamiltonians]
\label{lem:equiv_DBR_ITE} 
Let $\hat{H}_{f}$ be the projector Hamiltonian in Eq.~\eqref{eq:projector}.
Then, for any ITE evolution time $\tau$, there exists a time duration $s_{\tau}$ such that 
\begin{equation}\label{eq:dbr}
    \frac{e^{\tau \hat{H}_f}\ket{\psi_0}}{\|e^{\tau \hat{H}_f}\ket{\psi_0}\|_2} =e^{s_\tau[ \hat{H}_f,\psi_{0}]}\ket{\psi_0} =:\ket{\psi_{s_{\tau}}}.
\end{equation}
Since the middle expression in Eq.~\eqref{eq:dbr} covers all the states generated by ITE, we hereafter refer to $\ket{\psi_{s_{\tau}}}$ as the ITE state. 
\end{lemma}

In the second step of DBQA, the exponential of the commutator in Eq.~\eqref{eq:dbr} is approximately implemented using product formula techniques~\cite{dawson2006solovay,commutator_approximation_2022,product_formula2013}.
A standard approach is the group commutator
\begin{equation} \label{eq:group_commutator}
     e^{s[\hat{H}_f,\psi_{0}]}=e^{i\sqrt{s}\psi_{0}}e^{i\sqrt{s}\hat{H}_f}e^{-i\sqrt{s}\psi_{0}}e^{-i\sqrt{s}\hat{H}_f} + \mathcal O(s^{3/2})\ .
\end{equation}
This can be systematically generalized to higher-order approximations~\cite{product_formula2013}
\begin{equation}
\begin{split} \label{eq:sequence_gr}
     e^{s[\hat{H}_f,\psi_{0}]}  &=e^{it_{2\mathcal{N}}\psi_{0}}\ldots e^{it_{3}\hat{H}_f}e^{it_{2}\psi_{0}} e^{it_{1}\hat{H}_f}  + \mathcal{O}(s^{m/2}) 
\end{split}
\end{equation}
for certain $m\in\mathbb{Z_{+}}$ and $t_{k}=c_{k}\sqrt{s}$ with angles $\{c_{k}\}_{k=1}^{2\mathcal{N}}$; see App.~\ref{app:product_formula} for specifics in choosing angles.
Remarkably, the two exponentials in Eq.~\eqref{eq:sequence_gr} are exactly oracle operators $U_{f}(t_{k})$ and diffusion operators $D(t_{k})$.
Namely, these approximations of Eq.~\eqref{eq:dbr} naturally lead to the structure of Grover's algorithm up to a global phase.
\begin{corollary}[Grover's algorithm is an approximation of ITE] \label{lem:ITE_to_Grover}
    Let $\{t_i\}_{i=1}^{2\mathcal{N}}$ be parameters of the product formula approximation of the ITE dynamics using $\mathcal{N}$ queries, as shown in Eq.~\eqref{eq:sequence_gr}. Then,
    \begin{equation} \label{eq:ite_grover_approx}
    \begin{split}
    \left\|(-1)^{\mathcal{N}}
    \prod_{k=1}^{\mathcal{N}}G(t_{2k},t_{2k-1})\ket{\psi_{0}} - e^{s[\hat{H}_{f},\psi_{0}]}\ket{\psi_0} \right\|_{2} \le \mathcal O(s^{m/2}).
    \end{split}
    \end{equation}
\end{corollary}

Building on this observation, we further investigate how the ITE formulation offers insight into the design of original Grover’s algorithm and its variants; for instance, choosing angles $\{(\alpha_k, \beta_{k})\}$ for efficient implementation is non-trivial and not specified in Corollary~\ref{lem:ITE_to_Grover}.
More concretely, the formulation allows us to establish the following results: 

\begin{enumerate}
    \item An alternative perspective for the strategy underlying the original Grover' algorithm~\cite{grover1996fast}, i.e., $\alpha_k=\beta_k=\pi$.
    \item A unified framework that recovers $\pi/3$-algorithm~\cite{grover2005fixed} and yields a new fixed-point quantum search distinct from the recent approach of Ref.~\cite{yoder2014fixed}.
    \item A new systematic ITE formulation of widely-used quantum algorithms such as amplitude amplification (AA)~\cite{brassard2002quantum} and oblivious amplitude amplification (OAA)~\cite{berry2014exponential}.
\end{enumerate}

\section{Justifying the strategy of original Grover's algorithm via ITE formulation}
First, the ITE formulation admits a natural justification of the choice of original Grover's algorithm, $\alpha_k=\beta_k=\pi$.
Lemmas~\ref{lem:solvability},~\ref{lem:equiv_DBR_ITE} and Corollary~\ref{lem:ITE_to_Grover} show that unstructured search can be solved by ITE, reducing the unstructured search task to an efficient unitary compilations of ITE dynamics using the elementary gate set~$\{D(\alpha), U_{f}(\beta)\}$.
Concretely, the solution state $\ket{\psi^{*}}$ is obtained by evolving the initial state $\ket{\psi_{0}}$ for the optimal time duration $s^{*}$; i.e., $\ket{\psi^{*}}=e^{s^{*}[\hat{H}_f,\psi_{0}]}\ket{\psi_{0}}$ for 
\begin{equation} \label{eq:opt_s}
        s^{*}=\arccos(\sqrt{E_{0}})/\sqrt{V_{0}},
\end{equation}
where $E_0 = \braket{\psi_0|\hat{H}_f|\psi_0} = M/N$ and $V_0 = \braket{\psi_0|(\hat{H}_f - E_0)^2|\psi_0} = E_0(1 - E_0)$: see App.~\ref{app:geodesic_complexity} for the proof.
Hence, the goal of this task is reduced to approximating the optimal ITE using the minimal number of oracle queries~$\mathcal{N}$, by properly choosing the angles $\{(\alpha_{k},\beta_{k})\}$.

We firstly observe that the original Grover's algorithm, where the angles are chosen as $\alpha_k=\beta_k=\pi$, \textit{exactly} realizes the ITE state $\ket{\psi_{s_{\tau}}}$ in Eq.~\eqref{eq:dbr}.
Moreover, this choice provides a strategy to maximize the evolution achieved in each Grover iteration, particularly when the solution state has small overlap with the initial state.
Specifically, in the single-step approximation represented by Eq.~\eqref{eq:group_commutator}, setting $\sqrt{s}=\pi$ yields the  largest effective evolution time; that is, these angles drive the system the most along the ITE trajectory among all possible choices.
A formal proof is given in App.~\ref{app:rationale_original_Grover}.

\begin{theorem}[Original Grover's algorithm]
\label{thm:justification_orig_Grover}
The original Grover's algorithm generates the state $\ket{\psi_{s_{\tau}}}$ in Eq.~\eqref{eq:dbr}; that is, there exists a parameter $s(\mathcal{N})$ such that 
\begin{align}
    (-1)^{\mathcal{N}}G(\pi,\pi)^{\mathcal{N}}\ket{\psi_{0}} = e^{s(\mathcal{N})[ \hat{H}_f,\psi_{0}]}\ket{\psi_0}.
\end{align}
Moreover, the original Grover algorithm maximizes the fidelity of the first iteration. That is, within the first order approximations in Eq.~\eqref{eq:group_commutator}, corresponding to $\mathcal{N}=2$, the fidelity $F_{2}=|\braket{\psi^{*}|\prod_{k=1}^{2}G(\alpha_k,\beta_k)|\psi_{0}}|^2$ is maximized when $\alpha_k=\beta_k=\pi$, provided $E_{0}\le 1/8$, e.g.,  when $M\ll N$.

\end{theorem}

Theorem~\ref{thm:justification_orig_Grover} supports the choice of the original Grover algorithm~\cite{grover1996fast} from the ITE perspective: this choice realizes the ITE state with equality (which is somewhat unexpected given Eq.~\eqref{eq:ite_grover_approx}) and maximizes the fidelity for small initial overlap and a limited number of iterations $\mathcal{N}$.
However, the same mechanism also explains the possibility of overshooting.
As the fidelity with the solution state increases, repeatedly applying a maximal-step evolution drives the system beyond the optimal ITE state, leading to a different state in the trajectory.

\section{ITE formulation for fixed-point quantum search}
The ITE formulation also sheds light on fixed-point quantum search, which addresses the overshooting problem.

Theorem~\ref{thm:justification_orig_Grover} identifies the origin of overshooting and hence emphasizes the importance of controlling the step size for the ITE approximation.
If the energy $E_{0}=M/N$ is known, this issue can be avoided; App.~\ref{app:optimal_complexity} shows that a suitable choice of $\{(\alpha_{k},\beta_{k})\}$ from a simple product formula approximation yields the solution state with the optimal query scaling $\mathcal{O}(\sqrt{N})$.
Furthermore, Refs.~\cite{long2001Grover,Roy2022Grover} demonstrate more sophisticated tunings that could improve the constant prefactor for the optimal scaling.

However, $E_{0}$ is unknown in practice.
An early remedy of this difficulty is provided by $\pi/3$-algorithm~\cite{grover2005fixed} defined as
\begin{equation} \label{eq:pi/3-alg}
    U_{k+1} = U_{k} D(\theta) U_{k}^{\dagger} U_{f}(\theta) U_{k},
\end{equation}
where $\theta=\pi/3$ and $k$ denotes a recursion step.
It has been shown that one recursive step improves the fidelity from $1-\delta$ to $1-\delta^3$ for infidelity $\delta$, which guarantees monotonic convergence~\cite{grover2005fixed,chakraborty2005bounds,grover2006quantum}.

Our ITE formulation yields two key observations on this approach.
First, the recursive structure of the $\pi/3$-algorithm arises naturally within the double-bracket quantum imaginary-time evolution (DB-QITE) framework~\cite{gluza_DB_QITE_2024}, which implements ITE for general Hamiltonians.
Second, $\pi/3$ is not the only possible choice in angle for achieving monotonic convergence; in fact any angle smaller than $\pi/2$ will suffice to avoid overshooting. 
A detailed analysis is provided in App.~\ref{app:pi/3_algorithm}.

\begin{prop}[ITE formulation for the $\pi/3$-algorithm and the maximal step size with monotonic convergence] \label{prop:pi/3_alg}
The DB-QITE framework in Ref.~\cite{gluza_DB_QITE_2024} reproduces the recursive structure of $\pi/3$-algorithm of Eq.~\eqref{eq:pi/3-alg}.
Moreover, the choice of $\theta=\pi/2$ is the largest step size for the ITE approximation while ensuring monotonic fidelity contraction, independent of the number of marked items $M$.
\end{prop}

From the perspective of ITE dynamics, larger step sizes will improve the rate of convergence.
The analysis confirms this intuition -- picking the maximal step size for monotonic convergence, namely $\pi/2$, leads to faster convergence than $\pi/3$, when $E_0 = M/N$ is small.
In particular, for $E_{0} \ll 1$, the \textit{$\pi/2$-algorithm} yields a superlinear advantage in queries, i.e., $\mathcal{N}\in \mathcal{O}(N^{1/1.46})$, whereas $\pi/3$-algorithm scales as $\mathcal{N}\in \mathcal{O}(N)$.
However, in the regime $E_{0} > 2/3$ where the state is already close to the solution state, $\pi/3$-algorithm is significantly more effective, exhibiting an exponentially faster decay rate with respect to the recursion step $k$. 
Consequently, when a modest failure probability is allowed, $\pi/2$-algorithm can provide faster convergence than $\pi/3$-algorithm; with a termination condition set to a final-state infidelity below 0.1, numerical simulations show that $\pi/2$-algorithm requires no more queries than $\pi/3$-algorithm for any initial state (see App.~\ref{app:pi/3_algorithm}).
To the best of our knowledge, this is the first demonstration that choosing $\pi/2$ improves convergence in the small-$E_{0}$ regime.

\medskip
Although $\pi/3$-algorithm, and our $\pi/2$-algorithm, avoid the overshooting, this robustness comes at the expense of optimal query complexity. 
That is, the quadratic advantage in $N$ is lost, rendering its performance comparable to classical search. 
Another fixed-point quantum algorithm of Ref.~\cite{yoder2014fixed} was then proposed to resolve this limitation while retaining the optimal $\mathcal{O}(\sqrt{N})$ scaling.
This approach employs recursive quasi-Chebyshev polynomial constructions~\cite{yoder2014fixed,li2024revisiting} to determine the angles so that the final state has infidelity at most $\delta$ with respect to the solution state, for any number of marked items~$M$.

The ITE formulation offers an alternative realization of fixed-point quantum search.
In particular, we show that Grover iterations, viewed as the approximation of ITE dynamics, implement quantum signal processing (QSP)~\cite{low2017quantum,low2017optimal,equiangular,motlagh2024generalized}, which realizes polynomial transformations through interleaved sequences of unitaries on a two-dimensional subspace.
Further details on QSP are provided in End Matter and App.~\ref{app:thm_of_qsp_ite}.
This observation implies that techniques developed for QSP naturally enable efficient implementations of ITE.
Consequently, a new fixed-point quantum search emerges as a direct outcome of the QSP construction for ITE.
A complete proof is provided in App.~\ref{app:thm_of_qsp_ite}.

\begin{theorem}[ITE formulation provides a fixed-point search via a QSP framework]\label{thm:qsp_formula_ITE}
The Grover iterations $\prod_{k=1}^{\mathcal{N}}G(\alpha_k,\beta_k)\ket{\psi_{0}} $ implement QSP with a polynomial of degree $K=2\mathcal{N}$. 
In particular, QSP angles designed to approximate the sign function yields a fixed-point quantum search.
\end{theorem}

The proof proceeds by showing that the ITE trajectory follows a geodesic, i.e.,  a shortest path connecting the initial and the solution state.
This geodesic is defined in the complex projective space $\mathbb{C}P^{N-1}$ equipped with the Fubini–Study metric~\cite{anandan1990geometry,bengtsson2017geometry,mukunda1993quantum}.
See Apps.~\ref{app:geometry_review} and~\ref{app:geodesic_complexity}  for the details.
A key property of such geodesics is that the dynamics are confined to a two-dimensional subspace, facilitating the analysis and development of the algorithm.
In our setting, this subspace is spanned by $\{\ket{\psi_{0}}, \ket{\psi_{0}^{\perp}}\}$, where 
 \begin{equation} \label{eq:perp_init_state}
        \ket{\psi_{0}^{\perp}}= \frac{[\hat{H}_{f},\psi_{0}]}{\sqrt{V_0}}\ket{\psi_{0}} = \frac{\hat{H}_{f}-E_{0}I}{\sqrt{E_{0}(1-E_{0})}} \ket{\psi_{0}}
\end{equation}
with $\braket{\psi_{0}|\psi_{0}^{\perp}}=0$.
Note that our choice of basis differs from that adopted in most previous studies including Refs.~\cite{nielsen2010quantum,yoder2014fixed,li2024revisiting}, where the relevant subspace is spanned by
$\{\ket{\psi^{*}}, \ket{\psi^{*}_{\perp}}\}$ with $\ket{\psi^{*}_{\perp}}=(\hat{H}_{f}-I)\ket{\psi_{0}}/\sqrt{1-E_{0}}$.
The full ITE dynamics is characterized within the basis and hence enables a direct application of the QSP framework.
With this equivalence, implementing the sign function as a filter yields a new fixed-point quantum search.

\section{Wider applications of the ITE formulation beyond Grover's algorithm}
Lastly, we show the ITE formulation extends beyond Grover’s algorithm and applies to widely-used quantum subroutines.
Consider the task of preparing $V\ket{\phi}$ for a target unitary $V$ and an arbitrary initial state $\ket{\phi}$.
While $V$ is unavailable, we assume access to a unitary $U$ satisfying
\begin{equation}
    U\ket{0}\ket{\phi} = \sqrt{p}\ket{0}V\ket{\phi} + \sqrt{1-p} \ket{1}\ket{\phi'}
    \label{eq OAA def}
\end{equation}
for some $p\in(0,1)$, with $\ket{\phi'}$ an arbitrary state of the same system size. 
The goal is then to prepare the target state $V\ket{\phi}$ using as few queries to $U$ as possible.

One approach to solving this problem is oblivious amplitude amplification (OAA)~\cite{berry2014exponential}, a variant of amplitude amplification (AA)~\cite{brassard2002quantum} that increases the success probability $p$ through a sequence of two types of operators.
By applying the same ITE formulation used to derive Grover’s algorithm and reformulating the problem appropriately, the structure of OAA, $\prod_{k=1}^{\mathcal{N}}  U \tilde{D}(\alpha_{k}) U^{\dagger} \tilde{D}(\beta_{k}) \ket{\psi_0}$ 
with $ \tilde{D}(\theta)=e^{i\theta \ket{0}\bra{0} }\otimes I$, emerges naturally. 
The full proof is provided in App.~\ref{app:thm_oaa}.

\begin{theorem}[OAA is an approximation of ITE] \label{prop:ite_oaa}
    OAA is a product formula approximation of ITE in Eq.~\eqref{eq:sequence_gr} with the Hamiltonian $\hat{H}_{f}=\ket{0}\bra{0} \otimes V \ket{\phi} \bra{\phi} V^{\dagger}$ and the initial state $\ket{\psi_{0}}=U\ket{0}\ket{\phi}$.
\end{theorem}

The proof proceeds by first reproducing AA by choosing $\hat{H}_{f}$ and $\ket{\psi_{0}}$ in Theorem~\ref{prop:ite_oaa}.
This derivation initially yields full reflections of the form $e^{i\theta (\ket{0}\bra{0}\otimes \ket{\phi_{0}}\bra{\phi_{0}})}$ and appears to require direct access to the target unitary $V$.
Yet, the key observation here is that, for any state of the form $\ket{\Phi}=z_{0} \ket{0}\ket{\phi} + z_{1} \ket{1}\ket{\phi'} $ with $z_{0},z_{1} \in \mathbb{C}$, the identity
\begin{equation}
    e^{i\theta (\ket{0}\bra{0}\otimes \ket{\phi}\bra{\phi})} \ket{\Phi} = (e^{i\theta \ket{0}\bra{0}}\otimes I) \ket{\Phi}
\end{equation}
holds. 
Consequently, the partial reflection $\tilde{D}(\theta)$ suffices to realize AA without direct access to $V$.

Theorem~\ref{prop:ite_oaa} therefore demonstrates that the ITE formulation captures OAA, a core subroutine in Hamiltonian simulation~\cite{berry2014exponential,berry2015simulating}, thermal-state preparation~\cite{rall2023thermal}  and related applications~\cite{gilyen2019quantum}. 
Since the same reasoning applies to standard AA, the ITE framework provides a unified perspective on a broader class of quantum algorithms and suggests a systematic route toward the development of new algorithmic constructions.

\section{Discussion}

In the comment piece titled \textit{``Why haven't more quantum algorithms been found?"}~\cite{shor2003haven}, P.~Shor argued ``quantum computers operate in a manner so different from classical computers
that our techniques for designing algorithms and our intuitions for understanding the
process of computation no longer work''. 
Here, however, we show that Grover's algorithm, as well as widely-used quantum subroutine algorithms such as AA and OAA, can be viewed through the well-established lens of ITE.

Furthermore, although this connection has not been explicitly emphasized in the main text, ITE is closely linked to Riemannian optimization as detailed in the End Matter and App.~\ref{app:ITE_and_Riemannian_grad_flow}.
This suggests that Grover's algorithm simply performs Riemannian optimization, a standard classical optimization strategy, but on the manifold of unitaries.
Thus, our result provides a deep geometric implications on the design of quantum algorithms.

In this manner, our results highlight the role of geometric analysis in quantum algorithm design.
As exploited in Theorem~\ref{thm:qsp_formula_ITE}, explicit knowledge of the geodesic enables a reduction in effective dimensionality and thereby simplifies both the design and analysis of quantum algorithms.
Moreover, Apps.~\ref{app:geodesic_complexity} and~\ref{app:optimal_complexity} confirm that the query complexity associated with the geodesic length achieves the known optimal scaling in unstructured search.
These findings reinforce the significance of geometry-inspired approaches in developing quantum algorithms, advocated in Ref.~\cite{nielsen2006quantum}.
While determining geodesics may be challenging in general, our results underscore their relevance for the design and analysis of quantum algorithms via the geometry.

While the optimal query complexity for unstructured search is limited to a quadratic speed-up, ITE in general converges exponentially in $\tau$ to the target state.
This seeming tension could stem from the first-order approximation used to derive the Grover iteration from ITE.
Exploring this difference could clarify the fundamental cost of realizing thermodynamically-inspired approaches in a fully unitary framework. 
Specifically, given the optimality of Grover’s algorithm, it may suggest that implementing ITE through real-time (unitary) evolution may benefit a quadratic advantage at best.
For instance, DB-QITE, the ITE implementation via the DBQA proposed in prior work~\cite{gluza_DB_QITE_2024}, can exhibit exponential scaling; nevertheless, our analysis may help identify potential quadratic advantages in specific setting other than projector Hamiltonians.
A detailed investigation of this behavior would be an interesting direction for future study.

\medskip
\textit{Acknowledgments.}
ZH acknowledges support from the Sandoz Family Foundation-Monique de Meuron program for Academic Promotion.
MG, JS, BH and NN are supported by the start-up grant of the Nanyang Assistant Professorship at the Nanyang Technological University in Singapore. NN also acknowledges the National
Research Foundation, Singapore (W24Q3D0001) and the Tier 1 MOE grant RT1/23.

\bibliographystyle{naturemag}
\bibliography{apssamp}



\section*{End Matter}

\paragraph*{Links between Riemannian Algorithms, ITE, and DBF --} 

Double-bracket flows (DBF)~\cite{bloch1985completely,bloch1990steepest,bloch1992completely,Brockett1991DBF,BLOCH1985103,moore1994numerical,BROCKETT1989761,deift1983ordinary,Chu_iterations,wegner1994flow,wegner2006flow,hastings2022lieb,GlazekWilson,GlazekWilson2,kehrein_flow,smith1993geometric,optimization2012,brockett2005smooth} are matrix-valued ordinary differential equations that have been employed to perform matrix diagonalization, QR decomposition, eigenvalue sorting, and related tasks~\cite{optimization2012,Brockett1991DBF,deift1983ordinary,Chu_iterations}.
R.~Brockett originally introduced this flow in the context of minimizing a least-squares cost between two matrices using steepest descent techniques.

Consider the set of matrices $$\mathcal M(A) = \{U AU^\dagger\text{ s.t. } U^{-1}=U^\dagger\}$$ generated by unitary evolution of a Hermitian operator $A$ and the least-squares cost $ C_{B}(X)=-\frac{1}{2}\|X-B\|^2_\text{HS}$
with $X\in \mathcal{M}(A)$ and a Hermitian $B$.
Here, $\|\cdot\|_\text{HS}$ denotes the Hilbert-Schmidt norm.
Then, the  Riemannian gradient (i.e., the steepest gradient) of the cost $C_{B}$ at $X$ satisfies two properties: the tangency condition and the compatibility condition.
In short, these conditions state that the Riemannian gradient lies within the tangent space (tangency), and that its inner product with any tangent vector $\Omega$ equals the directional derivative of $C_{B}$ along $\Omega$ (compatibility); for further details, see App.~\ref{app:ITE_and_Riemannian_grad_flow}.
In this case, the Riemannian gradient is given by~\cite{optimization2012,riemannianflowPhysRevA.107.062421}
\begin{align} \label{eq: grad_riemannian_em}
\text{grad}_X C_B(X) = [[X,B],X].
\end{align}
Thus, the steepest descent flow defined as the trajectory $A(t)\in \mathcal{M}(A)$ satisfying Eq.~\eqref{eq: grad_riemannian_em} leads to the differential equation 
\begin{align}\label{eq: BrockettODE_em}
   \frac{\partial A(t)}{\partial t} = [[A(t), B],A(t)] \,
\end{align}
known as the DBF.
By properly choosing the matrix $B$, the flow has been utilized for the tasks mentioned above; see Refs.~\cite{bloch1985completely,bloch1990steepest,bloch1992completely,Brockett1991DBF,BLOCH1985103,moore1994numerical,BROCKETT1989761,deift1983ordinary,Chu_iterations,wegner1994flow,wegner2006flow,hastings2022lieb,GlazekWilson,GlazekWilson2,kehrein_flow,smith1993geometric,optimization2012,brockett2005smooth} for other applications of DBF.

A form of the DBF can naturally emerges from ITE.
By differentiating ITE, 
\begin{equation}
    \ket{\Psi(\tau)} = \frac{e^{-\tau \hat{H}}\ket{\Psi(0)}}{\|e^{-\tau \hat{H}}\ket{\Psi(0)}\|}, \, 
\end{equation}
with respect to $\tau$, we obtain 
\begin{equation}
    \frac{\partial \Psi(\tau)}{\partial \tau} = [[\Psi(\tau), \hat{H}],\Psi(\tau)], \, \Psi(\tau)=|\Psi(\tau)\rangle\langle\Psi(\tau)|,
\end{equation}
which matches the DBF in Eq.~\eqref{eq: BrockettODE_em}.
Since DBF is the steepest-descent trajectory of the least-squares cost, ITE also follows the same Riemannian gradient flow.
More specifically, for a set of pure states $\mathcal{M}(\Psi)=\{U\Psi U^{\dagger}\text{ s.t. } U^{-1}=U^\dagger\}$ with a positive semidefinite $\Psi$ satisfying $\Psi^{2}=\Psi,\, \mathrm{Tr}[\Psi]=1$, the cost function can be written as
\begin{equation}
     C_{\hat{H}}(\Psi(\tau)) = -\frac{1}{2}\|\Psi(\tau)-\hat{H}\|_{\text{HS}}^2 = \mathrm{Tr}[\Psi(\tau)\hat{H}] + \text{const.},
\end{equation}
implying that ITE follows the Riemannian gradient flow associated with energy minimization.
Since Grover's algorithm can be reformulated via ITE, the Grover's iterations implements the approximation of the associated Riemannian gradient flow.
Further discussion is provided in App.~\ref{app:ITE_and_Riemannian_grad_flow}.

\medskip
\paragraph*{ITE formulation for fixed-point quantum search via QSP --}

We show how our ITE formulation leads to a new implementation of the fixed-point quantum search algorithm~\cite{yoder2014fixed}.
While the original approach determines the required angles using recursive quasi-Chebyshev polynomials~\cite{yoder2014fixed,li2024revisiting}, our approach instead builds on the observation that the Grover iteration performs quantum signal processing (QSP)~\cite{low2017quantum,low2017optimal,equiangular,motlagh2024generalized}.

QSP is a well-established technique for realizing polynomial transformations of a scalar input through interleaved sequences of parameterized unitaries on a two-dimensional subspace.
Specifically, given a signal processing operator $S_{Z}(\phi)=diag(e^{i\phi},e^{-i\phi})$ and a signal operator
\begin{equation}
     R(x)=\begin{pmatrix}x & \sqrt{1-x^2}  \\  \sqrt{1-x^2} & -x\end{pmatrix},
\end{equation}
a QSP sequence applied to $\ket{0}=(1,0)^{T}$, with an appropriately chosen set of angles $\{\phi_{k}\}_{k=0}^{K}$, outputs a polynomial 
\begin{equation} \label{eq:qsp_conv_main}
    p_{QSP}(x)=\bra{0}S_{Z}(\phi_{K}) \prod_{k=0}^{K-1} R(x)S_{Z}(\phi_{k})\ket{0}
\end{equation}
satisfying the conditions in App.~\ref{app:qsp_overview}.
We refer to this construction as ($R(x)$, $S_{Z}$, $\ket{0}$)-QSP. 
Since the existence of phases for a wide class of achievable functions in various QSP settings has been rigorously established, QSP underpins many modern quantum algorithms~\cite{martyn2021grand,low2017quantum,low2017optimal,equiangular,motlagh2024generalized}; see App.~\ref{app:qsp_overview} for details.

We show that Grover iterations implement the ($R(x)$, $S_{Z}$, $\ket{0}$)-QSP in the two-dimensional subspace spanned by the initial and the solution states.
Most previous studies including Refs.~\cite{nielsen2010quantum,yoder2014fixed,li2024revisiting} consider a subspace spanned by $\{\ket{\psi^{*}}, \ket{\psi^{*}_{\perp}}\}$ with $\ket{\psi^{*}_{\perp}}=(\hat{H}_{f}-I)\ket{\psi_{0}}/\sqrt{1-E_{0}}$.
We instead adopt the basis $\{\ket{\psi_0}, \ket{\psi_0^\perp}\}$ defined by the initial state $\ket{\psi_{0}}$ and its orthogonal counterpart $ \ket{\psi_0^\perp}$ given in Eq.~\eqref{eq:perp_init_state}.
This basis is used to analyze the geodesic property of Grover's algorithm in App.~\ref{app:optimal_complexity}.
In this basis, the diffusion and oracle operations take the form
\begin{align}
    &D(\alpha)=  e^{i\alpha\psi_{0}}  = e^{i\alpha/2}S_{Z}(\alpha/2) \\
    &U_{f}(\beta)=e^{i\beta \hat{H}_{f}} = e^{i\beta/2} R(\sqrt{E_{0}}) S_{Z}(\beta/2) R(\sqrt{E_{0}}),
\end{align}
respectively. 
This leads to Theorem~\ref{thm:qsp_formula_ITE} proven in Appendix~\ref{app:thm_of_qsp_ite}; i.e., 
\begin{equation} \label{eq:dbr_qsp_em}
\begin{split}
    \prod_{k=1}^{\mathcal{N}}G(\alpha_k,\beta_k)\ket{\psi_{0}} = S_{Z}(\phi_{2\mathcal{N}}) \prod_{k=0}^{2\mathcal{N}-1} R(\sqrt{E_{0}}) S_{Z}(\phi_{k}) \ket{0},
\end{split}
\end{equation}
where $\phi_{0}=\mathcal{N}\pi+\sum_{l=1}^{N}(\alpha_l+\beta_{l})/2$, $\phi_{2l-1}=\beta_{l}/2$ and $\phi_{2l}=\alpha_{l}/2$ for $l=1,\ldots,\mathcal{N}$.

In the $\{\ket{\psi_{0}}, \ket{\psi_{0}^{\perp}}\}$ basis, the ITE state of Eq.~\eqref{eq:dbr} 
takes the simple form $ \ket{\psi_{s}} = (\cos(s\sqrt{V_{0}}),  \sin(s\sqrt{V_{0}}))^{T}$ with $V_{0}=E_{0}(1-E_{0})$; see the proof in App.~\ref{app:geodesic_complexity}.
It then follows that, if the function $\cos(sx\sqrt{1-x^2})$ can be realized by ($R(x)$, $S_{Z}$, $\ket{0}$)-QSP, the ITE state can be realized \textit{without prior knowledge of $E_{0}$}.
This is because the task reduces to tailoring QSP angles that approximate the desired polynomial.
We show in App.~\ref{app:thm_of_qsp_ite} that the target function meets the achievability conditions by the QSP convention. 
Thus the Jacobi-Anger expansion, which has been exploited for the efficient construction of trigonometric functions in the QSP community~\cite{gilyen2019quantum,martyn2021grand}, can be utilized.
Therefore, for sufficiently large $\mathcal{N}$, QSP may approximate the ITE state effectively.

We further use Theorem~\ref{thm:qsp_formula_ITE} to derive a new set of angles for a fixed-point algorithm.
Concretely, we realize the solution state without information on $E_0$, even though the optimal time duration $s^{*}$ in Eq.~\eqref{eq:opt_s} depends on it.
We note that applying $R(\sqrt{E_{0}})$ to the initial state directly prepares the solution state $\ket{\psi^{*}}=(\sqrt{E_{0}},\sqrt{1-E_{0}})^{T}$.
However, we cannot apply a single signal operator alone with the Grover iterations, as the oracle consists of $K=2\mathcal{N}$ signal operators.
Yet, if we find QSP sequence $W$ of length $K=2\mathcal{N}-1$ such that $W\ket{0}=\ket{0}$, the Grover iteration with $\mathcal{N}$ (i.e., $K=2\mathcal{N}$ in the QSP convention)
\begin{equation}
    \prod_{k=1}^{\mathcal{N}}G(\alpha_k,\beta_k)\ket{\psi_{0}}=S_{Z}(\phi_{2\mathcal{N}})R(\sqrt{E_{0}})W\ket{0}
\end{equation}
can prepare $\ket{\psi^{*}}$ by setting $\phi_{2\mathcal{N}}=0$.

The remaining step is to identify a function that QSP can implement to realize $W\ket{0}=\ket{0}$. 
The key requirement is $p_{QSP}(\sqrt{E_0}) = \langle 0| W | 0 \rangle = 1$, while the values of $ p_{QSP}(x) $ for negative $x$ are irrelevant since $\sqrt{E_0}$ is always positive. 
Also, as $W$ involves an odd number ($2\mathcal{N}-1$) of applications of  the signal operator, the function $p_{QSP}(x)$ must be odd.
A natural choice is the sign function and hence we consider $p_{QSP}(x) = \text{sgn}(x)$.

The sign function can be efficiently approximated via a polynomial function of degree $d\in\mathcal{O}(1/\sqrt{E_{0}}\log(1/\epsilon))$ with error $\epsilon$ and thus this approach achieves the optimal scaling in $N$ (see App.~\ref{app:thm_of_qsp_ite}).
However, the sign function does not satisfy the conditions for achievable functions by ($R(x)$, $S_{Z}$, $\ket{0}$)-QSP. 
Nevertheless, Ref.~\cite{martyn2021grand} shows that introducing a small imaginary part into the QSP phases allows the sign function to be approximated. 
Although this might sacrifice the optimal scaling with $N$ as shown in App.~\ref{app:numerics}, we obtain a scaling similar to the standard fixed-point Grover algorithm.


\clearpage
\newpage
\onecolumngrid
\part{Appendix}
\parttoc 

\appendix

\newtheorem{theoremA}{Theorem}[section]
\newtheorem{propositionA}[theoremA]{Proposition}
\newtheorem{lemmaA}[theoremA]{Lemma}
\newtheorem{definitionA}[theoremA]{Definition}
\newtheorem{corollaryA}[theoremA]{Corollary}
\newtheorem{remarkA}[theoremA]{Remark}

\renewcommand{\thetheoremA}{\Alph{section}.\arabic{theoremA}}

\newpage

\section{A quick overview of Grover's algorithm} \label{app:related_work}

We begin by briefly reviewing the setup of Grover’s algorithm.
Given a search space $\mathcal{X}\in\{0,1,\ldots,N-1\}$, the goal of unstructured search is to identify $M$ target items for which a binary function $f(x)\in\{0,1\}$ satisfies $f(x)=1$.
The algorithm starts by preparing a uniform superposition over all $N = 2^n$ computational basis states by applying the Hadamard gate $H$ to each of the $n$ qubits;
\begin{equation} \label{app_eq:init}
    \ket{\psi_{0}} = H^{\otimes n}\ket{0^{\otimes n}} = \frac{1}{\sqrt{N}} \sum_{x=0}^{N-1} \ket{x}.
\end{equation}
The algorithm then proceeds by repeatedly applying two key operations: the diffusion operator $D(\alpha)$ and the oracle operator $U_f(\beta)$ with $\alpha,\beta\in\mathbb{R}$ defined as
\begin{eqnarray}
    U_{f}(\beta) &= e^{i\beta \hat{H}_{f}}=I-(1-e^{i\beta})\hat{H}_{f},\\
    D(\alpha) &= e^{i\alpha \psi_{0}}=I-(1-e^{i\alpha})\psi_{0},\label{eq:diffusion_op}
\end{eqnarray}
respectively.
Here, $\psi_{0}=|\psi_{0}\rangle\langle\psi_{0}|$ and 
\begin{equation} \label{app_eq:projector}
    \hat{H}_f=\sum_{x\in\{x|f(x)=1\}} |x\rangle\langle x|
\end{equation}
is the projector onto the subspace of $M$ marked states.
Through $\mathcal{N}$ repeated applications of $G(\alpha_k,\beta_k) = -D(\alpha_k) U_f(\beta_k)$ to the initial state, we obtain a final state $\prod_{k=1}^{\mathcal{N}}G(\alpha_k,\beta_k)\ket{\psi_{0}},$ which is intended to approximate the solution state
\begin{equation} \label{app_eq:solution}
\ket{\psi^*} = \frac{1}{\sqrt{M}} \sum_{x\in\{x|f(x) = 1\}} \ket{x}.
\end{equation}
The original work in Ref.~\cite{grover1996fast} uses $\alpha_k = \beta_k = \pi$ and shows that the query complexity defined as the number of calls $\mathcal{N}$ to the oracle operator $U_{f}$ is optimal~\cite{nielsen2010quantum}.

Following the introduction of Grover’s algorithm~\cite{grover1996fast}, numerous works have attempted to characterize its computational capabilities.
A key line of research has rigorously established its asymptotic optimality in terms of query complexity~\cite{bennett1997strengths,beals1998tight,zalka1999grover,nielsen2010quantum}.
In parallel, other work has addressed the so-called souffl\'{e} problem~\cite{brassard1997searching}, in which repeated applications of Grover’s iteration can overshoot the optimal solution; this makes it difficult to determine the correct number of iterations without prior knowledge of the number of marked items.
The $\pi/3$-algorithm~\cite{grover2005fixed}, where the angles are set to
\begin{equation}
    \alpha_k=\beta_k=\frac{\pi}{3},
\end{equation}
mitigates overshooting but sacrifices the optimal quadratic speed-up. 
To overcome this trade-off, the fixed-point search algorithm was introduced to alleviate overshooting while retaining optimal query complexity~\cite{yoder2014fixed}.
The key point of the fixed-point algorithm is to utilize recursive quasi-Chebyshev polynomials~\cite{yoder2014fixed,li2024revisiting} for generating phase angles that guarantee monotonic convergence; the angles are given by
\begin{equation}
\begin{split}
    \alpha_{k}&=\beta_{\mathcal{N}-k+1}=-\cot^{-1}\left(\tan\left(\frac{2\pi k}{\mathcal{N}}\right)\sqrt{1-\frac{1}{\gamma^2}}\right),
\end{split}
\end{equation}
where $\gamma=T_{1/\mathcal{N}}(1/\delta)$ with the Chebyshev polynomials of the first kind $T_{d}(\cdot)$ of the degree $d$ for a desired precision of the final fidelity, that is, $F\ge1-\delta^2$.
This approach is especially practical, as it ensures reliable performance without requiring prior knowledge of the initial overlap, provided the overlap exceeds a certain threshold.
The aforementioned algorithms generally do not guarantee unit success probability, Refs.~\cite{long2001Grover,Roy2022Grover} present deterministic approaches to the solution, i.e., algorithms with phase angles that achieve zero error in the final state.
Since unstructured search is a fundamental computational task, the discovery of Grover’s algorithm has broadly impacted not only quantum computing, but also various applications such as optimization~\cite{durr1996quantum,gilliam2021grover,baritompa2005grover}, machine learning~\cite{dong2008quantum,du2021grover,muser2024provable} and cryptography~\cite{brassard1998quantum,bernstein2009cost}.

There have also been efforts to understand the power of Grover’s algorithm in terms of its geometric structure~\cite{miyake2001geometric,cafaro2012grover}.
Ref.~\cite{miyake2001geometric} considers the complex projective space to investigate the entanglement through the Grover's algorithm, while Ref.~\cite{cafaro2012grover} introduces Wigner-Yanase quantum information metric to recast the unstructured search task in the framework of information geometry.

Previous studies have also explored quantum search algorithms from the perspective of Hamiltonian simulation~\cite{nielsen2010quantum,grover2001schrodinger,farhi1998analog,lin2019application,carlini2006time}. 
For instance, it has been shown that real-time evolution $e^{it\tilde{H}}$ under the Hamiltonian $\tilde{H}=|\psi^{*}\rangle\langle\psi^{*}| + |\psi_0\rangle\langle\psi_0|$ can drive the initial state $\ket{\psi_0}$ toward the target solution state $\ket{\psi^{*}}$.
We emphasize that, while these results represent important early attempts to understand unstructured search in terms of Hamiltonian dynamics, they do not provide a clear explanation for the origin of Grover’s circuit structure or its connection to ITE, as established in our results.
In addition, classical analogs have been explored to further investigate the power of quantum search algorithms~\cite{grover2002classical}.

Beyond this, the conceptual framework of Grover’s algorithm has also been generalized to amplitude amplification~\cite{brassard2002quantum}, a fundamental technique that underpins many advanced quantum algorithms. 
This broader framework is further unified by the quantum singular value transformation formalism~\cite{gilyen2019quantum,martyn2021grand}. 
As such, the principles underlying Grover’s algorithm form a cornerstone of modern quantum algorithm design.
This indicates that understanding its mathematical structure plays a core role in both advancing theoretical insights and guiding the practical development of quantum algorithms.

\section{Preliminaries}
\subsection{Interplay between double-bracket flow, Riemannian gradient flow and ITE} \label{app:ITE_and_Riemannian_grad_flow}

We give an overview showing that the dynamics of imaginary-time evolution (ITE) corresponds to the steepest descent direction of a least-squares cost function defined on a Riemannian manifold.
To establish this, we proceed in two steps.
First, we demonstrate that the double-bracket flow (DBF)~\cite{bloch1985completely,bloch1990steepest,bloch1992completely,Brockett1991DBF,BLOCH1985103,moore1994numerical,BROCKETT1989761,deift1983ordinary,Chu_iterations,wegner1994flow,wegner2006flow,hastings2022lieb,GlazekWilson,GlazekWilson2,kehrein_flow,smith1993geometric,optimization2012,brockett2005smooth} characterizes the Riemannian gradient flow of a least-squares cost function.
Second, we show that ITE naturally takes the form of a DBF.

We begin with a brief overview of DBF.
In general, DBFs are matrix-valued ordinary differential equations that have been employed to perform matrix diagonalization, QR decomposition, eigenvalue sorting, and related tasks~\cite{optimization2012,Brockett1991DBF,deift1983ordinary,Chu_iterations}.
Brockett originally introduced this flow in the context of minimizing a least-squares cost between two matrices using steepest descent techniques.
Consider the set of matrices 
\begin{equation} \label{app_eq_manifold_for_DBF}
    \mathcal M(A) = \{U AU^\dagger\text{ s.t. } U^{-1}=U^\dagger\}
\end{equation}
generated by unitary evolution of a Hermitian operator $A$ and the least-squares cost $ C_{A,B}(U)=-\frac{1}{2}\|U A U^\dagger-B\|^2_\text{HS}$
with a Hermitian $B$ and the Hilbert-Schmidt norm $\|\cdot\|_\text{HS}$.
Let $X=UAU^{\dagger}$ be a point on the manifold $\mathcal M(A) $ and define the cost as
\begin{equation} \label{app_eq:cost_least}
    C_{B}(X)=-\frac{1}{2}\|X-B\|^2_\text{HS}.
\end{equation}
Then, the Riemannian gradient of the cost $C_{B}$ at $X$ is given by~\cite{optimization2012,riemannianflowPhysRevA.107.062421}
\begin{align} \label{eq: grad_riemannian}
\text{grad}_X C_B(X) = [[X,B],X].
\end{align}
We note that the matrix $B$ should be chosen properly depending on the task at hand.
For example, if $B$ is a diagonal matrix with entries $diag(1,2,3,\ldots)$, then the DBF in Eq.~\eqref{eq: grad_riemannian} drives the matrix $X$ towards a diagonal matrix whose entries are the eigenvalues of $X$ sorted in ascending order~\cite{Brockett1991DBF}.
By varying the diagonal entries of $B$, the order in which the eigenvalues appear in the diagonalized form of $X$ can be controlled. 
See Refs.~\cite{bloch1985completely,bloch1990steepest,bloch1992completely,Brockett1991DBF,BLOCH1985103,moore1994numerical,BROCKETT1989761,deift1983ordinary,Chu_iterations,wegner1994flow,wegner2006flow,hastings2022lieb,GlazekWilson,GlazekWilson2,kehrein_flow,smith1993geometric,optimization2012,brockett2005smooth} for other applications of DBF.


\begin{figure}[t]
\centering
\begin{tikzpicture}
\definecolor{lightgray}{HTML}{F4F4F4}
\definecolor{littlelightgray}{HTML}{ecececff}
\definecolor{pale}{HTML}{7ca3d4ff}
\definecolor{lightred}{HTML}{d8a2a2}
\node[anchor=center] (russell) at (-12,-2.3)
{\centering\includegraphics[width=0.3\textwidth]{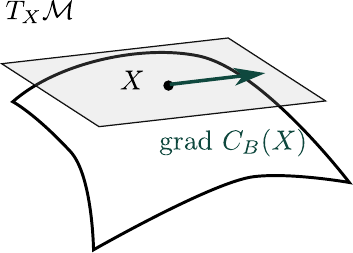}};
\node[text width=1cm] at (-15.0,0.5){(a)};
\node[anchor=center] (russell) at (-3.2,-2.3)
{\centering\includegraphics[width=0.60\textwidth]{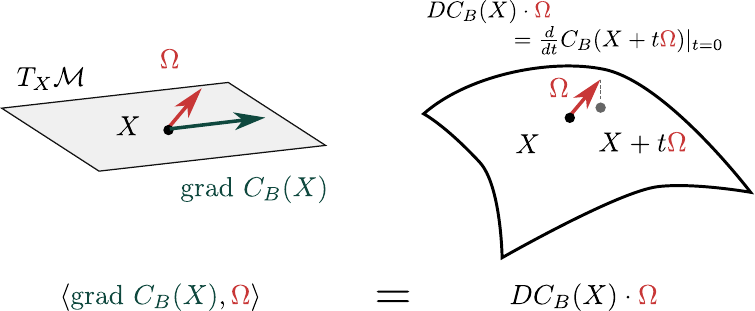}};
\node[text width=1cm] at (-8.5,0.5){(b)};\end{tikzpicture}
\caption{\textbf{Two conditions that characterize the Riemannian gradient.} The Riemannian gradient is uniquely characterized by two conditions: (a) the tangency condition and (b) the compatibility condition.
}
\label{fig:riemannian_conditions}
\end{figure}


Now, we elaborate on how Eq.~\eqref{eq: grad_riemannian} is derived.
A key point of the derivation is two properties that uniquely characterize the Riemannian gradient on manifolds: 
\begin{itemize}
    \item \textbf{Tangency condition}: the Riemannian gradient of a smooth function $C_{B}$ at point $X\in\mathcal{M}(A)$ lies in the tangent space $T_{X}\mathcal{M}(A)$ of the manifold $\mathcal{M}(A)$,
    $$\text{grad } C_{B}(X)\in T_{X}\mathcal{M}(A),\, \forall X \in \mathcal{M}(A).$$
    This property guarantees that the Riemannian gradient must be an element of the tangent space $T_{X}\mathcal{M}(A)$ of the manifold.
    
    \item \textbf{Compatibility condition}: given the differential $DC_{B}\in \cup _{X\in\mathcal{M}(A)}T_{X}^{*}\mathcal{M}(A)$ (i.e., the differential lives in a section of the cotangent bundle over $\mathcal{M}(A)$) defined as $DC_{B}(X): T_{X}\mathcal{M}(A) \to \mathbb{R}$ and the Riemannian metric $\langle\cdot,\cdot\rangle$, the inner product of the Riemannian gradient and any tangent vector $\Omega \in T_{X}\mathcal{M}(A)$ in the tangent space is consistent with the directional derivative of $C_{B}$ at $X$ denoted as $DC_{B}|_{X}(\Omega)=DC_{B}(X) \cdot \Omega$, 
    $$DC_{B}(X)\cdot \Omega=\langle\text{grad } C_{B}(X), \Omega\rangle, \, \forall\Omega\in T_{X}\mathcal{M}(A).$$
    Intuitively, this condition states that the rate of change of the cost function $C_{B}$ at the point $X$ in the direction of any tangent vector $\Omega\in T_{X}\mathcal{M}(A)$ (i.e., the directional derivative) is equal to how well the Riemannian gradient aligns with the given tangent vector $\Omega$ in terms of the defined Riemannian metric. 
\end{itemize}
See Fig.~\ref{fig:riemannian_conditions} for schematic views on the two conditions. 
In our setting, the tangent space of $\mathcal{M}(A)$ at $X\in\mathcal{M}(A)$ is given by
\begin{equation}
    T_{X}\mathcal{M}(A)=\{[X,\xi]|\xi^{\dagger}=-\xi\}.
\end{equation}
See Lemma 1.3 of Ref.~\cite{optimization2012} for the proof.
Then, since the differential is given by  $DC_{B}|_{X}([X,\Omega])=-\mathrm{Tr}[[X,B]^{\dagger}\Omega]$, the compatibility condition indicates that 
\begin{equation} \label{app_eq:comp_specific}
\begin{split}
    -\mathrm{Tr}[[X,B]^{\dagger}\Omega]&=\langle\text{grad } C(X), \Omega\rangle 
    = \langle [X,\xi],[X,\Omega] \rangle 
    = \mathrm{Tr}[\xi^{\dagger} \Omega],
\end{split}
\end{equation}
where we use the normal Riemannian metric defined as $\langle[X,\Omega_{1}],[X,\Omega_{2}] \rangle=\mathrm{Tr}[\Omega_{1}^{\dagger} \Omega_{2}]$ for $[X,\Omega_{1}],[X,\Omega_{2}]\in T_{X}\mathcal{M}(A) $.
As a result, Eq.~\eqref{app_eq:comp_specific} reveals that $ \xi = [B,X]$.
Thus, we arrive at Eq.~\eqref{eq: grad_riemannian}: see Refs.~\cite{optimization2012,riemannianflowPhysRevA.107.062421} for a detailed derivation.
This suggests that the right-hand side of Eq.~\eqref{eq: grad_riemannian} represents the steepest descent direction of the cost function on the Riemannian manifold $\mathcal M(A) $. 
The steepest descent flow is defined as the trajectory $A(t)\in \mathcal{M}(A)$ satisfying
\begin{align}
    \frac{\partial A(t)}{\partial t} = \text{grad}_{A(t)}C_B(A(t)),
\end{align}
which leads to the differential equation
\begin{align}\label{eq: BrockettODE}
   \frac{\partial A(t)}{\partial t} = [[A(t), B],A(t)]\ 
\end{align}
known as the DBF.

Next, we show that ITE takes the form of a DBF.
For clarity, we recall the definition of ITE.
Given a target Hamiltonian $\hat{H}$ and an initial state $\Psi(0)$, the ITE state can be written as
\begin{equation} \label{app_ep:ite_recap}
    \ket{\Psi(\tau)} = \frac{e^{-\tau \hat{H}}\ket{\Psi(0)}}{\|e^{-\tau \hat{H}}\ket{\Psi(0)}\|}
\end{equation}
for $\tau \in \mathbb{R}$.
Then, by taking the derivative of Eq.~\eqref{app_ep:ite_recap}, we have
\begin{equation} \label{eq:ITE_statevec}
\begin{split}
    \frac{\partial \ket{\Psi(\tau)}}{\partial \tau} 
    &= \frac{ -\hat{H} e^{-\tau \hat{H}} {| \Psi(0) \rangle} }{\|e^{-\tau \hat{H}}\ket{\Psi(0)}\|} + \frac{ e^{-\tau \hat{H}} {| \Psi(0) \rangle} \cdot \left(  - \frac{\partial}{\partial\tau}\|e^{-\tau \hat{H}}\ket{\Psi(0)}\|  \right ) }{\|e^{-\tau \hat{H}}\ket{\Psi(0)}\|^2} \\
    &= -\hat{H} |\Psi(\tau) \rangle + \frac{ \langle \Psi(0)| \hat{H} e^{-2\tau \hat{H}} |\Psi(0)\rangle }{\|e^{-\tau \hat{H}}\ket{\Psi(0)}\|^2 } |\Psi(\tau) \rangle  \\
    &= -\hat{H} |\Psi(\tau) \rangle + \langle \Psi(\tau) | \hat{H} | \Psi(\tau) \rangle |\Psi(\tau) \rangle\
    \\ &=   - (\hat{H}-E(\tau)I)\ket{\Psi(\tau)} \\
    &= - [\hat{H}, \Psi(\tau)]\ket{\Psi(\tau)}
\end{split}
\end{equation}
with $E(\tau)=\braket{\Psi(\tau)|\hat{H}|\Psi(\tau)}$ and $\|e^{-\tau \hat{H}}\ket{\Psi(0)}\|=\sqrt{\braket{\Psi(0)|e^{-2\tau \hat{H}}|\Psi(0)}}$.
Finally, the density matrix representation of Eq.~\eqref{eq:ITE_statevec} using $\Psi(\tau)=|\Psi(\tau)\rangle\langle\Psi(\tau)|$ reveals that
\begin{equation}
    \frac{\partial \Psi(\tau)}{\partial \tau} = [[\Psi(\tau), \hat{H}],\Psi(\tau)],
\end{equation}
which is the exact form of DBF, as shown in Eq.~\eqref{eq: BrockettODE}.
Namely, ITE can be interpreted as the DBF with $A=\Psi(\tau)$ and $B=\hat{H}$ after setting $\Psi(0)=\ket{\psi_0}\bra{\psi_0}$.
Since DBF represents the steepest descent direction of the cost in Eq.~\eqref{app_eq:cost_least}, ITE corresponds to the Riemannian gradient flow of the same least-squares cost.
We can verify this by rewriting the cost function as
\begin{equation} \label{app_eq:cost_ite_least}
   \begin{split}
        C_{\hat{H}}(\Psi(\tau)) &= -\frac{1}{2}\|\Psi(\tau)-\hat{H}\|_{\text{HS}}^2 \\
        &= \underbrace{-\frac{1}{2}\|\Psi(\tau)\|_{\text{HS}}^2 -\frac{1}{2}\|\hat{H}\|_{\text{HS}}^2}_{\text{const.}} + \underbrace{\mathrm{Tr}[\Psi(\tau)\hat{H}]}_{=E(\tau)}.
   \end{split}
\end{equation}
Thus, Eq.~\eqref{app_eq:cost_ite_least} shows that minimizing the least-squares cost is equivalent to minimizing the energy expectation value $E(\tau)$, which aligns with the objective of ITE, that is, ground-state preparation.

Lastly, we elaborate on the case for the unstructured search task.
As shown in Lemma~\ref{lem:solvability}, the ITE state for unstructured search is given by
\begin{equation} \label{app_eq:ITE_state_Grover}
    \ket{\Phi({\tau})}=\frac{e^{\tau \hat{H}_{f}}}{\sqrt{\braket{\psi_{0}|e^{2\tau \hat{H}_{f}}|\psi_{0}}}} \ket{\psi_0},
\end{equation}
where $\hat{H}_{f}=\sum_{x\in\{x|f(x)=1\}} |x\rangle\langle x|$.
In this case, using $\Phi({\tau})=\ket{\Phi({\tau})}\bra{\Phi({\tau})}$ and $\Phi(0)=\ket{\psi_{0}}\bra{\psi_{0}}$, its continuous differential equation form is expressed as
\begin{equation} \label{app_eq:ITE_diff_equation}
    \frac{\partial \Phi({\tau})}{\partial \tau} =\left[\left[\hat{H}_{f}, \Phi({\tau})\right], \Phi({\tau}) \right],
\end{equation}
which corresponds to the minimization of the cost
\begin{equation} \label{eq:cost_for_grover}
    C_{\hat{H}_{f}}(\Psi) = \frac{1}{2} \| \hat{H}_{f}-\Psi \|_\text{HS}^2.
\end{equation}
This is consistent with the goal of unstructured search that aims to identify the items satisfying $f(x)=1$; that is, to find the eigenvectors corresponding to the largest eigenvalues.

We note a subtle sign difference between the cost functions in Eq.~\eqref{app_eq:cost_ite_least} and Eq.~\eqref{eq:cost_for_grover}.
This discrepancy reflects the fact that ITE typically targets the ground state (the eigenstate with the smallest eigenvalue), whereas unstructured search seeks the largest-eigenvalue states.
Nevertheless, both objectives can be expressed within the same formalism.
By rewriting the cost function as
\begin{equation}
    \tilde{C}_{\hat{H}_{f}}(\Psi) = -\frac{1}{2}  \| (-\hat{H}_{f})-\Psi \|_\text{HS}^2,
\end{equation}
we effectively transform the task of maximization into a minimization problem similar to the ITE formulation for the ground-state preparation with $\hat{H}=-\hat{H}_{f}$.

\subsection{Review of product formula for commutators} \label{app:product_formula}

We here give an overview of known product formulae for exponentials of commutators, following Refs.~\cite{dawson2006solovay,commutator_approximation_2022,product_formula2013}.

First, we provide the second-order formula known as group commutator formula:
\begin{equation} \label{app_eq:gci_general}
    e^{ sA}e^{ sB}e^{- sA}e^{- sB}=e^{s^2[A,B]} + \mathcal{O}(s^3).
\end{equation}
Using the identity $[A,B]=[iB,iA]$, we can also obtain the following approximation
\begin{equation} 
    e^{ isB}e^{ isA}e^{- isB}e^{- isA}=e^{s^2[A,B]} + \mathcal{O}(s^3),
\end{equation}
whereas $[A,B]=[-iA,iB]$ leads to
\begin{equation} 
    e^{ -isA}e^{ isB}e^{ isA}e^{ -isB}=e^{s^2[A,B]} + \mathcal{O}(s^3),
\end{equation}

Next, we recall the third-order formula for the commutator:
\begin{equation}
    e^{\phi sA}e^{\phi sB}e^{- sA}e^{-(\phi+1) sB} e^{(1-\phi) sA}e^{sB}=e^{s^2[A,B]} + \mathcal{O}(s^4)
\end{equation}
with $\phi=(\sqrt{5}-1)/2$.

Moreover, we can improve the order in recursive manners.
Given a $n$-th order formula $f_{n}(x)=\exp(s^2[A,B]+\mathcal{O}(s^{n+1}))$, the following recursive formulae are known:
\begin{itemize}
    \item 2 copies formula (for even $n=2k$)
    \begin{equation}
    \begin{split}
        f_{n+1}(s)&=\exp(s^2[A,B]+\mathcal{O}(s^{n+2}))\\
        &=f_{n}(s/\sqrt{2})f_{n}(-s/\sqrt{2})\\
    \end{split}
    \end{equation}
    \item The Jean-Koseleff formula
    \begin{equation}
    \begin{split}
        f_{n+1}(s)&=\exp(s^2[A,B]+\mathcal{O}(s^{n+2}))\\
        &= f_{n}(ts)f_{n}(ws)f_{n}(ts) \text{ (if $n$ is even)} \\
        &= f_{n}(us)f_{n}(vs)^{-1}f_{n}(us) \text{ (if $n$ is odd)} \\
    \end{split}
    \end{equation}
    with $t=(2+2^{2/(n+1)})^{-1/2}$, $w=-2^{1/(n+1)}t$, $u=(2-2^{2/(n+1)})^{-1/2}$ and $v=2^{1/(n+1)}u$.\\
    \item 5-copies formula
    \begin{equation}
    \begin{split}
        f_{n+1}(s)&=\exp(s^2[A,B]+\mathcal{O}(s^{n+2}))\\
        &= f_{n}(\nu s)^2f_{n}(\mu s)^{-1}f_{n}(\nu s)^2
    \end{split}
    \end{equation}
    with $\mu=(4\sigma)^{1/2}$, $\nu=(1/4+\sigma)^{1/2}$ and $\sigma=4^{2/(n+1)}/(4(4-4^{2/(n+1)}))$.\\
\end{itemize}
See Ref.~\cite{kuperberg2023breaking} for an example of how alternative product formulae, originally developed in the context of pure mathematics~\cite{elkasapy2015length,elkasapy2016new}, have been applied. 
As a related aside, Ref.~\cite{peetz2025hamiltonian} provides an example of randomized compilation techniques employed in a similar context.
Finally, see Ref.~\cite{lai2025grover} for the discussion of product formulas in the context of retractions for Riemannian manifolds~\cite{absil2008optimization}.

By combining the techniques mentioned above, we can construct higher-order approximations of exponentials of commutators. These approximations generally take the form
\begin{equation}
    e^{it_1 A}e^{it_2 B}e^{it_3 A}e^{it_4 B}\ldots = e^{s^2[A,B]}  + \mathcal{O}(s^{m}) 
\end{equation}
with appropriately chosen coefficients $\{t_{i}\}$ to achieve an approximation of order $m$.

Lastly, we give an example to demonstrate the link between the approximation of ITE of Eq.~\eqref{eq:dbr} and the existing algorithms, especially the original Grover's algorithm.
As a simple approach, we utilize the group commutator formula in Eq.~\eqref{app_eq:gci_general} together with the fragmentation for the approximation;
\begin{equation}
    e^{s[\hat{H}_{f},\ket{\psi_{0}}\bra{\psi_{0}}]} \approx \left(e^{i\sqrt{2s/\mathcal{N}}\ket{\psi_{0}}\bra{\psi_{0}}}e^{i\sqrt{2s/\mathcal{N}}\hat{H}_{f}}e^{-i\sqrt{2s/\mathcal{N}}\ket{\psi_{0}}\bra{\psi_{0}}}e^{-i\sqrt{2s/\mathcal{N}}\hat{H}_{f}}\right)^{\mathcal{N}/2}.
\end{equation}
In the original Grover's algorithm~\cite{grover1996fast}, the phase angles are $\alpha_{k}=\beta_{k}=\pi$ for all $k$.
Observing that  $$e^{i\pi \hat{H}_{f}}=I+(e^{i\pi}-1)\hat{H}_{f}=I+(e^{-i\pi}-1)\hat{H}_{f}=e^{-i\pi \hat{H}_{f}},$$ $$e^{i\pi \ket{\psi_{0}}\bra{\psi_{0}}}=I+(e^{i\pi}-1)\ket{\psi_{0}}\bra{\psi_{0}}=I+(e^{-i\pi}-1)\ket{\psi_{0}}\bra{\psi_{0}}=e^{-i\pi \ket{\psi_{0}}\bra{\psi_{0}}},$$ we can choose $\sqrt{2s/\mathcal{N}}=\pi$ to match the structure of the original algorithm.
This observation implies that Grover’s algorithm can be interpreted as an approximation to the unitary generated by the commutator, i.e., $e^{s[\hat{H}_{f},\ket{\psi_{0}}\bra{\psi_{0}}]}$ with $s=\pi^2 \mathcal{N}/2$.

Given these large angles, one might expect that the product formula approximation might not work.
Nevertheless, the approximation technique works for some cases of unstructured search.
Ref.~\cite{double_bracket2024} provides a general bound for the exponential of commutators;
\begin{equation}\label{eq:GCI_compilation_}
\begin{split}
    &\left\|  e^{i\sqrt{s}\Psi}e^{i\sqrt{s}\hat{H}}    
    e^{-i\sqrt{s}\Psi}
    e^{-i\sqrt{s}\hat{H}}- e^{s[\hat{H},\Psi]} \right\|_{\text{op}} \leq   s^{3/2} \bigg( \|[\hat{H}, [\hat{H}, \Psi]]\|_{\text{op}} + \|[\Psi, [\Psi, \hat{H}]]  \|_{\text{op}}\bigg)
\end{split}
\end{equation}
where $s\ge0$, $\Psi$ is the density matrix and $\hat{H}$ denotes an arbitrary Hermitian matrix.
Here, $\|\cdot\|_{\text{op}}$ represents the operator norm.
Applying this to our case, $s=\pi^2$, $\Psi=|\psi_{0}\rangle\langle\psi_{0}|$ and $\hat{H}=\hat{H}_{f}$ given $\mathcal{N}=2$, we obtain the bound
\begin{align}
  \|e^{ i\pi \psi_0}e^{ i\pi \hat{H}_f }e^{ -i\pi \psi_0}e^{- i\pi \hat{H}_f }-e^{\pi^2[\hat{H}_f,\psi_0]}\| &\le 2\pi^3\|[\hat{H}_f,\psi_0]  \|\\
   &\le 4\pi^3\sqrt{V_0},
\end{align}
where $V_{0}=\braket{\psi_{0}|(\hat{H}_{f}-E_{0})^2|\psi_{0}}=E_{0}(1-E_{0})$ and $E_{0}=M/N$.
For further details of the calculation, see App.~\ref{app:optimal_complexity}.
This demonstrates that the variance $V_{0}$ fully governs the error of the product formula.
Thus, when the number of target items $M$ is small compared to the total number $N$ (i.e., $V_0 = \Theta(1/N)$), which is the case of primary interest, the error remains small and the product formula can be highly accurate.
On the other hand, when $M$ is comparable to $N$, i.e., $V_{0}\in \Theta(1)$, the error grows accordingly and hence the approximation fails.
This behavior suggests a potential for overshooting: as the state approaches the solution, the error introduced by the approximation increases, potentially preventing convergence to the solution state.

\subsection{Overview on Geometry of Manifolds} \label{app:geometry_review}

In this section, we give an overview of the relationships among the manifolds discussed in the main text: the set of matrices $\mathcal{M}(A)$ in Eq.~\eqref{app_eq_manifold_for_DBF} generated by unitary evolution of a Hermitian operator $A$, the special unitary group $SU(N)$, the complex projective spaces $\mathbb{C}\mathcal{P}$ and $\mathbb{C}\mathcal{P}^{N-1}$.
Our discussion presented below follows Refs.~\cite{optimization2012,bengtsson2017geometry,kobayashi1996foundations}. 

\medskip
\medskip

\textbf{Link between the set $\mathcal{M}(A)$ and the special unitary group $SU(N)$ -- }
We first clarify the relationship between the set $\mathcal{M}(A)$ and the special unitary group $SU(N)$.
Here, we follow the discussion on Ref.~\cite{optimization2012}.
Let us recall their definitions:
$$U(N)=\{U\in\mathbb{C}^{N\times N}|UU^{\dagger}=I\},$$
$$\mathcal{M}(A) = \{U AU^\dagger| U\in U(N)\},$$
$$SU(N)=\{U\in U(N)| \det{(U)}=1\}.$$
Since $e^{i\theta}UAU^{\dagger}e^{-i\theta}=UAU^{\dagger}$, we have the equality $\{U AU^\dagger| U\in U(N)\}=\{U AU^\dagger| U\in SU(N)\}$.
Then, $\mathcal{M}(A)$ is an \textit{orbit} space of the group action of $SU(N)$ on $\mathbb{C}^{N\times N}$.
In differential geometry, the orbit space of $\sigma:G\times \mathcal{M} \to \mathcal{M}$ with a compact Lie group $G$ and a smooth manifold $\mathcal{M}$, is generally defined as the set of all equivalence classes of $\mathcal{M}$; namely $\mathcal{M}/G \equiv \{O(x)|x\in\mathcal{M}\}$ with $O(x)=\{g\cdot x|g\in G\}$ the orbit of $x\in \mathcal{M}$.
Suppose the smooth Lie group action $\sigma: SU(N) \times \mathbb{C}^{N\times N} \to \mathbb{C}^{N\times N}$ defined as $\sigma_{H}(U)=UHU^{\dagger}, \, H\in\mathcal{M}(A)$.
Then, we can easily see that $\mathcal{M}(A)$ is an \textit{orbit} of the group action $\sigma$.
Roughly speaking, $\mathcal{M}(A)$ describes the quotient of the space by the group action of $SU(N)$.
We note that the link is used to show the set $\mathcal{M}(A)$ is a smooth and compact manifold.
$\mathcal{M}(A)$ is a \textit{homogeneous} space, meaning there exists a transitive group action: for any $H_{1},H_{2}\in \mathcal{M}(A)$, there exists $U\in SU(N)$ such that $H_2=UH_1U^{\dagger}$.
It is well known that orbits of a compact Lie group acting on a manifold are themselves smooth, compact submanifolds. 
Since $SU(N)$ is a compact Lie group, $\mathcal{M}(A)$ is indeed a smooth and compact manifold.

We also recall that the stabilizer subgroup $\text{Stab}(N) \subset SU(N)$ of $A\in \mathbb{C}^{N\times N}$ is defined as
$$\text{Stab}(N)=\{U\in SU(N)|UAU^{\dagger}=A\}.$$
This subgroup consists of elements that leave $A$ invariant under the action.
With the subgroup, the orbit $\mathcal{M}(A)$ is diffeomorphic to the homogeneous space, i.e.,
$$\mathcal{M}(A)\cong SU(N)/\text{Stab}(A).$$
Namely, there exists a continuous differentiable map $f:\mathcal{M}(A)\to SU(N)/\text{Stab}(A)$ that is a bijection and its inverse is also differentiable.

Furthermore, we discuss the connection between these manifolds in terms of their tangent spaces.
Since the map $\sigma$ is a submersion (see e.g., App.~C.6 of Ref.~\cite{optimization2012} for its definition), its differential at the identity induces a surjective linear map on tangent spaces. 
Note that the tangent space of $SU(N)$ at the identity $I$ is the Lie algebra: $\mathfrak{su}(N)=T_{I}SU(N)=\{\Omega\in\mathbb{C}^{N\times N}|\Omega^{\dagger}=-\Omega,\ \mathrm{Tr}[\Omega]=0\}$.
Then, the derivative of $\sigma_{H}(U)=UHU^{\dagger}$ at $I$ is the surjective linear map $D\sigma_{H}|_{I}: T_{I}SU(N) \to T_{H}\mathcal{M}(A)$ for $H\in \mathcal{M}$ defined as $$D\sigma_{H}|_{I}(\Omega)= \Omega H - H\Omega  = [\Omega,H].$$
This suggests that elements in the tangent space of the set $\mathcal{M}(A)$ at $H\in \mathcal{M}(A)$ can be obtained from those in the tangent space of $SU(N)$ at the identity.
Using the expression together with the tangency and compatibility conditions, the Riemannian gradient flow in App.~\ref{app:ITE_and_Riemannian_grad_flow} is derived~\cite{optimization2012}.
Similarly, Ref.~\cite{riemannianflowPhysRevA.107.062421} derives the Riemannian gradient flow on $SU(d)$ in the context of the ground-state preparation.
Their result coincides with the formulation given in Ref.~\cite{gluza_DB_QITE_2024}, where the manifold $\mathcal{M}(A)$ is explicitly considered for the same ground-state preparation task.

\medskip
\medskip
\textbf{Relation of set $\mathcal{M}(A)$ and the special unitary group $SU(N)$ to the complex projective space $\mathbb{C}P^{N-1}$ -- }
We now argue how the manifolds discussed above relate to the complex projective space.
Let $\mathbb{C}^{N}$ be the complex vector space of dimension $N$.
The complex projective space $\mathbb{C}P^{N-1}$ is defined as the set of equivalence class of non-zero vectors in $\mathbb{C}^{N}$, where two vectors are considered equivalent if these vectors differ by a non-zero complex scalar factor.
Formally, $\mathbb{C}P^{N-1}=(\mathbb{C}^{N}\setminus\{0\})/\sim$, where $$y\sim x \Longleftrightarrow y = \lambda x$$ for some $\lambda \in \mathbb{C}\setminus\{0\} $.
The manifold naturally arises in the context of quantum mechanics, as pure quantum states are defined only up to a global phase. 
That is, the state vectors $\ket{\psi}$ and $e^{i\theta}\ket{\psi}$ are physically indistinguishable for $\theta\in\mathbb{R}$.
Consequently, complex projective space plays a central role in analyzing the geometric structures in quantum mechanics~\cite{bengtsson2017geometry}.

The complex projective space admits the diffeomorphisms~\cite{kobayashi1996foundations}, i.e., $\mathbb{C}P^{N-1} \cong U(N)/(U(1)\times U(N-1)) \cong SU(N)/S(U(1) \times U(N-1))$, with the subgroup $S(U(1)\times U(N-1))=\{(e^{i\theta},V)\in U(1)\times U(N-1)|e^{i\theta} \det(V)=1 \}$. 
This establishes a direct geometric link between complex projective space and quotient spaces of both the unitary and special unitary groups.
Moreover, when the Hermitian operator $A$ is the density matrix of a pure state (a rank-one Hermitian projector with unit trace) as is the case in ground-state preparation, the set $\mathcal{M}(A)$ is diffeomorphic to $\mathbb{C}P^{N-1}$ since $\text{Stab}(A)=S(U(1)\times U(N-1))$.

\medskip
\medskip

\textbf{Geodesics on $\mathbb{C}P^{N-1}$ and its link to $\mathbb{C}P^{1}$ -- }
Lastly, we mention the connection between $\mathbb{C}P^{N-1}$ and $\mathbb{C}P^{1}$ to ensure that ITE follows a geodesic in $\mathbb{C}P^{N-1}$, as proven in App.~\ref{app:optimal_complexity}.
Equipped with the Fubini-Study metric, the distance on the complex projective space for two states $\ket{\psi}$, $\ket{\phi}$ in complex projective space $\mathbb{C}P^{N-1}$ is given by
\begin{equation} \label{app_eq:fs_dist}
    d_{\mathrm{FS}}(\ket{\psi},\ket{\phi}) = \arccos(|\braket{\psi|\phi}|).
\end{equation}
In this case, it is known that the geodesic in $\mathbb{C}P^{N-1}$ between two points lies within a submanifold isomorphic to $\mathbb{C}P^{1}$ known as a complex projective line~\cite{bengtsson2017geometry}.
Intuitive understanding of this result is as follows.
From Eq.~\eqref{app_eq:fs_dist}, we observe that the Fubini–Study distance between two pure states depends only on angle of their overlap; that is, any states lying outside two-dimensional complex subspace spanned by them do not affect the overlap and only increases the path length.
Since a geodesic minimizes the distance locally, it must lie entirely within the span of the two states.
Fig.~\ref{fig:geodesic_ex} illustrates the geodesic in the case of a single-qubit system, i.e., the geodesic on the Bloch sphere (the complex projective space $\mathbb{C}P^{1}$) as an example.
This $\mathbb{C}P^{1}$ plays a role analogous to a ``great circle" on a sphere, capturing the minimal-length path in the curved geometry of projective space.
As shown later, the ITE dynamics for unctructured search represents the unitary evolution within a two-dimensional subspace of $\mathbb{C}^{N}$, which projects down to a $\mathbb{C}P^{1}$ in $\mathbb{C}P^{N-1}$.
Since the target state lies along this path, the evolution indeed follows a geodesic in the complex projective space.

\begin{figure}[t]
\centering
\begin{tikzpicture}
\node[anchor=center] (russell) at (0.,0.0){\centering\includegraphics[width=0.25\textwidth]{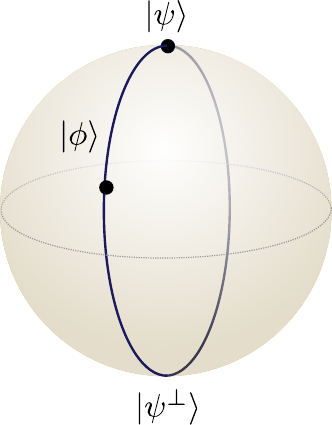}};
\end{tikzpicture}
\caption{\textbf{Schematic view of a geodesic on the Bloch sphere.}
Given two pure states \( \ket{\phi}, \ket{\psi} \in \mathbb{C}^2 \), the geodesic in the complex projective space equipped with the Fubini–Study metric traces a great circle on the Bloch sphere.
Here, \( \ket{\psi^\perp} \) denotes a state orthogonal to \( \ket{\psi} \), i.e., \( \braket{\psi | \psi^\perp} = 0 \).
More generally, in \( \mathbb{C}P^{N-1} \), the geodesic always lies within a two-dimensional subspace spanned by two quantum states.
}
\label{fig:geodesic_ex}
\end{figure}

\subsection{Overview of Quantum Signal Processing} \label{app:qsp_overview}
Here, we elaborate on quantum signal processing (QSP)~\cite{equiangular}. 
In the original work, a quantum circuit $U_{QSP}$ was introduced with a sequential structure comprising of two types of operators:
\textit{signal operators} ${W}$ and \textit{signal processing operators} ${S}(\phi)$, where the phase $\phi$ is drawn from a set ${\phi_k}$.
The desired polynomial transformation is then obtained by performing a measurement in the so-called \textit{signal basis}.
Consider a degree-$K$ polynomial of a scalar input $x\in[-1,1]$.
Then, it was demonstrated that there exists a sequence of QSP phases $\{\phi_{k}\}$ such that the following circuit
\begin{equation} \label{app eq:qsp_conv_uni}
U_{QSP} = S_{z}(\phi_{K}) \prod_{k=0}^{K-1} W(x)S_{z}(\phi_{k})
\end{equation}
with the operators \vspace{-0.15cm}
\begin{equation*}
    W(x) = e^{ixX/2} = \begin{pmatrix}
        x & i\sqrt{1-x^2} \\
        i\sqrt{1-x^2} & x \\
\end{pmatrix},
\end{equation*}
\begin{equation}
    S_{z}(\phi) = e^{i\phi Z} = \begin{pmatrix}
        e^{i\phi} & 0 \\
        0 & e^{-i\phi} \\
\end{pmatrix},
\end{equation}
followed by measurement in the computational basis $M=\{\ket{0},\ket{1}\}$ can realize any polynomial $p(x) \in\mathbb{C}[x]$ of degree-$K$ , provided that
\begin{enumerate}
\item Degree of $p(x)$ is equal to or less than $K$,
\item $p(x)$ has a parity $K$ mod 2,
\item $\forall x\in[-1,1]$, $|p(x)|\le1$,
\item $\forall x\in(-\infty,-1] \cup [1,\infty)$, $|p(x)|\ge1$,
\item if $K$ is even, then $\forall x\in\mathbb{R}$, $|p(ix)p^*(ix)|\ge1$.
\end{enumerate}
We refer to the setting as ($R(x)$, $S_{Z}$, $\ket{0}$)-QSP.
When the target polynomial is real, choosing the basis $M=\{\ket{+},\ket{-}\}$ allows the function to be implemented by considering the first three of the above five conditions only.
The ($R(x)$, $S_{Z}$, $\ket{+}$)-QSP convention is often preferred in practice, as it allows for a broader class of achievable functions compared to the ($R(x)$, $S_{Z}$, $\ket{0}$)-QSP setting.

Equivalently, we can use the notation
\begin{equation}
     R(x)=\begin{pmatrix}x & \sqrt{1-x^2}  \\  \sqrt{1-x^2} & -x\end{pmatrix}
\end{equation}
to have the equality
\begin{equation}
    S_{z}(\phi_{K}) \prod_{k=0}^{K-1} W(x)S_{z}(\phi_k) = S_{z}(\phi'_{K}) \prod_{k=0}^{K-1} R(x)S_{z}(\phi'_k)
\end{equation}
with the relationship $R(x)=-ie^{i\pi Z/4}W(x)e^{i\pi Z/4}$ and hence $\phi'_{K}=\phi_{K}+(2K-1)\pi/4$, $\phi'_{0}=\phi_{0}-\pi/4$ and $\phi'_{k}=\phi_{k}-\pi/2$ for $k=1,\ldots,K-1$.
See Ref.~\cite{martyn2021grand,low2017quantum,low2017optimal,equiangular,motlagh2024generalized} for more details.

\section{Proofs in the Main Texts and Relevant Findings}
\subsection{Proof of Lemma~\ref{lem:solvability}} \label{app:proof_of_solvability}

For completeness, we provide the statement from the main text again.
\begin{lemmaA}[ITE solves the unstructured search problem] \label{app_lem:solvability}
     Given the projector Hamiltonian $\hat{H}_f$ in Eq.~\eqref{app_eq:projector} and the initial state in Eq.~\eqref{app_eq:init}, the ITE state 
     \begin{equation} \label{app_eq:ITE_state_for_Grover}
         \ket{\psi(\tau)}=\frac{e^{\tau \hat{H}_{f}}\ket{\psi_0}}{\|e^{\tau \hat{H}_{f}}\ket{\psi_0}\|_2}
     \end{equation}
     with $\|e^{\tau \hat{H}_{f}}\ket{\psi_0}\|_2=\sqrt{\braket{\psi_{0}|e^{2\tau H_{f}}|\psi_{0}}}$ converges to the solution state in Eq.~\eqref{app_eq:solution} as $\tau\to\infty$, i.e., 
    \begin{equation} \label{app_eq:grover_solution}
         \lim_{\tau\to\infty} \frac{e^{\tau \hat{H}_{f}}\ket{\psi_0}}{\|e^{\tau \hat{H}_{f}}\ket{\psi_0}\|_2}=\ket{\psi^{*}}.
\end{equation}
\end{lemmaA}

\begin{proof}
    By taking the limit $\tau \to \infty$ of Eq.~\eqref{app_eq:ITE_state_for_Grover}, we have
\begin{equation}
\begin{split}
    \lim_{\tau\to\infty} \frac{e^{\tau \hat{H}_{f}}}{\sqrt{\braket{\psi_{0}|e^{2\tau \hat{H}_{f}}|\psi_{0}}}} \ket{\psi_0} &= \lim_{\tau\to\infty} \frac{I+(e^{\tau}-1)\hat{H}_f}{\sqrt{\braket{\psi_{0}|I+(e^{2\tau}-1)\hat{H}_f|\psi_{0}}}} \ket{\psi_0} \\
    &= \lim_{\tau\to\infty} \frac{I/e^{\tau}+(1-1/e^{\tau})\hat{H}_f}{\sqrt{\braket{\psi_{0}|I/e^{2\tau}+(1-1/e^{2\tau})\hat{H}_f|\psi_{0}}}} \ket{\psi_0} \\
    &= \frac{\hat{H}_f\ket{\psi_{0}}}{\sqrt{\braket{\psi_{0}|\hat{H}_f|\psi_{0}}}} \\
    &= \frac{1}{\sqrt{M}}\sum_{x\in\{x|f(x)=1\}} \ket{x} = \ket{\psi^{*}},
\end{split} 
\end{equation}
where we use the equality
\begin{equation} \label{eq:poject_ite_equivalence}
    e^{\tau \hat{H}_f} = I+(e^{\tau}-1)\hat{H}_f.
\end{equation}
in the first equality.
This identity holds because $\hat{H}_f$ is a projector, i.e., $\hat{H}_f^2=\hat{H}_f$.
\end{proof}
Note that the convergence to the solution state $\ket{\psi^{*}}$ is guaranteed when the initial state includes all the answer states with a uniform probability; otherwise, some solutions excluded from the initial state can not be obtained at the final step; for example, if the initial state is $\ket{0}$ given $f(0)=1$, the resultant state is different from the solution state, i.e., $\lim_{\tau\to\infty}\ket{\psi(\tau)}=\ket{0}$.

\subsection{Proof of Lemma~\ref{lem:equiv_DBR_ITE}} \label{app:DBR_ITE_relation}

For clarity, we restate the result from the main text.
\begin{lemmaA}[[ITE can be realized by its first-order approximation]
\label{app_lem:equiv_DBR_ITE} 
Let $\hat{H}_{f}$ be the projector Hamiltonian in Eq.~\eqref{app_eq:projector}.
Then, for any ITE evolution time $\tau$, there exists a time duration $s_{\tau}$ such that 
\begin{equation} \label{app_eq:ite_dbr_equiv}
    \frac{e^{\tau \hat{H}_f}\ket{\psi_0}}{\|e^{\tau \hat{H}_f}\ket{\psi_0}\|} =e^{s_\tau[ \hat{H}_f,\psi_{0}]}\ket{\psi_0}
\end{equation}
for any $\tau$.
\end{lemmaA}

\begin{proof}
Our proof fundamentally relies on Lemma~B.3 in Ref.~\cite{suzuki2025double}, which we restate below for completeness:
\begin{lemmaA}[Lemma~B.3 in Ref.~\cite{suzuki2025double} ]
\label{prop any linear polynomial synthesis}

    Let $x,y\in \mathbb R$ and  $(x,y)\neq (0,0)$.
    Define the parameter
    \begin{align} \label{eq: app_duration_general}
        s = -\frac{\mathrm{sgn}(y)}{\sqrt{V_{\Psi}}} \mathrm{arccos} \left(\frac{x+yE_{\Psi} }{\|(xI+y H)\ket{\Psi}\|} \right),
    \end{align} 
    with $E_{\Psi}=\braket{\Psi|\hat{H}|\Psi}$ and $V_{\Psi}=\braket{\Psi|(\hat{H}-E_{\Psi})^2|\Psi}=\braket{\Psi|\hat{H}^2|\Psi}-E_{\Psi}^2$, given a state vector $\ket{\Psi}$ and an arbitrary Hermitian matrix $\hat{H}$.
    Then, 
    \begin{align}
       \frac{(xI+y \hat{H})\ket{\Psi} }{\|(xI+y \hat{H})\ket{\Psi}\|} = (a(s)I + b(s)H)\ket \Psi=e^{s[\ket{\Psi}\bra{\Psi}, \hat{H}]} \ket \Psi\ ,
    \end{align}
  where $a(s),b(s)$ are real-valued coefficients given by
\begin{align}
    a(s) &= \frac{ E_\Psi }{\sqrt{V_{\Psi}}}\sin\left(s \sqrt{V_{\Psi}}\right) + \cos\left(s \sqrt{V_{\Psi}}\right), \label{a(s) main}
\\
        b(s) &= - \frac{1}{\sqrt{V_{\Psi}}} \sin\left(s \sqrt{V_{\Psi}}\right)\ .\label{b(s) main}
    \end{align}

\end{lemmaA}
Considering a specific case of Lemma~\ref{prop any linear polynomial synthesis}, we have
\begin{equation} \label{eq:equivalence_DBR_linear_poly}
    \frac{(I+c\hat{H})\ket{\Psi}}{\|(I+c\hat{H})\ket{\Psi}\|} = e^{s[ \hat{H}, |\Psi\rangle\langle\Psi|]}\ket{\Psi}\ 
\end{equation}
when 
\begin{equation}
    s = \frac{1}{\sqrt{V_{\Psi}}}\mathrm{arccos} \left(\frac{1+cE_{\Psi} }{\|(I+c H)\ket{\Psi}\|} \right)
\end{equation}
for $c>0$.
In our case, we consider the projector $\hat{H}=\hat{H}_{f}$ and the initial state $\ket{\Psi}=\ket{\psi_0}$ defined in Eq.~\eqref{app_eq:init}.
Using the equality shown in Eq.~\eqref{eq:poject_ite_equivalence}, the ITE state in Eq.~\eqref{app_eq:ITE_state_for_Grover} can be simplified to
\begin{equation} 
    \frac{e^{\tau \hat{H}_{f}}\ket{\psi_0}}{\|e^{\tau \hat{H}_{f}}\ket{\psi_0}\|} =  \frac{(I+(e^{\tau}-1)\hat{H}_{f})\ket{\psi_0}}{\|(I+(e^{\tau}-1)\hat{H}_{f})\ket{\psi_0}\|} \equiv \frac{(I+c\hat{H})\ket{\Psi}}{\|(I+c\hat{H})\ket{\Psi}\|}
\end{equation}
with $c=e^{\tau}-1$, suggesting that the existence of the time duration $s_\tau$ such that the ITE state for unstructured search can be realized by the exponential of commutators in Eq.~\eqref{eq:equivalence_DBR_linear_poly}, since $e^{\tau}-1>0$ for any $\tau>0$.
More concretely, using the time duration
\begin{equation} \label{eq:duration_ite}
     s_{\tau} = \frac{1}{\sqrt{E_{0}(1-E_{0})}}\mathrm{arccos} \left(\frac{1+(e^{\tau}-1)E_{0} }{\sqrt{1+(e^{2\tau}-1)E_{0}}} \right)
\end{equation}
with $E_{0}=\braket{\psi_{0}|\hat{H}|\psi_0}$ and $V_{0}=\braket{\psi_{0}|\hat{H}^2|\psi_0}-E_{0}=E_{0}(1-E_{0})$, the equality in Eq.~\eqref{app_eq:ite_dbr_equiv} holds for any $\tau$.
\end{proof}

Additionally, we investigate how the imaginary time $\tau$ relates to the corresponding time duration $s_{\tau}$, to gain insights into the trajectories of the ITE and its first-order approximation.
By computing the derivative of Eq.~\eqref{eq:duration_ite}, we have
\begin{equation}
\begin{split}
    \frac{d s_{\tau}}{d \tau} &=  \underbrace{\frac{1}{\sqrt{E_{0}(1-E_{0})}}}_a \cdot \underbrace{\left(\frac{-1}{\sqrt{1-\left(\frac{1+(e^{\tau}-1)E_{0} }{\sqrt{1+(e^{2\tau}-1)E_{0}}} \right)^2}}\right)}_b \cdot \underbrace{\frac{d}{d \tau}\left(\frac{1+(e^{\tau}-1)E_{0} }{\sqrt{1+(e^{2\tau}-1)E_{0}}} \right)}_{c},\\
    &=abc,
\end{split}
\end{equation}
where we use the equality $d(\arccos(x))/dx=-1/\sqrt{1-x^2}$.
Since $E_{0}$ ranges from 0 to 1, $ab$ is always negative for $\tau>0$.
As for the term $c$, we can obtain
\begin{equation}
\begin{split}
    \frac{d}{d \tau}\left(\frac{1+(e^{\tau}-1)E_{0} }{\sqrt{1+(e^{2\tau}-1)E_{0}}} \right) &= \frac{e^{\tau}E_{0}\left(\sqrt{1+(e^{2\tau}-1)E_{0}}\right)-\left(1+(e^{\tau}-1)E_{0}\right)\frac{e^{2\tau}E_{0}}{\sqrt{1+(e^{2\tau}-1)E_{0}}}}{1+(e^{2\tau}-1)E_{0}}  \\
    &= \frac{E_{0}e^{\tau}\left(1+(e^{2\tau}-1)E_{0}\right)-E_{0}e^{2\tau}\left(1+(e^{\tau}-1)E_{0}\right)}{\left(1+(e^{2\tau}-1)E_{0}\right)^{3/2}} \\
    &= \frac{-E_{0}e^{\tau}(1-E_{0})(e^{\tau}-1)}{\left(1+(e^{2\tau}-1)E_{0}\right)^{3/2}} \le 0.
\end{split}
\end{equation}
Consequently, combining the above results, the derivative of the time duration is non-negative, i.e., $ds_{\tau}/d\tau\ge0$ for all $\tau \ge 0$.
Note that $s_{\tau}=0$ when $\tau=0$.
Thus, $s_{\tau}$ is increasing from $0$ and reaches a plateau for a large value of $\tau$ (as $\lim_{\tau\to\infty}ds_{\tau}/d\tau=0$).
This suggests that the time duration $s_{\tau}$ increases with $\tau$, but eventually saturates, reflecting a deviation between the ITE and its first-order approximation; the latter overshoots the ITE trajectory.

\subsection{Geometric properties of Grover's algorithm} \label{app:geodesic_complexity}

We here exploit the ITE formulation for unstructured search to examine its geometric structure by considering the complex projective space $\mathbb{C}P^{N-1}$ as a manifold.
This space is the set of equivalence classes of non-zero vectors in $\mathbb{C}^{N}$, where $\ket{\psi}\sim \lambda \ket{\psi}$ for $\lambda \in \mathbb{C}\setminus\{0\} $.
It naturally arises in quantum mechanics, since pure states are defined only up to a global phase~\cite{bengtsson2017geometry}.
Equipped with the Fubini–Study metric~\cite{anandan1990geometry,bengtsson2017geometry,mukunda1993quantum}, the distance between two normalized states $\ket{\psi}$ and $\ket{\phi}$ on the manifold is defined as
\begin{equation} \label{eq:geodesic_length_FS}
    d_{\mathrm{FS}}(\ket{\psi},\ket{\phi}) = \arccos(|\braket{\psi|\phi}|).
\end{equation}
The geodesic, which locally minimizes this distance, lies in the two-dimensional subspace spanned by the two states.
For instance, given orthonormal states $\ket{\phi_1}$ and $\ket{\phi_2}$, the geodesic is parameterized as $\ket{\phi(t)}=\cos(t)\ket{\phi_1} + \sin(t)\ket{\phi_2} $.

With this geometric foundation, we show that the ITE trajectory follows the geodesic.
\begin{theoremA}[ITE traces the geodesic] \label{app_thm:geodesic_DBR}
    Let $\hat{H}_{f}$ be the projector in Eq.~\eqref{app_eq:projector} and $\ket{\psi_{0}}$ be the initial state defined in Eq.~\eqref{app_eq:init}.
    Define the orthonormal state
    \begin{equation} \label{app_eq:perp_init_state}
        \ket{\psi_{0}^{\perp}}= \frac{\hat{H}_{f}-E_{0}I}{\sqrt{E_{0}(1-E_{0})}} \ket{\psi_{0}}
    \end{equation}
    such that $\braket{\psi_{0}|\psi_{0}^{\perp}}=0$, where $E_{0}=\braket{\psi_{0}|\hat{H}_{f}|\psi_{0}}=M/N$.
    Then, the ITE state in Eq.~\eqref{eq:dbr} for a time duration $s$ is given by
    \begin{equation}\label{app_eq: psi s}
        \ket{\psi_{s}} = \cos(s\sqrt{V_{0}}) \ket{\psi_{0}} + \sin(s\sqrt{V_{0}}) \ket{\psi_{0}^{\perp}},
    \end{equation}
    where $E_0 = \braket{\psi_0|\hat{H}_f|\psi_0} = M/N$ and $V_{0}=\braket{\psi_{0}|(\hat{H}_{f}-E_{0})^2|\psi_{0}}=E_{0}(1-E_{0})$. 
    Since Eq.~\eqref{app_eq: psi s} can realize the solution state $\ket{\psi^{*}}$ when 
    \begin{equation} \label{app_eq:opt_s}
        s^{*}=\arccos(\sqrt{E_{0}})/\sqrt{V_{0}},
    \end{equation}
    ITE follows the trajectory of the geodesic on $\mathbb{C}\mathcal{P}^{N-1}$.
\end{theoremA}

\begin{proof}
    We first show that the ITE dynamics can be expressed as Eq.~\eqref{app_eq: psi s}.
    We here utilize the proof shown in Lemma~1 of Ref.~\cite{suzuki2025double}.
    The exponential of the commutator $\hat{W}_{\hat{H}}\equiv[\hat{H}, |\Psi\rangle\langle\Psi|]$ for an arbitrary Hermitian matrix $\hat{H}$ and a state vector $\ket{\Psi}$ can be expressed as
    \begin{equation}
        e^{s\hat{W}_{\hat{H}}} = \sum_{k=0}^\infty\frac{s^k}{k!} \hat{W}_{\hat{H}}^k.
    \end{equation}
    Since $$\hat{W}_{\hat{H}} \ket{\Psi} =\hat{H}\ket{\Psi} - E_{\Psi}\ket{\Psi}$$
and $$\hat{W}_{\hat{H}}^2 \ket{\Psi} = \hat{W}_{\hat{H}}\hat{H}\ket{\Psi} - E_\Psi\hat{W}_{\hat{H}} \ket{\Psi}= E_{\Psi}\hat{H}\ket{\Psi} - \bra\Psi \hat{H}^2\ket\Psi\ket\Psi - E_{\Psi}\hat{H}\ket{\Psi} + E_{\Psi}^2 \ket{\Psi} 
    = - V_{\Psi}\ket{\Psi} $$ with $E_{\Psi}=\braket{\Psi|\hat{H}|\Psi}$ and $V_{\Psi}=\braket{\Psi|(\hat{H}-E_{\Psi})^2|\Psi}=\braket{\Psi|\hat{H}^2|\Psi}-E_{\Psi}^2$, any even power of the  commutator $\hat{W}_{\hat{H}}$ acting on the state $\ket\Psi$ gives
\begin{align}
\hat{W}_{\hat{H}}^{2k}\ket{\Psi} = (-V_{\Psi})^k\ket{\Psi}\ .
\end{align}
Similarly, we have $\hat{W}_{\hat{H}}^{2k+1}\ket{\Psi} = (-V_{\Psi})^k \hat{W}_{\hat{H}}\ket{\Psi}$ for any odd power.
Therefore, separating the odd and even terms leads to a weighted sum of  $\ket{\Psi}$ and $\hat{W}_{\hat{H}}\ket{\Psi}$ with coefficients expressed by sine and cosine functions as
\begin{equation} \label{eq:general_dbr_expression}
    e^{s\hat{W}_{\hat{H}}}\ket{\Psi} 
     = \cos\left(s\sqrt{V_{\Psi}}\right)\ket{\Psi} + \sin\left(s\sqrt{V_{\Psi}}\right)\frac{\hat{W}_{\hat{H}}}{\sqrt{V_{\Psi}}}\ket{\Psi} \ .
\end{equation}
In our case, we consider $\hat{H}=\hat{H}_{f}$ and $\ket{\Psi}=\ket{\psi_{0}}$ with $E_{0}=\braket{\psi_{0}|\hat{H}_{f}|\psi_{0}}$ and $V_{0}=\braket{\psi_{0}|(\hat{H}_{f}-E_{0})^2|\psi_{0}}=E_{0}(1-E_{0})$.
Thus Eq.~\eqref{eq:general_dbr_expression} is re-expressed as
\begin{equation}
\begin{split}\label{eq: app_eq: psi s derivation}
    e^{s[\hat{H}_{f},\psi_{0}]}\ket{\psi_0} &= \cos\left(s\sqrt{V_{0}}\right)\ket{\psi_{0}} + \sin\left(s\sqrt{V_{0}}\right)\frac{[\hat{H}_{f}, |\psi_{0}\rangle\langle\psi_{0}|]}{\sqrt{V_{0}}}\ket{\psi_{0}} \\
    &= \cos\left(s\sqrt{V_{0}}\right)\ket{\psi_{0}} + \sin\left(s\sqrt{V_{0}}\right)\ket{\psi_{0}^{\perp}}.
\end{split}
\end{equation}

Next, we verify that there exists a time duration $s$ such that the ITE state $\ket{\psi_{s}}$ results in the solution state of Eq.~\eqref{app_eq:solution}.
By expressing Eq.~\eqref{eq: app_eq: psi s derivation} solely in terms of $\ket{\psi_{0}}$ and substituting $V_{0}=E_{0}(1-E_{0})$, we have
\begin{equation} \label{app_eq:full_desc_1st_app_ITE}
    \ket{\psi_{s}} = \left(\left(-\frac{E_{0}}{\sqrt{E_{0}(1-E_{0})}}\sin\left(s\sqrt{V_{0}}\right)+\cos\left(s\sqrt{V_{0}}\right)\right)I+ \sin\left(s\sqrt{V_{0}}\right)\frac{\hat{H}_{f}}{\sqrt{E_{0}(1-E_{0})}}  \right) \ket{\psi_{0}}.
\end{equation}
Note that the solution state can be written as $\ket{\psi^{*}}=\hat{H}_f\ket{\psi_{0}}/\sqrt{E_{0}}$.
Thus, by solving
\begin{align}
     -\frac{E_{0}}{\sqrt{E_{0}(1-E_{0})}}\sin\left(s\sqrt{V_{0}}\right)+\cos\left(s\sqrt{V_{0}}\right) &= 0, \label{app_eq:eq_coef_id} \\
     \frac{\sin\left(s\sqrt{V_{0}}\right)}{\sqrt{E_{0}(1-E_{0})}} &= \frac{1}{\sqrt{E_{0}}}, \label{app_eq:eq_coef_H}
\end{align}
i.e., computing \eqref{app_eq:eq_coef_id} $+$ \eqref{app_eq:eq_coef_H} $\times$ $E_{0}$, we obtain $s^{*}=\arccos(\sqrt{E_{0}})/\sqrt{V_{0}}$.
These results suggest that ITE dynamics in Eq.~\eqref{eq:dbr} describes a great circle connecting the initial and the solution states.
\end{proof}

Note that ITE for a general Hamiltonian follows the steepest descent direction with respect to a least-squares cost function~\cite{gluza_DB_QITE_2024,mcmahon2025equating}, whereas geodesics are defined independently of any cost functions.
Nevertheless, Theorem~\ref{app_thm:geodesic_DBR} shows that, when the Hamiltonian is a projector, the ITE trajectory coincides with the geodesic on $\mathbb{C}P^{N-1}$.

This observation sharpens the two-dimensional picture used in earlier analyses of Grover’s algorithm. 
Studies such as Refs.~\cite{nielsen2010quantum,yoder2014fixed,li2024revisiting} consider the span of $\{\ket{\psi^{*}}, \ket{\psi^{*}_{\perp}}\}$ with $\ket{\psi^{*}_{\perp}}=(\hat{H}_{f}-I)\ket{\psi_{0}}/\sqrt{1-E_{0}}$, and the great circle arising from that span matches ours.
However, those analyses do not explain why the algorithm should remain confined to the subspace, since their argument only applies when $\alpha_{k}=\beta_{k}=\pi$; this interpretation cannot be applied to the $\pi/3$ algorithm and fixed-point quantum search.
Our result shows that this plane naturally arises as the trajectory of ITE and provides a dynamical viewpoint that applies beyond the specific setting.

We also note that the time duration $s^{*}$ in Eq.~\eqref{eq:opt_s} attains the quantum speed limit~\cite{margolus1998maximum,mandelstam1991uncertainty}, i.e., the minimum time required for one state to reach the target state under a given dynamics.
This highlights that ITE for unstructured search saturates a fundamental limit on how fast a physical process can transform quantum states.
From a broader perspective, this result might provide a direct connection between the efficiency of quantum algorithms and fundamental physical limits on information processing~\cite{lloyd2000ultimate}.

\subsection{Optimal Query Complexity via the ITE Formulation} \label{app:optimal_complexity}

Next, we connect the geometric ITE perspective to the query complexity of Grover's algorithm~\cite{bennett1997strengths,beals1998tight,zalka1999grover,nielsen2010quantum}. Ref.~\cite{nielsen2006quantum} shows that analyzing geodesics provides a potential framework for understanding the difficulty of implementing quantum algorithms.
Specifically, the geodesic length on a Riemannian manifold determines the number of elementary gates required to realize a target unitary operation.

Inspired by Ref.~\cite{nielsen2006quantum}, we confirm that the query complexity of Grover's algorithm is determined by the geodesic length of ITE. 
This link also underlies the ability to achieve the optimal query complexity.

\begin{theoremA}[Geodesic length of ITE determines query complexity of Grover's algorithm] \label{app_thm:geodesic_complexity}
    Given a projector Hamiltonian $\hat{H}_f$, consider the ITE evolution generated by the operator $e^{s^* [\hat{H}_f, \psi_0]}$, where the optimal time $s^*$ of Eq.~\eqref{app_eq:opt_s} ensures that ITE reaches the solution state $\ket{\psi_{s^*}} = \ket{\psi^*}$.
    Then, there exists a Grover iteration $\prod_{k=1}^{\mathcal{N}}G(\alpha_k,\beta_k)$ satisfying 
\begin{equation} \label{app_eq:norm_grover}
    \left\|e^{s^{*}[\hat{H}_{f},\psi_{0}]}-(-1)^{\mathcal{N}}\prod_{k=1}^{\mathcal{N}}G(\alpha_k,\beta_k)\right\|_{\text{op}} \le \epsilon,
\end{equation}
    for the operator norm $\|\cdot\|_{\text{op}}$ and any $\epsilon \in (0,2)$, using the number of queries 
\begin{equation} \label{app_eq:bound_circuit_depth}
    \mathcal{N} \in \mathcal{O} \left(\frac{1}{\epsilon^2|\pi/2-d_{\mathrm{FS}}|}\right)
\end{equation}
where $d_{\mathrm{FS}}\equiv d_{\mathrm{FS}}(\ket{\psi_{0}},\ket{\psi^{*}})$ is the geodesic length between the initial and solution states.
\end{theoremA}

\begin{proof}
    Consider the Grover iterations $\prod_{k=1}^{\mathcal{N}}G(\alpha_k,\beta_k)$ with angles $\alpha_{2k}=\beta_{2k}=-\alpha_{2k-1}=-\beta_{2k-1}=\sqrt{2s^{*}/\mathcal{N}}$.
    Note that this set of angles are given by the simple approximation of the exponential of commutators using the group commutator and fragmentation; see App.~\ref{app:product_formula} for the detail. 
    Without loss of generality, we also assume $\mathcal{N}$ is even; for odd $\mathcal{N}$, setting the last angles as $\alpha_{\mathcal{N}}=\beta_{\mathcal{N}}=0$ reduces to the situation of the even case with one additional (i.e., constant) query.
    Then, the error bound of Eq.~\eqref{app_eq:norm_grover} is rewritten as
    \begin{equation} \label{eq:norm_of_implementation}
    \left\|e^{s^{*}[\hat{H}_{f},\ket{\psi_{0}}\bra{\psi_{0}}]}-\left(e^{i\sqrt{2s^{*}/\mathcal{N}}\ket{\psi_{0}}\bra{\psi_{0}}}e^{i\sqrt{2s^{*}/\mathcal{N}}\hat{H}_{f}}e^{-i\sqrt{2s^{*}/\mathcal{N}}\ket{\psi_{0}}\bra{\psi_{0}}}e^{-i\sqrt{2s^{*}/\mathcal{N}}\hat{H}_{f}}\right)^{\mathcal{N}/2}\right\|_{\text{op}} \le \epsilon,
    \end{equation}
    where $\|\cdot\|_{\text{op}}$ represents the operator norm.
    The upper bound of Eq.~\eqref{eq:norm_of_implementation} is given by
    \begin{equation} \label{eq:query_deriv_mid_1}
    \begin{split}
        & \left\|e^{s^{*}[\hat{H}_{f},\ket{\psi_{0}}\bra{\psi_{0}}]}-\left(e^{i\sqrt{2s^{*}/\mathcal{N}}\ket{\psi_{0}}\bra{\psi_{0}}}e^{i\sqrt{2s^{*}/\mathcal{N}}\hat{H}_{f}}e^{-i\sqrt{2s^{*}/\mathcal{N}}\ket{\psi_{0}}\bra{\psi_{0}}}e^{-i\sqrt{2s^{*}/\mathcal{N}}\hat{H}_{f}}\right)^{\mathcal{N}/2}\right\|_{\text{op}}\\
        &\le \frac{\mathcal{N}}{2} \left\|e^{\frac{2s^{*}}{\mathcal{N}}[\hat{H}_{f},\ket{\psi_{0}}\bra{\psi_{0}}]}-e^{i\sqrt{2s^{*}/\mathcal{N}}\ket{\psi_{0}}\bra{\psi_{0}}}e^{i\sqrt{2s^{*}/\mathcal{N}}\hat{H}_{f}}e^{-i\sqrt{2s^{*}/\mathcal{N}}\ket{\psi_{0}}\bra{\psi_{0}}}e^{-i\sqrt{2s^{*}/\mathcal{N}}\hat{H}_{f}}\right\|_{\text{op}} \\
        & \le  \frac{\mathcal{N}}{2} \left(\frac{2s^{*}}{\mathcal{N}}\right)^{\frac{3}{2}} \bigg( \left\|[\hat{H}_{f}, [\hat{H}_{f}, \ket{\psi_{0}}\bra{\psi_{0}}]]\right\|_{\text{op}} + \left\|[\ket{\psi_{0}}\bra{\psi_{0}}, [\ket{\psi_{0}}\bra{\psi_{0}}, \hat{H}_{f}]]  \right\|_{\text{op}}\bigg).
    \end{split}
    \end{equation}
    In the first equality, we apply the telescoping followed by the triangle inequality $\mathcal{N}/2$ times utilizing the fact that $e^{\frac{2s^{*}}{\mathcal{N}}[\hat{H}_{f},\ket{\psi_{0}}\bra{\psi_{0}}]}$ and $e^{i\sqrt{2s^{*}/\mathcal{N}}\ket{\psi_{0}}\bra{\psi_{0}}}e^{i\sqrt{2s^{*}/\mathcal{N}}\hat{H}_{f}}e^{-i\sqrt{2s^{*}/\mathcal{N}}\ket{\psi_{0}}\bra{\psi_{0}}}e^{-i\sqrt{2s^{*}/\mathcal{N}}\hat{H}_{f}}$ are unitary.
    Lastly, we use the inequality proved in Ref.~\cite{double_bracket2024};
    \begin{equation}\label{eq:GCI_compilation}
\begin{split}
    &\left\|  e^{i\sqrt{s}\Psi}e^{i\sqrt{s}\hat{H}}    
    e^{-i\sqrt{s}\Psi}
    e^{-i\sqrt{s}\hat{H}}- e^{s[\hat{H},\Psi]} \right\|_{\text{op}}\leq   s^{3/2} \bigg( \|[\hat{H}, [\hat{H}, \Psi]]\|_{\text{op}} + \|[\Psi, [\Psi, \hat{H}]]  \|_{\text{op}}\bigg)
\end{split}
\end{equation}
    for the density matrix representation of a pure state $\Psi=\ket{\Psi}\bra{\Psi}$, an arbitrary Hermitian matrix $\hat{H}$ and $s\ge0$.
    Note that 
    \begin{equation}
    \begin{split}
        [\hat{H}_{f}, [\hat{H}_{f}, \ket{\psi_{0}}\bra{\psi_{0}}]] &= [\hat{H}_{f},\hat{H}_{f}\ket{\psi_{0}}\bra{\psi_{0}}-\ket{\psi_{0}}\bra{\psi_{0}}\hat{H}_{f}]\\& = \hat{H}_{f}\ket{\psi_{0}}\bra{\psi_{0}}-2 \hat{H}_{f}\ket{\psi_{0}}\bra{\psi_{0}}\hat{H}_{f} + \ket{\psi_{0}}\bra{\psi_{0}}\hat{H}_{f},
    \end{split}
    \end{equation}
    as $\hat{H}_{f}$ is a projector, i.e., $\hat{H}_{f}^{2}=\hat{H}_{f}$.
    Since the Hilbert-Schmidt norm $\|\cdot\|_{\text{HS}}$ always upper bounds the operator norm, we obtain
    \begin{equation} \label{app_eq:norm_bound_1}
    \begin{split}
        \left\|[\hat{H}_{f}, [\hat{H}_{f}, \ket{\psi_{0}}\bra{\psi_{0}}]]\right\|_{\text{op}} & \le \left\|[\hat{H}_{f}, [\hat{H}_{f}, \ket{\psi_{0}}\bra{\psi_{0}}]]\right\|_{\text{HS}} \\
        &\le \left\|\hat{H}_{f}\ket{\psi_{0}}\bra{\psi_{0}}-2 \hat{H}_{f}\ket{\psi_{0}}\bra{\psi_{0}}\hat{H}_{f} + \ket{\psi_{0}}\bra{\psi_{0}}\hat{H}_{f}\right\|_{\text{HS}} = \sqrt{2V_{0}}.
    \end{split}
    \end{equation}
    Similarly, as
    \begin{equation}
    \begin{split}
        [\ket{\psi_{0}}\bra{\psi_{0}}, [\ket{\psi_{0}}\bra{\psi_{0}}, \hat{H}_{f}]] & = [\ket{\psi_{0}}\bra{\psi_{0}}, \ket{\psi_{0}}\bra{\psi_{0}}\hat{H}_{f}-\hat{H}_{f}\ket{\psi_{0}}\bra{\psi_{0}}]\\&= \hat{H}_{f}\ket{\psi_{0}}\bra{\psi_{0}}-2 E_{0}\ket{\psi_{0}}\bra{\psi_{0}} + \ket{\psi_{0}}\bra{\psi_{0}}\hat{H}_{f},
    \end{split}
    \end{equation}
    we also have $\left\|[\ket{\psi_{0}}\bra{\psi_{0}}, [\ket{\psi_{0}}\bra{\psi_{0}}, \hat{H}_{f}]]  \right\|_{\text{op}}\le \sqrt{2V_{0}}$.
    Consequently, Eq.~\eqref{eq:query_deriv_mid_1} can be further bounded as follows;
    \begin{equation} \label{eq:query_deriv_mid_2}
    \begin{split}
        \frac{\mathcal{N}}{2} \left(\frac{2s^{*}}{\mathcal{N}}\right)^{\frac{3}{2}} \bigg( \left\|\left[\hat{H}_{f}, [\hat{H}_{f}, \ket{\psi_{0}}\bra{\psi_{0}}] \right]\right\|_{\text{op}} + \left\|\left[\ket{\psi_{0}}\bra{\psi_{0}}, [\ket{\psi_{0}}\bra{\psi_{0}}, \hat{H}_{f}] \right]  \right\|_{\text{op}}\bigg) & \le \frac{\mathcal{N}}{2} \left(\frac{2s^{*}}{\mathcal{N}}\right)^{\frac{3}{2}} \cdot 2\sqrt{2}\sqrt{V_{0}}\\
        & = 4\frac{\sqrt{V_{0}}{s^{*}}^{3/2}}{\mathcal{N}^{1/2}} \\
        &\le 2\pi \left(\frac{s^{*}}{\mathcal{N}}\right)^{1/2}.
    \end{split}
    \end{equation}
    In the last inequality, we use the relation $s^{*} = \arccos(\sqrt{E_{0}})/\sqrt{V_{0}}$, which implies $\sqrt{V_{0}}s^{*} = \arccos(\sqrt{E_{0}}) \le \pi/2 $ since $0\le\sqrt{E_{0}}\le 1$.

    Now, we relate the upper bound to the geodesic length on the complex projective manifold with respect to the Fubini-Study metric; 
    \begin{equation} \label{app_eq_geodesic_length}
    d_{\mathrm{FS}} = \arccos{\left(\left|\braket{\psi_{0}|\psi^{*}}\right|\right)} = \arccos{\left(\sqrt{E_{0}}\right)}.
    \end{equation}
    Indeed, using the geodesic length in Eq.~\eqref{app_eq_geodesic_length}, the optimal time duration can be expressed as 
    \begin{equation}
       s^{*}= \arccos(\sqrt{E_{0}})/\sqrt{V_{0}}= \frac{d_{\mathrm{FS}}}{\cos{(d_{\mathrm{FS}})}\sin{(d_{\mathrm{FS}})}} = \frac{1}{\text{sinc }(2d_{\mathrm{FS}})},
    \end{equation}
    with $\text{sinc }(x)=\sin(x)/x$.
    Here, from the definition of $d_{FS}$ in Eq.~\eqref{app_eq_geodesic_length}, we utilize the identities $\cos(d_{FS})=\sqrt{E_{0}}$ and $\sin(d_{FS})=\sqrt{1-E_{0}}$ to rewrite $V_{0}=E_{0}(1-E_{0})$. 
    Then we have
    \begin{equation}
    \frac{1}{\text{sinc }(2x)} \le \frac{2}{|\pi/2-x|}
\end{equation}
for $0\le x\le \pi/2$.
Thus, the upper bound of Eq.~\eqref{eq:query_deriv_mid_2} is further given by
\begin{equation}
    2\pi \left(\frac{s^{*}}{\mathcal{N}}\right)^{1/2} \le 2\sqrt{2}\pi \left(\frac{1}{\mathcal{N}|\pi/2-d_{\mathrm{FS}}|}\right)^{1/2}.
\end{equation}
As a result, to achieve the error $\epsilon$, it suffices to have
\begin{equation}
    \mathcal{N} = \left\lceil \frac{(2\sqrt{2}\pi)^2}{\epsilon^2}\frac{1}{|\pi/2-d_{\mathrm{FS}}|} \right\rceil.
\end{equation}
\end{proof}

Theorem~\ref{app_thm:geodesic_complexity} shows that the number of queries is determined by the geodesic length.
Here, the angles $\{(\alpha_{k},\beta_{k})\}_{k=1}^{\mathcal{N}}$ for the Grover iteration are obtained by a simple product formula approximation of ITE.
Note that $d_{\mathrm{FS}}$ ranges over $[0,\pi/2]$ and increases with decreasing overlap between the states.
This implies that larger distances in the complex projective manifold lead to higher query complexities.

Additionally, $d_{\mathrm{FS}}$ aligns with the geodesic length on the special unitary group equipped with a bi-invariant metric up to a multiplicative factor, and hence a similar result could hold in this setting.
The geodesic between two unitary operators $U,V$ on the special unitary group~\cite{lewis2025geodesic} is given by $e^{i\Gamma t}$ for $t\in[0,1]$ with 
\begin{equation}
    \Gamma = -i\log(U^{\dagger}V).
\end{equation}
Accordingly, the geodesic length is expressed as $\| \log(U^{\dagger}V) \|_{\text{HS}}$.
In the context of unstructured search, we are interested in two unitary operators $I$ and $e^{s^{*}[\hat{H}_{f},|\psi_{0}\rangle\langle\psi_{0}|]}$ with $s^{*}$ of Eq.~\eqref{app_eq:opt_s}, since actions of these operators on $\ket{\psi_{0}}$ yield the initial and the solution states, respectively i.e., $\ket{\psi_{0}}=I\ket{\psi_{0}}$ and $\ket{\psi^{*}}=e^{s^{*}[\hat{H}_{f},|\psi_{0}\rangle\langle\psi_{0}|]}\ket{\psi_{0}}$. 
Thus, the geodesic length on the manifold is given by
\begin{equation}
    \| \log(e^{s^{*}[\hat{H}_{f},|\psi_{0}\rangle\langle\psi_{0}|]}) \|_{\text{HS}} = s^{*}\|[\hat{H}_{f},|\psi_{0}\rangle\langle\psi_{0}|]\|_{\text{HS}} = \arccos{(\sqrt{E_{0}})}/\sqrt{V_{0}} \cdot \sqrt{2V_{0}} = \sqrt{2}\arccos{(\sqrt{E_{0}})} = \sqrt{2}d_{\mathrm{FS}}.
\end{equation}

\medskip

We further confirm this scaling is optimal in $N$, even though the Grover iteration is a simple product formula approximation of ITE.
From Eq.~\eqref{eq:query_deriv_mid_2}, we can further bound the left-hand side of Eq.~\eqref{eq:norm_of_implementation} by providing an explicit upper bound on $s$.
Since we have 
\begin{equation}
    s^{*} =\frac{1}{\text{sinc }(2d_{FS})} \le \frac{2}{|\pi/2-d_{\mathrm{FS}}|}=\frac{2}{\arcsin(\sqrt{E_{0}})} \le \frac{2}{\sqrt{E_{0}}},
\end{equation}
we obtain
\begin{equation} \label{app_eq:optimality}
    \frac{(2\sqrt{2}\pi)^2}{\epsilon^2}\frac{1}{\sqrt{E_{0}}} \le \mathcal{N}.
\end{equation}
We remind that $E_{0}=M/N$. Hence, Eq.~\eqref{app_eq:optimality} indicates that the approach also achieves the quadratic speed-up.

We remark that our query complexity might not be the best in practice; for instance, Refs.~\cite{long2001Grover,Roy2022Grover} present specific angle sets that could achieve zero error $\epsilon=0$ with improved multiplicative factor for $\sqrt{N}$. 
However, our main point is to show that the observation from Ref.~\cite{nielsen2006quantum} can also be applied to unstructured search --  i.e., the geodesic length determines the efficiency of the quantum algorithm.

\subsection{Proof of Theorem~\ref{thm:justification_orig_Grover}} \label{app:rationale_original_Grover}

For clarity, we restate Theorem~\ref{thm:justification_orig_Grover} in the main text.

\begin{theoremA}[Original Grover's algorithm] \label{app_thm:justification_orig_Grover}
The original Grover's algorithm generates the state $\ket{\psi_{s_{\tau}}}$ in Eq.~\eqref{eq:dbr}; that is, there exists a parameter $s(\mathcal{N})$ such that 
\begin{align}
    (-1)^{\mathcal{N}}G(\pi,\pi)^{\mathcal{N}}\ket{\psi_{0}} = e^{s(\mathcal{N})[ \hat{H}_f,\psi_{0}]}\ket{\psi_0}.
\end{align}
Moreover, the original Grover algorithm maximizes the fidelity of the first iteration. That is, within the first order approximations in Eq.~\eqref{eq:group_commutator}, corresponding to $\mathcal{N}=2$, the fidelity $F_{2}=|\braket{\psi^{*}|\prod_{k=1}^{2}G(\alpha_k,\beta_k)|\psi_{0}}|^2$ is maximized when $\alpha_k=\beta_k=\pi$, provided $E_{0}\le 1/8$, e.g.,  when $M\ll N$.
\end{theoremA}

\begin{proof}

We first show that the original Grover algorithm with $\mathcal{N}$ iteration, where the angles are fixed to $\alpha_k=\beta_{k}=\pi$, realizes the ITE dynamics for a suitable step size $s(\mathcal{N})$ up to a global phase.
Namely, we prove that there exists $s(\mathcal{N})$ such that 
\begin{equation} \label{app_eq:grover_ite_relation}
    (-1)^{\mathcal{N}}G(\pi,\pi)^{\mathcal{N}}\ket{\psi_{0}} = e^{s(\mathcal{N})[\hat{H}_f, \psi_{0}]} \ket{\psi_{0}}.
\end{equation}
The statement is provided by induction on $\mathcal{N}$.
\begin{itemize}
    \item \textbf{Base case: $\mathcal{N}=1$.}
    For one Grover iteration, we have
    \begin{equation}
    \begin{split}
        G(\pi,\pi)\ket{\psi_{0}} &= (I-2\psi_{0})(I-2\hat{H}_f)\ket{\psi_{0}} \\
        &= (I - 2\psi_{0} -2\hat{H}_f + 4\psi_{0} \hat{H}_f )\ket{\psi_{0}} \\
        &= (4E_{0}-1)\ket{\psi_{0}}  - 2\hat{H}_f\ket{\psi_{0}} ,
    \end{split}
    \end{equation}
    where we use $\psi_{0} \hat{H}_f\ket{\psi_{0}} = \braket{\psi_{0}|\hat{H}_f|\psi_{0}}\ket{\psi_{0}}=E_{0}\ket{\psi_{0}}$.
    On the other hand, the ITE trajectory starting from $\ket{\psi_0}$ takes the form
    \begin{equation} \label{eq:full_form_ite_approx}
        \ket{\psi_{s}} = \left(\left(-\frac{E_{0}}{\sqrt{V_0}}\sin\left(s\sqrt{V_{0}}\right)+\cos\left(s\sqrt{V_{0}}\right)\right)I+ \sin\left(s\sqrt{V_{0}}\right)\frac{\hat{H}_{f}}{\sqrt{V_{0}}}  \right) \ket{\psi_{0}},
    \end{equation}
    where $V_{0}=E_{0}(1-E_{0})$, as shown in Eq.~\eqref{app_eq:full_desc_1st_app_ITE}.
    Thus, matching coefficients of $I$ and $\hat{H}_f$ gives
    \begin{align}
     -\frac{E_{0}}{\sqrt{V_0}}\sin\left(s(1)\sqrt{V_{0}}\right)+\cos\left(s(1)\sqrt{V_{0}}\right) &= 4E_{0}-1, \label{app_eq:N=1_coef_I} \\
     \frac{\sin\left(s(1)\sqrt{V_{0}}\right)}{\sqrt{V_0}} &= - 2. \label{app_eq:N=1_coef_H}
    \end{align}
    By computing \eqref{app_eq:N=1_coef_I} $+$ \eqref{app_eq:N=1_coef_H} $\times$ $E_{0}$, we obtain $\cos\left(s_1\sqrt{V_{0}}\right) = 2E_{0}-1$.
    Since $\sin\left(s(1)\sqrt{V_{0}}\right)$ is negative from Eq.~\eqref{app_eq:N=1_coef_H}, a valid solution is
    \begin{equation}
        s(1)= \frac{1}{\sqrt{V_{0}}}\left(2\pi- \arccos(2E_{0}-1)\right).
    \end{equation}
    Therefore, Eq.~\eqref{app_eq:grover_ite_relation} holds for $\mathcal{N}=1$.
    
    \item \textbf{Inductive step.}
    Assume that for some $k\ge2$, there exists $s(k-1)$ such that 
    \begin{equation}
        G(\pi,\pi)^{k-1}\ket{\psi_{0}} = e^{s(k-1)[\hat{H}_f, \psi_{0}]} \ket{\psi_{0}}.
    \end{equation}
    That is, Eq.~\eqref{eq:full_form_ite_approx} indicates that the state can be written as $\ket{\psi_{k-1}}:= G(\pi,\pi)^{k-1}\ket{\psi_{0}} = (aI+b\hat{H}_f)\ket{\psi_{0}} $ with real coefficients $a,b\in \mathbb{R}$.
    Also, the normalization condition, i.e., $|\braket{\psi_{k-1}|\psi_{k-1}}|=1$, gives $a^2+2ab E_{0} + b^2=1$.

    Applying one additional Grover iteration yields 
    \begin{equation}
    \begin{split}
        G(\pi,\pi)\ket{\psi_{k}} &= (I-2\psi_{0})(I-2\hat{H}_f)\ket{\psi_{k}} \\
        &= (I - 2\psi_{0} -2\hat{H}_f + 4\psi_{0} \hat{H}_f ) (aI+b\hat{H}_{f})\ket{\psi_{0}} \\
        &= (-a+2(2a+b)E_{0})\ket{\psi_{0}}  - (2a+b)\hat{H}_f\ket{\psi_{0}}.
    \end{split}
    \end{equation}
    Then, comparing again with Eq.~\eqref{eq:full_form_ite_approx} leads to
    \begin{align}
     -\frac{E_{0}}{\sqrt{V_0}}\sin\left(s(k)\sqrt{V_{0}}\right)+\cos\left(s(k)\sqrt{V_{0}}\right) &= -a+2(2a+b)E_{0}, \label{app_eq:N=k-1_coef_I} \\
     \frac{\sin\left(s(k)\sqrt{V_{0}}\right)}{\sqrt{V_0}} &= - (2a+b) \label{app_eq:N=k-1_coef_H}.
    \end{align}
    By computing \eqref{app_eq:N=k-1_coef_I} $+$ \eqref{app_eq:N=k-1_coef_H} $\times$ $E_{0}$, we obtain
    \begin{align}
    \cos\left(s(k)\sqrt{V_{0}}\right) &= -a+(2a+b)E_{0}, \label{app_eq:N=k-1_coef_I_2} \\
     \sin\left(s(k)\sqrt{V_{0}}\right) &= - (2a+b)\sqrt{E_{0}(1-E_{0})}. \label{app_eq:N=k-1_coef_H_2}
    \end{align}

    It remains to show that the right-hand sides define a valid sine-cosine pair. 
    Indeed,
    \begin{equation} \label{app_eq:normalization_N=k-1}
    \begin{split}
        &\left(-a+(2a+b)E_{0}\right)^2 + \left( - (2a+b)\sqrt{E_{0}(1-E_{0})}\right)^2 \\
        &= a^2 -2a(2a+b)E_{0} + (2a+b)^2E_{0}^2 + (2a+b)^2 E_{0}(1-E_{0}) \\
        &= a^2 + (2a+b)bE_{0} \\
        &=1,
    \end{split}
    \end{equation}
    where we used the normalization condition in the last line.
    Since $a,b$ are real coefficients, the right-hand sides of Eqs.~\eqref{app_eq:N=k-1_coef_I_2} and~\eqref{app_eq:N=k-1_coef_H_2} are also real.
    Together with Eq.~\eqref{app_eq:normalization_N=k-1}, this implies that the right-hand sides of Eqs.~\eqref{app_eq:N=k-1_coef_I_2} and~\eqref{app_eq:N=k-1_coef_H_2} lie within the range $[-1,1]$.
    Consequently, this indicates that there exist $s(k)$ satisfying  
     \begin{equation}
        G(\pi,\pi)^{k}\ket{\psi_{0}} = e^{s(k)[\hat{H}_f, \psi_{0}]} \ket{\psi_{0}}.
    \end{equation}
\end{itemize}
By induction, the statement holds for all $\mathcal{N}$.

\medskip
\medskip

Next, we show that, under the first-order group commutator approximation of the exponential of a commutator, choosing $\alpha_k=\beta_k=\pi$ is optimal when the initial energy $E_0$ is small.

We recall the group commutator approximation of the exponential of commutator:  
\begin{equation} \label{app_eq:group_commutator}
    e^{s[\hat{H}_f,\psi_{0}]}= \underbrace{e^{i\sqrt{s}\psi_{0}}e^{i\sqrt{s}\hat{H}_f}e^{-i\sqrt{s}\psi_{0}}e^{-i\sqrt{s}\hat{H}_f}}_{:=\mathcal{P}(\sqrt{s})} + \mathcal O(s^{3/2}).
\end{equation}
To determine the maximal effective step size, we consider the energy
$E(\sqrt{s}):=\braket{\psi_{0}|\mathcal{P}^{\dagger}(\sqrt{s}) \hat{H}_{f} \mathcal{P}(\sqrt{s})|\psi_{0}}$, which quantifies the improvement of the state under the approximate ITE.

For any projector $\Psi$, we obtain $e^{i\theta\Psi}=I+c_{\theta}\Psi$ with $c_{\theta}=e^{i\theta}-1$.
Using this identity, the state after applying the operator $\mathcal{P}(\theta)$ reads
\begin{equation}
\begin{split}
    \mathcal{P}(\theta)\ket{\psi_{0}} &= (I+c_{\theta}\psi_{0})(I+c_{\theta}\hat{H}_{f})(I+c^{*}_{\theta}\psi_{0})(I+c^{*}_{\theta}\hat{H}_{f}) \ket{\psi_{0}} \\
    & = \Bigl( I + (c_{\theta}+c^{*}_{\theta})(\hat{H_{f}} + \psi_{0}) + c_{\theta}c^{*}_{\theta} (\hat{H}_{f}^{2} + \psi_{0}^{2}) + (c^{2}_{\theta}+{c^{*}_{\theta}}^2+c_{\theta}c^{*}_{\theta} ) \psi_{0}\hat{H}_{f} + c_{\theta} c^{*}_{\theta} \hat{H}_{f}\psi_{0} \\
    & \qquad \quad + c^{2}_{\theta} c^{*}_{\theta} (\psi_{0}\hat{H}_{f}^2 + \psi_{0}\hat{H}_{f}\psi_{0}) +  c_{\theta} {c^{*}_{\theta}}^{2} (\psi_{0}^2\hat{H}_{f} + \hat{H}_{f}\psi_{0}\hat{H}_{f}) + c^{2}_{\theta} {c^{*}_{\theta}}^2  \psi_{0}\hat{H}_{f}\psi_{0}\hat{H}_{f})\Bigr)  \ket{\psi_{0}} \\
    &= \Bigl( 1+(c_{\theta}+c^{*}_{\theta}+c_{\theta}c^{*}_{\theta})+(c_{\theta}^2+{c^{*}}^2_{\theta}+c_{\theta}c^{*}_{\theta})E_{0} + c_{\theta}c^{*}_{\theta}E_{0}(2c_{\theta} + c^{*}_{\theta} + E_{0}c_{\theta}c^{*}_{\theta}) \Bigr) \ket{\psi_{0}} \\
    & \quad +  \Bigl( (c_{\theta}+c^{*}_{\theta}+c_{\theta}c^{*}_{\theta}) + c_{\theta}c^{*}_{\theta}(1+E_{0}c^{*}_{\theta})\Bigr) \hat{H}_{f} \ket{\psi_{0}} \\
    &= (1+(c_{\theta}^2+{c^{*}}^2_{\theta}+c_{\theta}c^{*}_{\theta})E_{0}+c_{\theta}c^{*}_{\theta}E_{0}(2c_{\theta} + c^{*}_{\theta} + E_{0}c_{\theta}c^{*}_{\theta})) \ket{\psi_{0}} + (c_{\theta}c^{*}_{\theta}(1+E_{0}c^{*}_{\theta}) ) \hat{H}_{f} \ket{\psi_{0}},
\end{split} 
\end{equation}
where we introduce $\hat{H}_{f}^2=\hat{H}_{f}$, $\psi_{0}^2=\psi_{0}$ and $E_{0}=\braket{\psi_{0}|\hat{H}_{f}|\psi_{0}}$ in the last equality.
Also, we utilize $c_{\theta}+c^{*}_{\theta}+c_{\theta}c^{*}_{\theta}=0$ by showing
\begin{align}
    c_{\theta}c^{*}_{\theta} &= (e^{i\theta}-1)(e^{-i\theta}-1) = 2-2\cos(\theta), \\
    c_{\theta}+c^{*}_{\theta} &= (e^{i\theta}-1)+(e^{-i\theta}-1) =-2+ 2\cos(\theta).
\end{align}

With the above form, the energy can be rewritten as 
\begin{equation}
\begin{split}
    \braket{\psi_{0}|\mathcal{P}^{\dagger}(\theta) \hat{H}_{f} \mathcal{P}(\theta)|\psi_{0}} &= \left(C^{I}_\theta\ket{\psi_{0}} + C^{\hat{H}_{f}}_\theta\hat{H}_{f}\ket{\psi_{0}}\right)^{\dagger} \hat{H}_{f} \left(C^{I}_\theta\ket{\psi_{0}} + C^{\hat{H}_{f}}_\theta\hat{H}_{f}\ket{\psi_{0}}\right) \\
    &= E_{0} \left|C^{I}_\theta+C^{\hat{H}_{f}}_\theta\right|^2 \\
    &:= E_{0} g(\theta)^2,
\end{split} 
\end{equation}
where we define $C^{I}_\theta =1+(c_{\theta}^2+{c^{*}}^2_{\theta}+c_{\theta}c^{*}_{\theta})E_{0}+c_{\theta} c^{*}_{\theta}E_{0}(2c_{\theta} + c^{*}_{\theta} + E_{0}c_{\theta}c^{*}_{\theta}) $, $C^{\hat{H}_{f}}_{\theta}=c_{\theta}c^{*}_{\theta}(1+E_{0}c^{*}_{\theta})$ and $g(\theta)= |C^{I}_{\theta}+C^{\hat{H}_{f}}_{\theta}|$.
Since $E_{0}\in(0,1]$, the value of $\theta$ that maximizes $g(\theta)$ corresponds to the optimal step size.

The function $g(\theta)$ can be simplified further:
\begin{equation}
\begin{split}
    g(\theta)&=\left|C^{I}_{\theta}+C^{\hat{H}_{f}}_{\theta}\right| \\
    &= | 1+(c_{\theta}^2+{c^{*}}^2_{\theta}+c_{\theta}c^{*}_{\theta})E_{0}+c_{\theta}c^{*}_{\theta}E_{0}(2c_{\theta} + c^{*}_{\theta} + E_{0}c_{\theta}c^{*}_{\theta}) + c_{\theta}c^{*}_{\theta}(1+E_{0}c^{*}_{\theta}) | \\
    &= |1+  c_{\theta}c^{*}_{\theta} + c_{\theta}c^{*}_{\theta}E_{0}(2(c_{\theta} + c^{*}_{\theta}) + E_{0}c_{\theta}c^{*}_{\theta}) + (c_{\theta}^2+{c^{*}}^2_{\theta}+c_{\theta}c^{*}_{\theta})E_{0} |\\
    &= |1+2(1-\cos(\theta)) + 2(1-\cos(\theta))E_{0}(-4(1-\cos(\theta)) + 2E_{0}(1-\cos(\theta))) + 2(1-2\cos(\theta))(1-\cos(\theta))E_{0} |, \\
\end{split}
\end{equation}
where we use
\begin{equation}
\begin{split}
     c_{\theta}^2+{c^{*}}^2_{\theta}+c_{\theta}c^{*}_{\theta} &= (e^{i\theta}-1)^2 +(e^{-i\theta}-1)^2 + (e^{i\theta}-1)(e^{-i\theta}-1) \\
     &= 2\cos(2\theta) + 4-\cos(\theta) \\
     &= 4\cos^2(\theta) -6\cos(\theta) + 2 \\
     &= 2(1-2\cos(\theta))(1-\cos(\theta)).
\end{split}
\end{equation}

Now, define $X=1-\cos(\theta)\in [0,2]$.
Then, we have
\begin{equation}
\begin{split}
    g(X) &= |1+2X + 2E_{0}X(-4X+2E_{0}X) + 2E_{0}X(2X-1)  | \\
    &= |1+2(1-E_{0})X - 4E_{0}(1-E_{0})X^2 | \\
    &=  \left|- 4E_{0}(1-E_{0})\left(X-\frac{1}{4E_{0}}  \right)^2 + \left(1 + \frac{1-E_{0}}{4E_{0}} \right)  \right|.
\end{split}
\end{equation}

Since $g(X)$ is a quadratic function in $X$, the maximum of $g(X)$ occurs either at the edges $X=0,2$ or at the stationary point $X=1/4E_{0}$.
Evaluating these gives 
\begin{equation}
    g(0)=1, \,\,  g(2)=|16E_{0}^2 -20E_{0} + 5|, \,\, g(1/4E_{0})= 1 + \frac{1-E_{0}}{4E_{0}}. 
\end{equation}
Since $X\in[0,2]$, we analyze the maximum case by case.
As a result, we obtain
\begin{itemize}
    \item If $1/4E_{0}\le2$, i.e., $E_{0}\ge 1/8$, the optimal choice is  $X=1/4E_{0}(2-E_{0})$, corresponding to $$\theta = \arccos(1-1/4E_{0}).$$
    \item If $0<E_{0}\le1/8$, the stationary point exceeds the interval, and the best choice is the boundary $X=2$, i.e., $$\theta=\pi.$$
\end{itemize}
This concludes the proof.

\end{proof}

We here also provide a justification for why this strategy works even in the worst-case scenario, where the marked fraction is exponentially small.
Note that this situation is equivalent to the case where the solution state is far from the initial state.
From out result, the unstructured search problem can be recast as achieving a small $\epsilon$ for
\begin{equation}
    \left\|e^{s^{*}[\hat{H}_f,\psi_{0}]}\ket{\psi_{0}} - \mathcal{P}(\sqrt{s})\ket{\psi_{0}}\right\| \le \epsilon,
\end{equation}
with $s^{*}=\arccos(\sqrt{E_{0}})/\sqrt{V_{0}}$.
Then, we have
\begin{equation}
\begin{split}
     \left\|e^{s^{*}[\hat{H}_f,\psi_{0}]}\ket{\psi_{0}} - \mathcal{P}(s)\ket{\psi_{0}}\right\| & \le \left\|e^{s^{*}[\hat{H}_f,\psi_{0}]}\ket{\psi_{0}} - e^{s[\hat{H}_f,\psi_{0}]}\ket{\psi_{0}}\right\| + \left\|e^{s^{*}[\hat{H}_f,\psi_{0}]}\ket{\psi_{0}} - \mathcal{P}(s)\ket{\psi_{0}}\right\| \\
     & \le |s^{*}-s|\|[\hat{H}_f,\psi_{0}]\|_{\text{HS}} + s^{3/2}   \bigg( \|[\hat{H}, [\hat{H}, \Psi]]\|_{\text{op}} + \|[\Psi, [\Psi, \hat{H}]]  \|_{\text{op}}\bigg) \\
     & \le \sqrt{2}\left(\left|\frac{\pi}{2}-s\sqrt{V_{0}}\right|\right) + 2\sqrt{2V_{0}} s^{3/2},
\end{split}
\end{equation}
where we use the triangle inequality in the first line and apply $|e^{iG}-e^{iG'}|_{\text{op}}\le \|G-G'\|_{\text{op}}$ and Eq.~\eqref{eq:GCI_compilation} in the second line.
In the last line, we use $\|[\hat{H}_f,\psi_{0}]\|_{\text{HS}} \le \sqrt{2V_{0}}$ together with 
Eq.~\eqref{app_eq:norm_bound_1}, followed by $\arccos(\sqrt{E_{0}}) \le \pi/2$.

This inequality shows that, when  $M$ is small (i.e., $V_0 = \Theta(1/N)$), the second term is negligible.
Therefore minimizing the first term, which corresponds to maximizing $s$, contributes to the error reduction.
This justifies the choice of $\sqrt{s}=\pi$ in in the original Grover’s algorithm in the worst-case scenario.
On the other hand, when $M$ is comparable to the number of total items $N$, the second term becomes non-negligible, and the first term can also grow large. This corresponds to the overshooting phenomenon. Consequently, the same choice of $s$ is no longer optimal in this regime.

\subsection{Proof of Proposition~\ref{prop:pi/3_alg}}
\label{app:pi/3_algorithm}

We restate the result for $\pi/3$ algorithm for clarity.

\begin{propositionA}[[ITE formulation for the $\pi/3$-algorithm and the maximal time duration with monotonic convergence] \label{app_prop:pi/3_alg}
For the Grover iterations $\prod_{k=1}^{\mathcal{N}}G_k(t_{2k},t_{2k-1})\ket{\psi_{0}}$, the DB-QITE construction in Ref.~\cite{gluza_DB_QITE_2024} reproduces the recursive structure of $\pi/3$ algorithm, i.e., 
\begin{equation} \label{app_eq:pi/3-alg}
    U_{k+1} = U_{k} D(\pi/3) U_{k}^{\dagger} U_{f}(\pi/3) U_{k},
\end{equation}
Moreover, the choice of $\pi/2$ instead of $\pi/3$ realizes the largest effective ITE dynamics within the same approximation for which the fidelity with the solution state contracts monotonically, independent of the number of marked items $M$.
\end{propositionA}

\begin{proof}

    We first show that the recursive strategy used in the double-bracket quantum imaginary-time evolution (DB-QITE) framework~\cite{gluza_DB_QITE_2024} reproduces the $\pi/3$ algorithm.
    Recall that $\pi/3$ algorithm is given by in Eq.~\eqref{app_eq:pi/3-alg}, i.e. 
    \begin{equation} 
        \underbrace{e^{ -i\sqrt{s}\hat{H}_{f}}e^{ i\sqrt{s}\psi}e^{ i\sqrt{s}\hat{H}_{f}}e^{ i\sqrt{s}\psi}}_{:=\mathcal{P}(s)}=e^{s[\hat{H}_{f},\psi]} + \mathcal{O}(s^{3/2}),
    \end{equation}
    where $\psi=\ket{\psi}\bra{\psi}$.
    Here, we use the identity $[A,B]=[-iA,iB]$ to fix both the order and the sign in the expansion; see App.~\ref{app:product_formula}.
    When the unitary acts on the state $\ket{\psi}$, the expression simplifies to 
    \begin{equation}
        e^{ i\sqrt{s}}\mathcal{P}(s) \ket{\psi} = e^{ -i\sqrt{s}\hat{H}_{f}}e^{ i\sqrt{s}\psi}e^{ i\sqrt{s}\hat{H}_{f}} \ket{\psi},
    \end{equation}
    where only the global phase $e^{ i\sqrt{s}}$ remains after applying $e^{ i\sqrt{s}\psi}$ to the state.
    \vspace{0.1cm}
    We now apply the recursive construction. 
    Defining $\tilde{U}_{k}\ket{\psi_{0}}=\ket{\psi_{k}}$, the update rule reads
     \begin{equation}
    \begin{split}
        \ket{\psi_{k+1}} &= e^{ -i\sqrt{s}\hat{H}_{f}}e^{ i\sqrt{s}\psi_{k}}e^{ i\sqrt{s}\hat{H}_{f}} \ket{\psi_{k}} = e^{ -i\sqrt{s}\hat{H}_{f}}\tilde{U}_{k}e^{ i\sqrt{s}\psi_{0}} \tilde{U}_{k}^{\dagger}e^{ i\sqrt{s}\hat{H}_{f}} \tilde{U}_{k}\ket{\psi_{0}}.
    \end{split}
    \end{equation}
    In fact, the leftmost Hamiltonian evolution $e^{ -i\sqrt{s}\hat{H}_{f}}$ can be omitted when estimating the energy, since it commutes with $\hat{H}_{f}$ and cancels.
    The same reasoning applies at every recursive step, so this factor can be consistently absorbed and removed without affecting the recursive structure.
    Consequently, for the purpose of energy estimation and state preparation, the effective recursion reduces to
    \begin{equation} \label{eq:ite_pi/3}
        \tilde{U}_{k+1} = \tilde{U}_{k}e^{ i\sqrt{s}\psi_{0}} \tilde{U}_{k}^{\dagger}e^{ i\sqrt{s}\hat{H}_{f}} \tilde{U}_{k},
    \end{equation}
    which coincides exactly with the update rule of $\pi/3$-algorithm.
    In particular, setting $\sqrt{s}=\pi/3$ recovers the $\pi/3$-algorithm.

    Next, we identify a set of angles that induces the time-evolution described by Eq.~\eqref{eq:ite_pi/3}, and guarantee monotonic convergence of the expectation value of the energy.
    Let  $E_{k}=\braket{\psi_{k}|\hat{H_{f}}|\psi_{k}}$ denote the energy at iteration $k$ .
    A direct expansion gives
    \begin{equation}
    \begin{split}
         \ket{\psi_{k+1}} &= e^{i\sqrt{s}\psi_{k} } e^{i\sqrt{s}\hat{H_{f}}} \ket{\psi_{k}} =\left(e^{i\sqrt{s}}+(e^{i\sqrt{s}}-1)^2E_{k} \right)\ket{\psi_{k}} + (e^{i\sqrt{s}}-1) \hat{H}_{f}\ket{\psi_{k}}. 
    \end{split}
    \end{equation}
    Then the energy at the next step satisfies
    \begin{equation}\label{eq: growth rate of Ek}
    \begin{split}
        E_{k+1} &=\braket{\psi_{k+1}|\hat{H_{f}}|\psi_{k+1}}  = E_{k} |a(s,k)+b(s)|^2,
    \end{split} 
    \end{equation}
    where $a(s,k)=e^{i\sqrt{s}}+(e^{i\sqrt{s}}-1)^2E_{k}$ and $b(s)=e^{i\sqrt{s}}-1$.
    Therefore, monotonic convergence requires $|a(s,k)+b(s)|^2\ge1$ for all $E_{k}\in(0,1]$.
    By combining the coefficient, we get
    \begin{equation}
        a(s,k)+b(s) = e^{i\sqrt{s}}+(e^{i\sqrt{s}}-1)^2E_{k} + e^{i\sqrt{s}}-1 = E_{k} e^{i2\sqrt{s}} + 2(1-E_{k})e^{i\sqrt{s}}- (1-E_{k}).
    \end{equation}
    and this leads to
    \begin{equation}
    \begin{split}
        &|a(s,k)+b(s)|^2 \\
        &= E_{k}^2 + 4(1-E_{k})^2 + (1-E_{k})^2 + 4E_{k}(1-E_{k})\cos(\sqrt{s}) -4 (1-E_{k})^2\cos(\sqrt{s}) - 2 E_{k}(1-E_{k}) (2\cos^2(\sqrt{s})-1) \\
        &= -4E_{k}(1-E_{k}) \cos^2(\sqrt{s}) - 4(1-2E_{k})(1-E_{k})\cos(\sqrt{s}) + 4E_{k}^2-8E_{k}+ 5. \label{eq: |a+b|^2}
    \end{split}
    \end{equation}
    Thus, imposing the monotonicity condition reveals 
    \begin{equation}
    \begin{split}
        &-4E_{k}(1-E_{k}) \cos^2(\sqrt{s}) - 4(1-2E_{k})(1-E_{k})\cos(\sqrt{s}) + 4E_{k}^2-8E_{k}+ 5 \ge 1 \\
        \Leftrightarrow  & -4E_{k}(1-E_{k}) \cos^2(\sqrt{s}) - 4(1-2E_{k})(1-E_{k})\cos(\sqrt{s}) + 4(1-E_{k})^2 \ge 0 \\
        \Leftrightarrow  & -4(1-E_{k})\left(1-\cos(\sqrt{s})\right)\left(E_{k}\cos(\sqrt{s})+(1-E_{k})\right)  \ge 0 \\
        \Leftrightarrow  & \left(\cos(\sqrt{s})+ \frac{1-E_{k}}{E_k} \right)(\cos(\sqrt{s})-1) \le 0, \\
    \end{split}
    \end{equation}
    where we utilize $E_{k}\in[0,1]$ in the last inequality.
    Since $\cos(\sqrt{s})\le 1$, the inequality therefore reduces to $\cos(\sqrt{s})  \ge 0$, i.e., 
    \begin{equation}
        0\le \sqrt{s} \le \frac{\pi}{2}.
    \end{equation}
    The bound shows that $\sqrt{s}=\pi/2$ is the largest possible evolution time that guarantees monotonic convergence of the energy for all intermediate states: that is, any larger choice violates monotonicity for any $E_{0}$.


\begin{figure}[t]
\centering
\begin{tikzpicture}
\definecolor{lightgray}{HTML}{F4F4F4}
\definecolor{littlelightgray}{HTML}{ecececff}
\definecolor{pale}{HTML}{7ca3d4ff}
\definecolor{lightred}{HTML}{d8a2a2}
\node[anchor=center] (russell) at (-12,-2.3)
{\centering\includegraphics[width=0.43\textwidth]{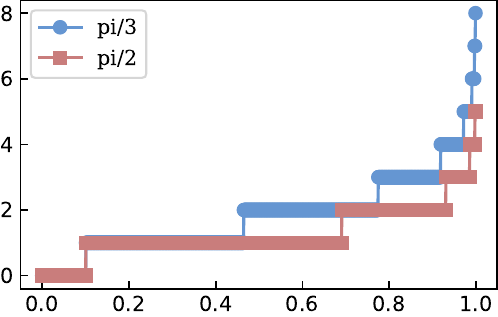}};
\node[text width=4cm] at (-11,-5.3){Initial Overlap $E_{0}$};
\node[rotate=90,text width=4cm] at (-16.2,-2.0) 
    {A Recursion Step $k$ Required};
\node[text width=1cm] at (-16.0,0.5){(a)};
\node[anchor=center] (russell) at (-3.2,-2.3)
{\centering\includegraphics[width=0.43\textwidth]{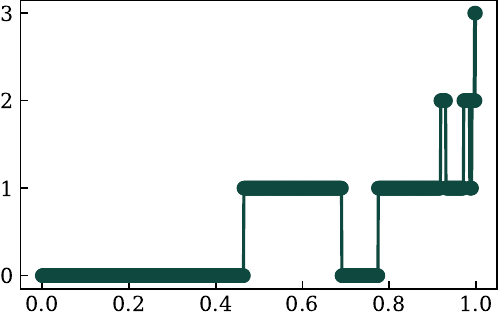}};
\node[text width=4cm] at (-1.9,-5.3){Initial Overlap $E_{0}$};
\node[rotate=90,text width=4cm] at (-7.4,-2.0) 
    {Difference in Recursion Steps};
\node[text width=1cm] at (-7.2,0.5){(b)};\end{tikzpicture}
\caption{\textbf{Comparison between our $\pi/2$-algorithm and $\pi/3$-algorithm.} 
(a) We show the number of recursion steps $k$ required by the $\pi/3$-algorithm and the $\pi/2$-algorithm to reach a final-state infidelity below 0.1.
The data are based on 1,000 points in the range of the initial overlap $E_{0} \in [0,0.999]$.
For each point, we record the recursion step at which the algorithm achieves a fidelity above 0.9.
(b) We plot the difference in recursion steps between the $\pi/2$- and $\pi/3$-algorithms, and confirm that the $\pi/2$-algorithm never requires more steps than the $\pi/3$-algorithm.
}
\label{fig:pi/2_vs_pi/3}
\end{figure}


    A natural question then arises. 
    The $\pi/3$-algorithm is widely regarded as effective, while the ``$\pi/2$-algorithm" has received little attention.
    By computing 
    \begin{equation}
    E_{k}|a( (\pi/2)^2,k)+b((\pi/2)^2)|^2 - E_{k} |a((\pi/3)^2,k)+b((\pi/3)^2)|^2 = E_{k}(1-E_{k})(2-3E_{k}),
\end{equation}
    we find that $E_{k}> 2/3$ is the condition that energy improvement for the $\pi/2$-algorithm becomes worse than the $\pi/3$-algorithm.
    In this regime, the $\pi/3$-algorithm provides a significantly larger guaranteed amplification, whereas the maximum possible difference in energy improvement is only $2(7\sqrt{7}-10)/243\approx 0.07$ attained at $E_{k}=(4-\sqrt{7})/9\approx0.15$.
    Concretely, by defining $E_{k}=1-\epsilon_{k}$, the $\pi/3$-algorithms satisfies $E_{k+1}=1-\epsilon_{k}^{3}$, whereas $\pi/2$-algorithm gives $E_{k+1}=1-\epsilon_{k} + 4\epsilon_{k}^2- 4\epsilon_{k}^{3}=E_{k} + \mathcal{O}(\epsilon_{k}^{2})$.
    Hence, in the regime $\epsilon_{k} \ll 1$, the $\pi/3$-algorithm still exhibits cubic error suppression, while the $\pi/2$-algorithm produces only a marginal improvement per iteration.
    This asymptotic difference explains the superior performance of the $\pi/3$-algorithm in this regime.
    
    On the other hand, the $\pi/2$-algorithm maintains monotonic convergence and achieves faster amplification when the energy is small, a situation that can occur in practice.
    In the regime $E_{k}\ll1$, the update satisfies $E_{k+1}=3E_{k}-3E_{k}^2+E_{k}^3=3E_{k}+\mathcal{O}(E_{k}^2)$ for the $\pi/3$-algorithm, and $E_{k+1}=5E_{k}-8E_{k}^2+4E_{k}^3=5E_{k}+\mathcal{O}(E_{k}^2)$ for the $\pi/2$-algorithm.
    Hence, in this region, both algorithms exhibit linear amplification, but with different growth rates.
    Since the number of queries after $k$ recursive steps scales as $\mathcal{N}_{k}=(3^{k}-1)/2$, the effective amplification for the $\pi/3$-algorithm yields linear scaling in $N$.
    Specifically, for the error after $k$ step, $\epsilon_k=(1-E_{0})^{3^{k}}$, to be less than $\delta$, we have
    \begin{equation}
    \begin{split}
        &  (1-E_{0})^{3^{k}} \le \delta \\
        & \Leftrightarrow (1-E_{0})^{2\mathcal{N}_{k}+1}  \le \delta\\
        &\Leftrightarrow \mathcal{N}_{k} \ge \frac{1}{2} \frac{\log(\delta)}{\log(1-E_{0})} - \frac{1}{2}. \\
    \end{split}
    \end{equation}
    Using $E_{0}=M/N$ and the bound $-1/\log(1-x)\le 1/x$, we obtain $\mathcal{N}_{k}\in \mathcal{O}\left( N\right)$, which coincides with the scaling of classical search.
    However,  the $\pi/2$-algorithm achieves a superlinear scaling $\mathcal{N}_{k}\in \mathcal{O}(N^{1/1.46}) $, because the error after $k$ step reads $(1-E_{0})^{5^{k}}$ and $5^{k}=(2\mathcal{N}_{k}+1)^{\log_{3}(5)}\approx (2\mathcal{N}_{k}+1)^{1.46}$.
    This represents the faster convergence of the $\pi/2$-algorithm without overshooting.
    Note that, for comparison, the original Grover's algorithm satisfies $E_{k+1}=9E_{k}-24E_{k}^2+16E_{k}^3=9E_{k}+\mathcal{O}(E_{k}^2)$, corresponding to a quadratic scaling $\mathcal{N}_{k} \in \mathcal{O}(N^{1/2})$.
    Therefore, while the $\pi/2$-algorithm converges faster than the $\pi/3$-algorithm, it remains asymptotically slower than the original Grover algorithm.
    Nevertheless, owing to its monotonicity without overshooting, the $\pi/2$-algorithm is particularly advantageous when a relatively large error tolerance is acceptable.
    For instance, Fig.~\ref{fig:pi/2_vs_pi/3} shows the number of recursion step $k$ needed for these algorithms to achieve $\epsilon_{k}\le 0.1$.
    This result shows that, as long as one is happy with a modest failure probability, the $\pi/2$-algorithm works better than the $\pi/3$-algorithm.
\end{proof}

\subsection{Proof of Theorem~\ref{thm:qsp_formula_ITE}} \label{app:thm_of_qsp_ite}

For clarity, we restate our result on fixed-point quantum search from the perspective of the ITE formulation.

\begin{theoremA}[ITE formulation provides a fixed-point search via a QSP framework]\label{app_thm:qsp_formula_ITE}
Consider the basis $\{\ket{\psi_{0}}, \ket{\psi_{0}^{\perp}}\}$ corresponding to $\{\ket{0},\ket{1}\}$; i.e., $\ket{\psi_{0}}=(1,0)^T$ and $\ket{\psi_{0}^{\perp}}=(0,1)^T$.
In the basis, the diffusion operator and the oracle operator can be written as
\begin{equation} \label{app_eq:matrix_form_diffusion}
    D(\alpha)=  e^{i\alpha\psi_{0}}  = e^{i\alpha/2}S_{Z}(\alpha/2),
\end{equation}
\begin{equation} \label{app_eq:qsp_basis_mat}
    U_{f}(\beta)=e^{i\beta \hat{H}_{f}} = e^{i\beta/2} R(\sqrt{E_{0}}) S_{Z}(\beta/2) R(\sqrt{E_{0}})
\end{equation}
with $E_{0}=\braket{\psi_{0}|\hat{H}_{f}|\psi_{0}}$, respectively.
Thus, the Grover iteration can execute the ($R(x)$, $S_{Z}$, $\ket{0}$)-QSP;
\begin{equation} \label{app_eq:dbr_qsp}
\begin{split}
    \prod_{k=1}^{\mathcal{N}}G_{k}(\alpha_k,\beta_k)\ket{\psi_{0}} = S_{Z}(\phi_{2\mathcal{N}}) \prod_{k=0}^{2\mathcal{N}-1} R(\sqrt{E_{0}}) S_{Z}(\phi_{k}) \ket{0},
\end{split}
\end{equation}
where $\phi_{0}=\mathcal{N}\pi+\sum_{l=1}^{N}(\alpha_l+\beta_{l})/2$, $\phi_{2l-1}=\beta_{l}/2$ and $\phi_{2l}=\alpha_{l}/2$ for $l=1,\ldots,\mathcal{N}$.

Moreover, by utilizing the QSP technique, angles $\{\phi_{k}\}_{k=0}^{2\mathcal{N}-1}$ that approximate the sign function as a filter can provide a new implementation of fixed-point quantum search.
\end{theoremA}

\begin{proof}
    We first verify the expression in Eq.~\eqref{app_eq:dbr_qsp}.
    By definition, the diffusion operator is given by 
    \begin{equation}
        e^{i\alpha \psi_{0}}=I-(1-e^{i\alpha})\psi_{0}.
    \end{equation}
    Thus, by introducing 
    \begin{equation*}
        I=\begin{pmatrix}1 & 0  \\  0 & 1\end{pmatrix}
    \end{equation*}
    and
    \begin{equation}
        \psi_{0} = |\psi_{0}\rangle\langle\psi_{0}| = \begin{pmatrix}1 & 0  \\  0 & 0\end{pmatrix},
    \end{equation}
    we obtain the expression 
    \begin{equation} \label{app_eq:matrix_form_diffusion_}
        D(\alpha) = \begin{pmatrix}e^{i\alpha} & 0  \\  0 & 1\end{pmatrix},
    \end{equation}

    Similarly, we can write the oracle operator given by 
    \begin{equation}
    U_{f}(\beta) = e^{i\beta \hat{H}_{f}}=I-(1-e^{i\beta})\hat{H}_{f},
    \end{equation}
    in the basis.
    Since $\langle\psi_0|\hat{H}_{f}|\psi_0\rangle=E_{0}$, $\langle\psi_0^{\perp}|\hat{H}_{f}|\psi_0\rangle=\langle\psi_0|\hat{H}_{f}|\psi_0^{\perp}\rangle=\sqrt{E_{0}(1-E_{0})}$ and $\langle\psi_0^{\perp}|\hat{H}_{f}|\psi_0^{\perp}\rangle=1-E_{0}$, the matrix form of $\hat{H}_{f}$ is expressed as
    \begin{equation}
        \hat{H}_{f} = \begin{pmatrix}E_{0} & \sqrt{E_{0}(1-E_{0})}  \\  \sqrt{E_{0}(1-E_{0})} & 1-E_{0}\end{pmatrix}.
    \end{equation}
    Now, we decompose the matrix as 
    \begin{equation}
    \begin{split}
        \begin{pmatrix}E_{0} & \sqrt{E_{0}(1-E_{0})}  \\  \sqrt{E_{0}(1-E_{0})} & 1-E_{0}\end{pmatrix} &=  \begin{pmatrix}\sqrt{E_{0}} & \sqrt{1-E_{0}}  \\  \sqrt{1-E_{0}} & -\sqrt{E_{0}}\end{pmatrix}  \begin{pmatrix}1 & 0  \\  0 & 0\end{pmatrix} \begin{pmatrix}\sqrt{E_{0}} & \sqrt{1-E_{0}}  \\  \sqrt{1-E_{0}} & -\sqrt{E_{0}}\end{pmatrix} \\
        &= R(\sqrt{E_{0}}) \psi_{0} R(\sqrt{E_{0}}),
    \end{split}
    \end{equation}
     where we use the notation defined in Eq.~\eqref{app_eq:qsp_basis_mat}.
     Here, the matrix $R(\sqrt{E_{0}})$ has the property that 
     \begin{equation}
         R(\sqrt{E_{0}})^2 = \begin{pmatrix}\sqrt{E_{0}} & \sqrt{1-E_{0}}  \\  \sqrt{1-E_{0}} & -\sqrt{E_{0}}\end{pmatrix} \begin{pmatrix}\sqrt{E_{0}} & \sqrt{1-E_{0}}  \\  \sqrt{1-E_{0}} & -\sqrt{E_{0}}\end{pmatrix} =\begin{pmatrix}1 & 0  \\  0 & 1\end{pmatrix} = I.
     \end{equation}
    Therefore, the oracle operators can be rewritten as
    \begin{equation} \label{app_eq:uf_2dim}
        U_{f}(\beta) = I-(1-e^{i\beta})\hat{H}_{f} = R(\sqrt{E_{0}})(I-(1-e^{i\beta})\psi_{0})R(\sqrt{E_{0}}) = R(\sqrt{E_{0}})D(\beta)R(\sqrt{E_{0}}).
    \end{equation}
    Combining Eqs.~\eqref{app_eq:matrix_form_diffusion_},~\eqref{app_eq:uf_2dim} together, we obtain
    \begin{equation} 
    \begin{split}
        &\prod_{k=1}^{\mathcal{N}}G_{k}(\alpha_k,\beta_k)\ket{\psi_{0}}  =(-1)^{\mathcal{N}}\prod_{k=1}^{\mathcal{N}} \underbrace{D(\alpha_k)}_{= e^{i\alpha_{k}\psi_{0}}}\underbrace{R(\sqrt{E_{0}})D(\beta_k)R(\sqrt{E_{0}})}_{= e^{i\beta_{k}H_{f}}} \begin{pmatrix} 1 \\ 0 \\ \end{pmatrix}.
    \end{split}
    \end{equation}
    Now, since
    \begin{equation}
        D(\alpha) = \begin{pmatrix}e^{i\alpha} & 0  \\  0 & 1\end{pmatrix} =  e^{i\alpha/2} \begin{pmatrix}e^{i\alpha/2} & 0  \\  0 & e^{-i\alpha/2}\end{pmatrix} = e^{i\alpha/2} S_{Z}(\alpha/2),
    \end{equation}
    we finally arrive at Eq.~\eqref{app_eq:dbr_qsp}.

This finding further allows us to realize the ITE state via the QSP framework.
In the basis $\{\ket{\psi_{0}}, \ket{\psi_{0}^{\perp}}\}$, the ITE state can be expressed as
\begin{equation} \label{app_eq:dbr_vector}
    \ket{\psi_{s}} = \begin{pmatrix}\cos(s\sqrt{V_{0}})   \\  \sin(s\sqrt{V_{0}}) \end{pmatrix}
\end{equation}
with $V_{0}=E_{0}(1-E_{0})$.
Namely, if QSP can realize the target function $\cos(s x \sqrt{1-x^2})$, i.e., an element in ITE state of Eq.~\eqref{app_eq:dbr_vector} with $x=\sqrt{E_{0}}$, the corresponding set of angles can be used for realizing the ITE state via the Grover iterations.

To determine whether this can be realized, we first examine the achievability of the function via QSP.
As discussed in App.~\ref{app:qsp_overview}, the function must satisfy five conditions for successful realization via QSP in this setting.
The first three conditions are straightforwardly satisfied; the cosine function $\cos(t)$ is an even function, bounded by 1 for $t\in[-1,1]$, and can be approximated by degree-$K$ polynomials using e.g., Taylor expansion.     
In general, the fourth and fifth conditions significantly restrict the class of realizable functions~\cite{martyn2021grand}. 
Despite this, the target function meets both conditions favorably.
For the fourth condition, observe that $\sqrt{1-x^{2}} \in \mathbb{C}$ for $x\in(-\infty,-1] \cup [1,\infty)$, and since $\cos(it)=\cosh(t)$, the function $\cos(sx\sqrt{1-x^2})$ exceeds 1 in magnitude, as required.
For the fifth condition, note that the function is even. 
To verify this condition, we evaluate the function under the substitution $x \to ix$ and then we have $\cos(isx\sqrt{1+x^2})=\cosh(sx\sqrt{1+x^2})$, which is real and greater than 1 for $x\in \mathbb{R}$; that is the requirement is satisfied.
As a result, all five conditions are met, and thus there exists a set of QSP angles that realizes the target function for sufficiently large polynomial degrees.

We further consider the efficiency of the QSP implementation of ITE for unstructured search. 
Theorem~\ref{thm:qsp_formula_ITE} establishes the existence of a set of angles that enables  QSP implementation of ITE for a large value of $\mathcal{N}$.
However, it does not provide insights into the efficiency and feasibility of the implementation.
Indeed, it is known that trigonometric functions admit a low-degree polynomial approximation:

\begin{lemmaA}[Polynomial approximation of trigonometric functions by  Jacobi-Anger expansion~\cite{gilyen2019quantum,martyn2021grand}]\label{lem:jacobi-anger}
    Consider $s\in\mathbb{R}$ and $\epsilon\in(0,\frac{1}{e})$.
    Let $K=\lfloor\frac{1}{2}r(s,\epsilon)\left(\frac{e|s|}{2},\frac{5}{4}\epsilon\right)\rfloor$ be a degree. Then, trigonometric functions can be approximated as follows;
    \begin{align} \label{eq:ja_expansion_cos}
        \|\cos(sx)-J_{0}(s)+2\sum_{l=1}^{K}(-1)^{l}J_{2l}(s)T_{2l}(x)\|_{[-1,1]}\le\epsilon,\\
        \|\sin(sx)-2\sum_{l=1}^{K}(-1)^{l}J_{2l+1}(s)T_{2l+1}(x)\|_{[-1,1]}\le\epsilon,
    \end{align}
    where $T_{m}(x)$ are the Chebyshev polynomials of the first kind, $J_{m}(s)$ are the Bessel functions of the first kind and $r(s,\epsilon)$ is a function that asymptotically scales as
    \begin{equation}
        r(s,\epsilon)=\Theta\left(|s|+\frac{\log(1/\epsilon)}{\log(e+\frac{\log(1/\epsilon)}{|s|})}\right).
    \end{equation}
\end{lemmaA}
This suggests that trigonometric functions can be approximated using Chebyshev polynomials of degree $K=\lfloor\frac{1}{2}r(s,\epsilon)\left(\frac{e|s|}{2},\frac{5}{4}\epsilon\right)\rfloor$ with the error $\epsilon$.
However, there is a caveat; the target function in our case is $\cos(sg(x))$ with $g(x)=x\sqrt{1-x^2}$, which introduces additional nonlinearity.
As a result, a higher polynomial degree may be required to achieve the same approximation accuracy.
Nevertheless, it does not incur significant additional cost.
Suppose $\cos(x)$ can be approximated by a  polynomial of degree $2K$, i.e., $\cos(x)\approx \sum_{l=1}^{K}c_l x^{2l}$ with a specific coefficient set $\{c_l\}$.
Now, by replacing $x$ with $g(x)=x\sqrt{1-x^{2}}$, the approximation of $\cos(g(x))$ is given by
\begin{equation}
    \cos(g(x)) \approx \sum_{l=1}^{K}c_l (x\sqrt{1-x^{2}})^{2l}=\sum_{l=1}^{K}c_l (x^2(1-x^2))^{l} = \sum_{l=1}^{K}c_l' x^{4l}
\end{equation}
with the set of rearranged coefficients $\{c_{l}'\}$.
Therefore, the degree required to realize $\cos(sg(x))$ is still $4K$, indicating the complexity is the same as $\cos(x)$ up to the multiplicative factor 2.

\medskip
\medskip

Finally, we show that the equivalence between Grover iterations and QSP can be utilized for fixed-point quantum search; namely,  the solution state $\ket{\psi^{*}}$ can be approximately realized using this QSP formulation without knowing $E_{0}$.
Here, given the Grover iteration $\prod_{k=1}^{\mathcal{N}}G_{k}(\alpha_k,\beta_k)$, we aim to realize $W\ket{\psi_{0}}=\ket{\psi_{0}}$, where the matrix $W$ is defined as
\begin{equation} \label{app_eq:w_matrix}
    W\equiv D(\beta_{\mathcal{N}})R(\sqrt{E_{0}})\prod_{k=1}^{\mathcal{N}-1} G_{k}(\alpha_k,\beta_k),
\end{equation}
by properly choosing $\{(\alpha_k,\beta_k)\}_{k=1}^{\mathcal{N}}/\{\alpha_{\mathcal{N}}\}$ via QSP.
Namely, if we can realize $W\ket{\psi_{0}}=\ket{\psi_{0}}$, we obtain
\begin{equation}
\begin{split}
    (-1)\prod_{k=1}^{\mathcal{N}}G_{k}(\alpha_k,\beta_k)\ket{\psi_{0}}&=D(0)R(\sqrt{E_{0}})W \ket{\psi_{0}} \\
    &= \begin{pmatrix} \sqrt{E_{0}} & \sqrt{E_{0}(1-E_{0})} \\  \sqrt{E_{0}(1-E_{0})} & -\sqrt{E_{0}}\end{pmatrix} \begin{pmatrix}1 & 0  \\  0 & \pm1\end{pmatrix} \begin{pmatrix}1   \\  0 \end{pmatrix}\\& =  \begin{pmatrix} \sqrt{E_{0}}  \\  \sqrt{E_{0}(1-E_{0})} \end{pmatrix} = \ket{\psi^{*}},
\end{split}
\end{equation}
by setting $\alpha_{\mathcal{N}}=0$.
That is, since a simple application of $R(\sqrt{E_{0}})$ to the initial state prepares the solution state $\ket{\psi^{*}}=(\sqrt{E_{0}},\sqrt{1-E_{0}})^{T}$, we aim to realize $\braket{\psi_{0}|W|\psi_{0}}=1$

An example of the target function to realize Eq.~\eqref{app_eq:w_matrix} via QSP would be the sign function.
Note that we require $p_{QSP}(\sqrt{E_0}) = \langle 0| W | 0 \rangle = 1$, but we do not care about the value of $ p_{QSP}(x) $ for negative $x$ values as $\sqrt{E_0}$ is always positive.
From this perspective, a constant function can also satisfy the condition.
However, $W$ is composed of an odd ($2\mathcal{N}-1$) number of applications of the signal operator $R(\sqrt{E_0})$, meaning the function is necessarily an odd function.
Thus, implementing the sign function by obtaining the set of angles $\{(\alpha_k,\beta_k)\}_{k=1}^{\mathcal{N}}/\{\alpha_{\mathcal{N}}\}$ can realize the fixed-point algorithm, since the QSP formulation does not require information on $\sqrt{E_{0}}$.
Note that, the solution state can be obtained by applying $R(\sqrt{E_{0}})$ to the initial state $\ket{\psi_{0}}$, i.e., $\ket{\psi^{*}}=R(\sqrt{E_{0}})\ket{\psi_{0}}$ by definition.
However, the Grover iteration inherently applies an even number of $R(\sqrt{E_{0}})$ and thus cannot implement a single $R(\sqrt{E_{0}})$.
Therefore, we construct the matrix Eq.~\eqref{app_eq:w_matrix} to realize the solution state.

We begin by analyzing how the approximation error in the target function, introduced through the QSP construction, propagates to the final fidelity, defined as $F=|\braket{\psi^{*}|\tilde{\psi}}|^2$, where $\ket{\tilde{\psi}}$ is the final state obtained from the QSP formulation of ITE.
Assume that the QSP approximation incurs an error of $\delta^2/2$ and hence the the matrix $W$ in Eq.~\eqref{app_eq:w_matrix} can be changed such that
\begin{equation}
    \tilde{W} = \begin{pmatrix}1-\delta^2/2 & *  \\  * & *\end{pmatrix}.
\end{equation}
Then, the fidelity is given by
\begin{equation} \label{eq:error_qsp_sign}
\begin{split}
    |\braket{\psi^{*}|(-1)D(0)R(\sqrt{E_{0}})\tilde{W}|\psi_{0}}|^{2} &=\left| \begin{pmatrix} \sqrt{E_{0}} , \sqrt{1-E_{0}} \\ \end{pmatrix} \begin{pmatrix}\sqrt{E_{0}} & \sqrt{1-E_{0}}  \\  \sqrt{1-E_{0}} & -\sqrt{E_{0}}\end{pmatrix}\begin{pmatrix}1-\delta^2/2 & *  \\  * & *\end{pmatrix}\begin{pmatrix} 1 \\ 0 \\ \end{pmatrix}\right|^2 \\
     &=\left| \begin{pmatrix} 1 , 0 \\ \end{pmatrix} \begin{pmatrix}1-\delta^2/2 & *  \\  * & *\end{pmatrix}\begin{pmatrix} 1 \\ 0 \\ \end{pmatrix}\right|^2 \\
    &= (1-\delta^2/2)^2 \ge 1-\delta^2,
\end{split}
\end{equation}
with $\alpha_{\mathcal{N}}=0$.
Thus, if the approximation to the sign function is sufficiently accurate, the final state remains close to the desired solution state.

The remaining question concerns the polynomial degree required to attain the desired precision.
Actually, there exists an effective approximation technique for the function:

\begin{lemmaA}[Approximation of the sign function $\mathrm{sgn}(x)$~\cite{lin2020near,martyn2021grand}]
Consider $\eta>0$, $x\in\mathbb{R}$ and $\Delta\in(0,1/2)$.
Given a degree $K=\mathcal{O}(\log(1/\Delta)/\eta)$, there exists an odd polynomial $p(x)$, such that
\begin{itemize}
    \item for all $x\in[-2,2]$: $|p(x)|\le1$ and
    \item for all $x\in[-2,2]\backslash(-\eta,\eta)$: $|p(x)-\mathrm{sgn}(x)|\le\Delta$
\end{itemize}
where 
\begin{equation}
    \mathrm{sgn}(x)=
  \begin{cases}
    1 & \text{if $x>0$,} \\
    -1                 & \text{if $x<0$,} \\
    0       & \text{if $x=0$.}
  \end{cases}
\end{equation}
\end{lemmaA}
Consequently, combining the error analysis in Eq.~\eqref{eq:error_qsp_sign} with the parameters $\eta=\sqrt{E_{0}}$ and $\Delta=\delta^2/2$ reveals that a degree of polynomial $K\in\mathcal{O}(\log(2/\delta^2)/\sqrt{E_{0}})$ ensures that the final fidelity exceeds $1-\delta^2$.
This scaling matches the complexity of the fixed-point algorithm.

We note that the sign function does not satisfy the conditions for the achievable functions via the ($R(x)$, $S_{Z}$, $\ket{0}$)-QSP setting, as shown in App.~\ref{app:qsp_overview}.
However, Ref.~\cite{martyn2021grand} demonstrates that, allowing for a small imaginary components enables the effective approximation of the sign function in practical settings.
Therefore, while additional errors may occur because of the imaginary components, the implementation can be feasible in practice.
\end{proof}

\subsection{Proof of Theorem~\ref{prop:ite_oaa}} \label{app:thm_oaa}

We restate the result from the main text on oblivious amplitude amplification (OAA)~\cite{berry2014exponential} in a self-contained form.

\begin{theoremA}[ITE formulation for OAA] \label{app_prop:ite_oaa}
    Consider the Hamiltonian $\hat{H}_{f}=\ket{0}\bra{0} \otimes V \ket{\phi} \bra{\phi} V^{\dagger}$ and the initial state $\ket{\psi}=U\ket{0}\ket{\phi}$.
    Then, our ITE formulation realizes the OAA procedure, i.e., 
    \begin{equation}
        e^{s[\hat{H}_{f}, \psi_{0}]} \ket{\psi_{0}} \approx \prod_{k=1}^{\mathcal{N}}  U \tilde{D}(\alpha_{k}) U^{\dagger} \tilde{D}(\beta_{k}) \ket{\psi_0}
    \end{equation}
    for proper choice of $\{\alpha_k, \beta_{k}\}$ corresponding to $s$. Here, $ \tilde{D}(\theta)=e^{i\theta \ket{0}\bra{0} }\otimes I$.
\end{theoremA}

\begin{proof}
The goal of OAA~\cite{berry2014exponential} is to prepare the state $V\ket{\phi}$, where $V$ is a target unitary and $\ket{\phi}$ is an arbitrary initial state. 
In this setting, $V$ is not directly implementable, but access to $V$ is provided through a unitary $U$ acting as
\begin{equation}
    U\ket{0}\ket{\phi} = \sqrt{p}\ket{0}V\ket{\phi} + \sqrt{1-p} \ket{1}\ket{\phi'}
    \label{app_eq OAA def}
\end{equation}
with an unknown parameter $p\in (0,1)$.
The OAA procedure applies $\mathcal{N}$ repeated iterations of the operator $W(\alpha,\beta)= -U\tilde{D}(\alpha)U^{\dagger}\tilde{D}(\beta)$, where 
\begin{equation}
    \tilde{D}(\theta)=   e^{i\theta \ket{0}\bra{0}}  \otimes I\ .
\end{equation}
Then, by properly choosing the phases, the state $\prod_{k=1}^{\mathcal{N}} W_{k}(\alpha_{k},\beta_{k}) U \ket{0}\ket{\psi}$ can yield $\ket{0}V\ket{\psi}$, removing the need for post-selection.
The resulting query complexity follows the same scaling as Grover’s algorithm~\cite{grover1996fast}.

We show that the quantum circuit used in OAA can be reproduced through our ITE formulation.
As stated in Prop.~\ref{app_prop:ite_oaa}, we define the projector Hamiltonian and the initial state as
\begin{equation}
    \hat{H}_{f} = \ket{0}\bra{0} \otimes V \ket{\phi} \bra{\phi} V^{\dagger}, \quad \ket{\psi_{0}}  =  U\ket{0}\ket{\phi}.
\end{equation}
The DBQA framework in this setting results in the Grover iteration
\begin{equation}
\begin{split} \label{app_eq:AA}
    G(\alpha,\beta) &=  e^{i\alpha \psi_{0}} e^{i\beta \hat{H}_{f}} \\
    &= U e^{i\alpha (\ket{0}\bra{0}\otimes \phi_{0})} U^{\dagger} (I\otimes V) e^{i\beta(\ket{0}\bra{0}\otimes \phi_{0})} (I\otimes V^{\dagger}) 
\end{split}  
\end{equation}
with $\phi_{0}=\ket{\phi}\bra{\phi}$.
This construction is equivalent to amplitude amplification (AA) and can be implemented assuming the target unitary $V$ is accessible as an \textit{oracle}.
However, it remains different from OAA, since Eq.~\eqref{app_eq:AA} requires the full reflection $e^{i\theta (\ket{0}\bra{0}\otimes \phi_{0})}$ and direct access to $V$.
To bridge this gap, we employ the following identity.

\begin{lemmaA} \label{lem: partial-full reflection equality}
    Consider a quantum state of the form $\ket{\Psi}=z_{0} \ket{0}\ket{\phi} + z_{1} \ket{1}\ket{\phi'} $ with $z_{0},z_{1} \in \mathbb{C}$ and arbitrary states $\ket{\phi},\ket{\phi'}$.
    Then, the reflection around the first component $\ket{0}\ket{\phi} $ in the state $\ket{\Psi}$ can be implemented by reflecting around the partial system;
\begin{equation}
    e^{i\theta (\ket{0}\bra{0}\otimes \phi)} \ket{\Phi} = (e^{i\theta \ket{0}\bra{0}}\otimes I) \ket{\Phi}.
\end{equation}
\end{lemmaA}

\begin{proof}
    The full reflection operator can be written as
    \begin{equation} \label{eq:full_refleciton_em}
    \begin{split}
        e^{i\theta (\ket{0}\bra{0}\otimes \phi)} = e^{i\theta (\ket{0}\bra{0}\otimes (I-\phi^{\perp})}= (e^{i\theta \ket{0}\bra{0} }\otimes I ) e^{i\theta (\ket{0}\bra{0}\otimes \phi^{\perp})},
    \end{split}
    \end{equation}
    where we define $\phi^{\perp} = I- \ket{\phi}\bra{\phi}$ satisfying $\braket{\phi|\phi^{\perp}|\phi}=0$.
    Applying Eq.~\eqref{eq:full_refleciton_em} to the state $\ket{\Psi}$ gives
    \begin{equation}
    \begin{split}
        e^{i\theta (\ket{0}\bra{0}\otimes \phi)} \ket{\Psi}&=(e^{i\theta \ket{0}\bra{0} }\otimes I ) e^{i\theta (\ket{0}\bra{0}\otimes \phi^{\perp})} \ket{\Psi} \\
        &= (e^{i\theta \ket{0}\bra{0} }\otimes I )\ket{\Psi}.
    \end{split}
    \end{equation}
    In the second equality, we use the expansion 
    \begin{equation}
        e^{i\theta (\ket{0}\bra{0}\otimes \phi^{\perp})} = I + (e^{i\theta}-1) (\ket{0}\bra{0}\otimes \phi^{\perp})
    \end{equation}
    and observe that $ (\ket{0}\bra{0}\otimes \phi^{\perp}) \ket{\Psi}=0$.
    This completes the proof.
\end{proof}

Now, using Lemma~\ref{lem: partial-full reflection equality}, the quantum circuit for AA in Eq.~\eqref{app_eq:AA} can be decomposed further.
Concretely, we have
\begin{equation}
\begin{split}
    &(I\otimes V) e^{i\beta (\ket{0}\bra{0}\otimes \phi_{0})} (I\otimes V^{\dagger}) \ket{\psi_{0}} \\
    &=(I\otimes V) e^{i\beta (\ket{0}\bra{0}\otimes \phi_{0})} \left(\sqrt{p}\ket{0}\ket{\phi} + \sqrt{1-p} \ket{1}V^{\dagger}\ket{\phi'}\right) \\
    &=  (e^{i\beta \ket{0}\bra{0}} \otimes I) \ket{\psi_{0}}.
\end{split}
\end{equation}
Eq.~\eqref{app_eq OAA def} is inserted and $I\otimes V^{\dagger}$ is applied to the state in the first equality, whereas, in the second equality, Lemma~\ref{lem: partial-full reflection equality} is used to simplify the action of the reflection and $I\otimes V$ is applied.
Moreover, Lemmas 3.6 and 3.7 of Ref.~\cite{berry2014exponential} give $U^{\dagger}(\ket{0}V\ket{\phi})=\sqrt{p}\ket{0}\ket{\phi}+\sqrt{1-p}\ket{1}\ket{\chi}$ and $U^{\dagger}(\ket{1}\ket{\phi'})=\sqrt{1-p}\ket{0}\ket{\phi} - \sqrt{p} \ket{1}\ket{\chi}$.
These relations allow another application of Lemma~\ref{lem: partial-full reflection equality}, which reveals
\begin{equation}
\begin{split}
    &U e^{i\alpha (\ket{0}\bra{0}\otimes \phi_{0})} U^{\dagger} (e^{i\beta \ket{0}\bra{0}} \otimes I) \ket{\psi_{0}}  \\
    &=  U  (e^{i\alpha \ket{0}\bra{0}} \otimes I)  U^{\dagger}  (e^{i\beta \ket{0}\bra{0}} \otimes I) \ket{\psi_{0}} \\
    &= U \tilde{D}(\alpha) U^{\dagger} \tilde{D}(\beta).
\end{split} 
\end{equation}
Repeating this procedure shows that the ITE formulation reproduces the OAA quantum circuit.
\end{proof}

\section{Numerical Simulations} \label{app:numerics}

In this section, we perform numerical simulations to verify the QSP formulation of ITE for unstructured search.
Specifically, we examine two aspects: (1) the accuracy of the approximation for realizing ITE and (2) the performance of the fixed-point algorithm using the approximated sign function, in comparison with existing methods.

First, we numerically demonstrate that the ITE state in Eq.~\eqref{app_eq:dbr_vector} can be realized by properly choosing $\bm{\alpha}=\{\alpha_{k}\}_{k=1}^{2\mathcal{N}}$ in
\begin{equation}
    U(\bm{\alpha},x)  = \prod_{k=1}^{2\mathcal{N}} D(\alpha_k) R(x)  = \prod_{k=1}^{2\mathcal{N}} \begin{pmatrix}e^{i\alpha_k} & 0  \\  0 & 1\end{pmatrix} \begin{pmatrix}x & \sqrt{1-x^2}  \\  \sqrt{1-x^2} & -x\end{pmatrix} ,
\end{equation}
via the QSP implementation of $p(x)=\cos(sx\sqrt{1-x^2})$.
In the numerical simulation, we heuristically optimize $\bm{\alpha}$ such that $\Re(\braket{0|U(\bm{\alpha},x)|0})\approx p(x),\, x\in[0,1] $.
Here, $\ket{0}=(1,0)^{T}$ corresponds to the initial state $\ket{\psi_{0}}$ and $\ket{1}=(0,1)^{T}$ corresponds to $\ket{\psi_{0}^{\perp}}$ for the ITE formulation.
We recall that the input domain is restricted to $[0,1]$, since $E_{0}$ only takes a non-negative value for unstructured search.
To perform the optimization, we introduce the following cost function;

\begin{equation}
    \mathcal{L}(\bm{\alpha}) = \frac{1}{n_d}\sum_{i=1}^{n_d} (p(x_i)-\Re{(\braket{0|U(\bm{\alpha},x_i)|0})})^2 + \lambda_{1} \frac{1}{n_d}\sum_{i=1}^{n_d} \Im{(\braket{0|U(\bm{\alpha},x_i)|0})}^2 + \lambda_{2} \frac{1}{n_d}\sum_{i=1}^{n_d} \arg\left(\frac{\braket{0|U(\bm{\alpha},x_i)|0}}{\braket{1|U(\bm{\alpha},x_i)|0}}\right)^2 ,
\end{equation}
where the first term represents the mean square errors between the approximated and the target functions, while the remaining two terms are introduced as the penalty with the coefficients $\lambda_1$ and $\lambda_2$.
Concretely, the second term penalizes the imaginary part of $\braket{0|U(\bm{\alpha})|0}$ to make it zero.
The third term suppresses the relative phase between $\braket{0|U(\bm{\alpha},x)|0}$ and $\braket{1|U(\bm{\alpha},x)|0}$, which is necessary because the resulting state after applying $U(\bm{\alpha})$ is
\begin{equation}
    U(\bm{\alpha},x)\ket{0}=\begin{pmatrix}\braket{0|U(\bm{\alpha},x)|0}   \\  \braket{1|U(\bm{\alpha},x)|0}\end{pmatrix},
\end{equation}
and we aim to approximate $(\cos(sx\sqrt{1-x^2}),\sin(sx\sqrt{1-x^2}))^{T}$ for an arbitrary $s\ge0 $ up to a global phase with this; that is, the relative phase is distinguishable and hence incurs the additional errors when implementing ITE with the phase obtained through the QSP construction.
In the following simulation, we employ the ``scipy.optimize.minimize" function from SciPy~\cite{2020SciPy-NMeth}, using the Sequential Least Squares Programming (SLSQP) method. 
The penalty coefficients are set $(\lambda_1,\lambda_2)=(0.01,0.1)$, and we use $n_d=50$ samples, uniformly drawn from the interval $[0,1]$.


\begin{figure}[t]
\centering
\begin{tikzpicture}
\definecolor{lightgray}{HTML}{F4F4F4}
\definecolor{littlelightgray}{HTML}{ecececff}
\definecolor{pale}{HTML}{7ca3d4ff}
\definecolor{lightred}{HTML}{d8a2a2}
\node[anchor=center] (russell) at (-7.6,-2.3)
{\centering\includegraphics[width=0.30\textwidth]{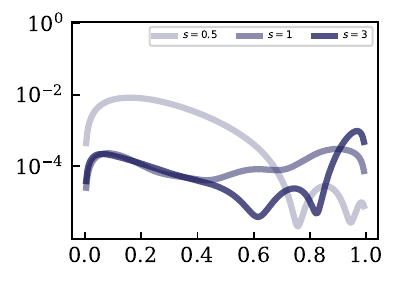}};
\node[text width=4cm] at (-6.4,-4.3){Initial Overlap $E_{0}$};
\node[rotate=90,text width=4cm] at (-10.3,-0.8) 
    {Infidelity $\mathcal{I}$};
\node[text width=1cm] at (-10,0){(a)};
\node[anchor=center] (russell) at (-1.9,-2.3)
{\centering\includegraphics[width=0.30\textwidth]{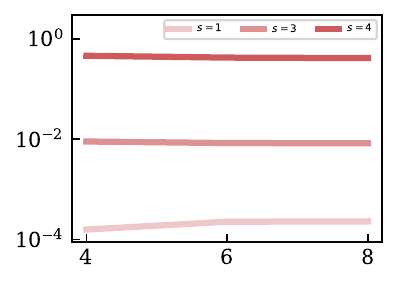}};
\node[text width=4cm] at (-0.8,-4.3){Number of Qubits};
\node[rotate=90,text width=4cm] at (-4.6,-0.8) 
    {Infidelity $\mathcal{I}$};
\node[text width=1cm] at (-4.3,0){(b)};
\node[anchor=center] (russell) at (3.9,-2.3)
{\centering\includegraphics[width=0.30\textwidth]{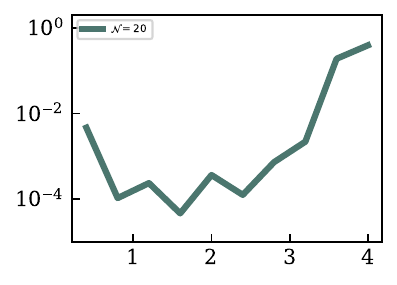}};
\node[text width=4cm] at (5.3,-4.3){Time Duration $s$};
\node[rotate=90,text width=4cm] at (1.2,-0.8) 
    {Infidelity $\mathcal{I}$};
\node[text width=1cm] at (1.5,0){(c)};
\end{tikzpicture}
\caption{\textbf{Numerical simulations of QSP formulation for ITE.}
We numerically verify the QSP formulation of ITE for unstructured search. (a) Using the set of angles derived from the QSP formulation, we compute the infidelity between the ITE state defined in Eq.~\eqref{eq:dbr} and the state obtained after applying the Grover iteration with $\mathcal{N} = 16$ to the initial state for $N=2^{8}$.
Across the entire range of initial overlaps $E_{0}$, the infidelity remains below $10^{-2}$ for $s = 0.5$, $1$, and $3$.
(b) We check the dependence on the system size $N=2^{n}$ with $n=4,6,8$ for $s=1,3,4$. As shown in Theorem~\ref{thm:qsp_formula_ITE}, the performance does not depend on the system size.
(c) We further check the dependence of performance on the time duration. Using $\mathcal{N}=20$ for the Hamiltonian $\hat{H}_{f}$ of size $N=2^6$, we confirm that the large value of $s$ incurs more errors, as indicated by Lemma~\ref{lem:jacobi-anger}.
}
\label{fig:qsp_formulation_ITE}
\end{figure}


Fig.~\ref{fig:qsp_formulation_ITE}~(a) shows the infidelity between the ITE state in Eq.~\eqref{eq:dbr} and the state $U(\bm{\alpha},x)\ket{\psi_{0}}$ with a set of angles obtained via the optimization for $\mathcal{N}=16$, $s=0.5,1,3$ and $N=2^{8}$.
The infidelity is defined as $\mathcal{I}=1-|\braket{\psi_{s}|U(\bm{\alpha},x)|\psi_{0}}|^2$, where $\ket{\psi_{s}}$ is the ITE state in Eq.~\eqref{eq:dbr}.
We observe that the Grover iteration using the optimized QSP angles can accurately reproduce the ITE state. 
The difference in accuracy with respect to the initial overlap $E_{0}$ can be attributed to numerical optimization error.
More specifically, while the target function is realized by $\braket{0|U(\bm{\alpha},x)|0}$ effectively, the numerical optimization does not fully eliminate the relative phase error, i.e., $\arg\left(\frac{\braket{0|U(\bm{\alpha},x_i)|0}}{\braket{1|U(\bm{\alpha},x_i)|0}}\right)\neq0$.
Consequently, we have an input-dependent difference caused by a nonzero, input-dependent relative phase.

We then check the system-size dependence by using the Hamiltonian $\hat{H_{f}}$ of the size $N=2^n$ with $n=4,6,8$ for $s=1,3,4$.
The number of Grover iterations is set as $\mathcal{N}=8$ and we compute the infidelity averaged over the choice of the number of targets $M=1,\ldots,2^{n}-1$.
As shown in Fig.~\ref{fig:qsp_formulation_ITE}~(b), we confirm that the approximation error remains independent of the system size. 
This behavior is consistent with Theorem~\ref{thm:qsp_formula_ITE}, which establishes the equivalence between the ITE state and its two-dimensional representation.
We also note that a large value of the time duration $s$ incurs a big error, which motivates the following numerical experiments.

We also examine the dependence of performance on the time duration $s$.
Here, we consider $\hat{H_{f}}$ of the size $N=2^6=64$ and the number of queries $\mathcal{N}=20$.
According to Lemma~\ref{lem:jacobi-anger}, the polynomial degree required to approximate the target trigonometric functions grows linearly with the time duration.
In agreement with this theoretical expectation, Fig.~\ref{fig:qsp_formulation_ITE}~(c) shows that the Grover iteration with a fixed number of polynomial degree fail to realize the ITE state for larger values of $s$.
Note that the case with $s < 1$ exhibits higher infidelity compared to the cases with $1 \leq s \leq 3$.
However, since the infidelity remains below $0.01$ in all cases, the observed deviation can be attributed to optimization error rather than fundamental limitations.
Overall, the numerical results verify that the ITE state in Eq.~\eqref{eq:dbr} can be effectively reconstructed using the QSP formulation.

\medskip
\medskip


\begin{figure}[t]
\centering
\begin{tikzpicture}
\definecolor{lightgray}{HTML}{F4F4F4}
\definecolor{littlelightgray}{HTML}{ecececff}
\definecolor{pale}{HTML}{7ca3d4ff}
\definecolor{lightred}{HTML}{d8a2a2}
\node[anchor=center] (russell) at (-12,-2.3)
{\centering\includegraphics[width=0.45\textwidth]{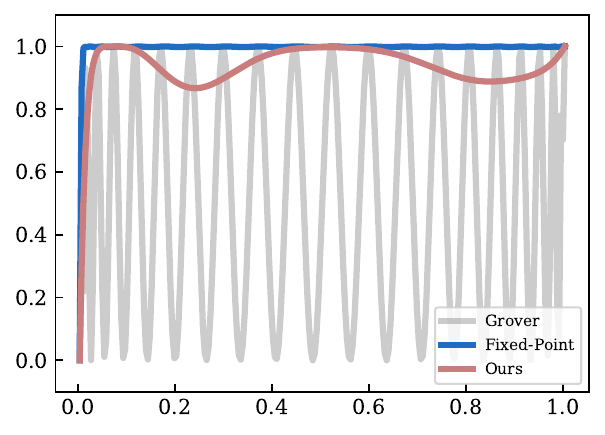}};
\node[text width=4cm] at (-10.9,-5.3){Initial Overlap $E_{0}$};
\node[rotate=90,text width=4cm] at (-16.2,-0.8) 
    {Final Overlap};
\node[text width=1cm] at (-16.0,0.5){(a)};
\node[anchor=center] (russell) at (-3.2,-2.3)
{\centering\includegraphics[width=0.45\textwidth]{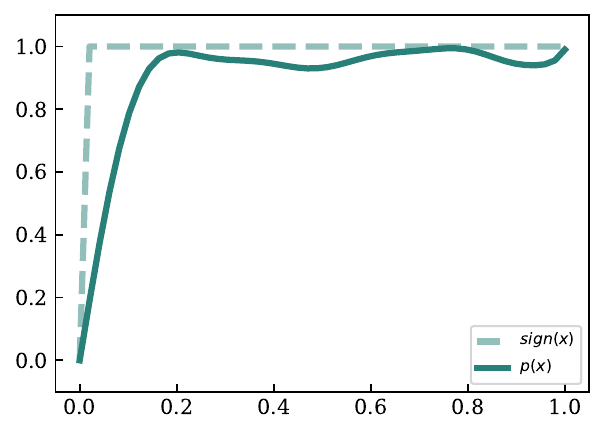}};
\node[text width=4cm] at (-1.7,-5.3){Input Value $x$};
\node[rotate=90,text width=4cm] at (-7.4,-0.6) 
    {Output};
\node[text width=1cm] at (-7.2,0.5){(b)};\end{tikzpicture}
\caption{\textbf{Comparison of our proposal to the original Grover algorithm and fixed-point algorithm.} We compare the performance of our proposed method with the original Grover and fixed-point algorithms.
(a) We demonstrate the overlap after employing these algorithms against the initial overlap $E_{0}$.
In this comparison, we fix $\mathcal{N} = 20$ and consider a Hamiltonian of dimension $2^8$, where the number of target items $M$ is varied according to the initial overlap $E_{0} = M/N$.
(b) The approximated sign function $p(x)$ obtained via heuristic optimization is illustrated.
While the results in (a) exhibit slightly lower performance than the original fixed-point algorithm due to optimization error, the proposed method successfully mitigates overshooting.
}
\label{fig:fixed_point}
\end{figure}


Next, we benchmark the performance of our QSP formulation of ITE for a fixed-point search by comparing it with the original Grover and fixed-point algorithms.
The set of angles for each algorithm is as follows;

\begin{itemize}

\item \textbf{Standard Grover method:} The phase angles for the original work~\cite{grover1996fast} is given by
\begin{align}
    \alpha_k=\beta_k= \pi \ .
\end{align}
    \item \textbf{Original fixed‑point angles:} In the original fixed-point algorithm~\cite{yoder2014fixed}, the phase angles are given by
\begin{equation}
\begin{split}
    \alpha_{k}&=\beta_{\mathcal{N}-k+1}=-\cot^{-1}\left(\tan\left(\frac{2\pi k}{\mathcal{N}}\right)\sqrt{1-\frac{1}{\gamma^2}}\right),
\end{split}
\end{equation}
where $\gamma=T_{1/\mathcal{N}}(1/\delta)$  is the Chebyshev polynomial of the first kind and we set the target infidelity $\delta^2=0.1$.

\item \textbf{Our proposal:} We use the heuristic approach mentioned above to derive phase angles for a polynomial function that approximates the sign function, such that $W\ket{0}=\ket{0}$.
We also set $\alpha_{2\mathcal{N}}=0$ to realize the solution state $\ket{\psi^{*}}$. 
\end{itemize}

Fig.~\ref{fig:fixed_point}~(a) illustrates the fidelity between the solution state and the output state generated by Grover iterations, $F=|\braket{\psi^{*}|\tilde{\psi}}|^2$, where $\ket{\tilde{\psi}}$ is a state generated by Grover iterations with $\mathcal{N}=20$ using the three different sets of phase angles
In this simulation, we set $n=8$, which corresponds to a search space of $2^{8}=256$ items, and consider $M=0,1,\ldots,N$ so that we check the performance depending on $E_{0}=M/N$.
As for the fixed-point algorithm, we set $\delta^{2}=0.1$
The results show that our method achieves fidelity comparable to that of the fixed-point algorithm, while avoiding the overshooting behavior exhibited by the original Grover algorithm. The minor discrepancy in performance can be attributed to the approximation error in the sign function used in the QSP construction.
Fig.~\ref{fig:fixed_point}~(b) shows  the approximated function realized via the QSP procedure, along with the associated approximation error. Despite this small error, our approach remains highly effective and closely matches the performance of the fixed-point method.
These results validate the practical effectiveness of our QSP formulation and provide new insights into the design of the fixed-point algorithm.


\end{document}